\documentclass{llncs}

\usepackage[usenames,dvipsnames]{color}

\usepackage{latexsym}
\usepackage{amsxtra} 
\usepackage{amssymb}
\usepackage{amsmath}
\usepackage{pslatex}
\usepackage{epsfig}
\usepackage{wrapfig}

\usepackage{algorithm}
\usepackage{algorithmicx}
\usepackage[noend]{algpseudocode}

\usepackage[draft]{commenting}

\renewcommand{\vec}[1]{{\bf {#1}}}

\newcommand{\defequiv}{\stackrel{\scriptscriptstyle{def}}{\equiv}}

\newcommand{\sig}{\mathtt{sig}}

\newcommand{\true}{\mathbf{true}}
\newcommand{\false}{\mathbf{false}}
\newcommand{\emp}{\mathbf{emp}}
\newcommand{\nil}{\mathbf{nil}}
\newcommand{\head}{\mathbf{head}}
\newcommand{\tail}{\mathbf{tail}}
\newcommand{\form}{\mathbf{form}}
\newcommand{\port}{\mathbf{port}}

\newcommand{\prm}{\mathbf{par}}
\newcommand{\alloc}{\mathbf{alloc}}
\newcommand{\subtree}[2]{{{#1}}_{|_{{#2}}}}

\newcommand{\DLL}{\mathtt{DLL}}
\newcommand{\TLL}{\mathtt{TLL}}
\newcommand{\TREE}{\mathtt{TREE}}
\newcommand{\ax}{\mathtt{a}}
\newcommand{\bx}{\mathtt{b}}
\newcommand{\cx}{\mathtt{c}}
\newcommand{\dx}{\mathtt{d}}

\newcommand{\len}[1]{{|{#1}|}}
\newcommand{\card}[1]{{|\!|{#1}|\!|}}

\newcommand{\lang}[1]{{\mathcal L}({#1})}

\newcommand{\arrow}[2]{\xrightarrow{{\scriptstyle #1}}_{{\scriptstyle #2}}}

\newcommand{\nat}{{\bf \mathbb{N}}}

\renewcommand{\paragraph}[1]{\noindent\emph{#1}}

\newif\ifLongVersion\LongVersiontrue

\definecolor{darkgreen}{rgb}{0,0.6,0}
\declareauthor{ri}{Radu}{blue}
\declareauthor{ar}{Adam}{red}
\declareauthor{tv}{Tomas}{darkgreen}
\declareauthor{tv2}{Tomas}{magenta}

\begin{document}

\title{Deciding Entailments in Inductive Separation Logic with Tree
Automata\vspace*{-4.5mm}}

\author{Radu Iosif \and Adam Rogalewicz \and Tom\'{a}\v{s}~Vojnar\vspace*{-1.5mm}}
  
\institute{CNRS/Verimag, France and FIT BUT, Czech Republic\\
\texttt{radu.iosif@imag.fr},
$\{$\texttt{rogalew,vojnar}$\}$\texttt{@fit.vutbr.cz}\vspace*{-5mm}} \maketitle

\begin{abstract}Separation Logic (SL) with inductive definitions is a natural
formalism for specifying complex recursive data structures, used in
compositional verification of programs manipulating such structures. The key
ingredient of any automated verification procedure based on SL is the
decidability of the entailment problem. In this work, we reduce the entailment
problem for a non-trivial subset of SL describing trees (and beyond) to the
language inclusion of tree automata (TA). Our reduction provides tight
complexity bounds for the problem and shows that entailment in our fragment is
EXPTIME-complete. For practical purposes, we leverage from recent advances in
automata theory, such as inclusion checking for non-deterministic TA avoiding
explicit determinization. We implemented our method and present promising
preliminary experimental results.\end{abstract}

\vspace*{-8mm}\section{Introduction}\vspace*{-2mm}

Separation Logic (SL) \cite{reynolds02} is a logical framework for describing
recursive mutable data structures. The attractiveness of SL as a specification
formalism comes from the possibility of writing higher-order {\em inductive
definitions} that are natural for describing the most common recursive data
structures, such as singly- or doubly-linked lists (SLLs/DLLs), trees, hash maps
(lists of lists), and more complex variations thereof, such as nested and
overlaid structures (e.g.\ lists with head and tail pointers, skip-lists, trees
with linked leaves, etc.). In addition to being an appealing specification tool,
SL is particularly suited for compositional reasoning about programs. Indeed,
the principle of {\em local reasoning} allows one to verify different elements
(functions, threads) of a program, operating on disjoint parts of the memory,
and to combine the results a-posteriori, into succinct verification conditions.

However, the expressive power of SL comes at the price of
undecidability \cite{Brotherston:2010}. To avoid this problem, most SL
dialects used by various tools (e.g.\ \textsc{Space Invader}
\cite{spaceinvader}, \textsc{Predator} \cite{predator}, or
\textsc{Infer} \cite{infer}) use hard-coded predicates, describing
SLLs and DLLs, for which entailments are, in general, tractable
\cite{cook-haase-ouaknine-parkinson-worell11}. For graph structures of
bounded tree width, a general decidability result was presented in
\cite{Iosif13}.  Entailment in this fragment is EXPTIME-hard as proved
in \cite{fossacs14}.

\enlargethispage{5mm}

In this paper, we present a novel decision procedure for a restriction of the
decidable fragment of \cite{Iosif13} that describes recursive structures in which
all edges are {\em local} with respect to a spanning tree.  Examples of such
structures include SLLs, DLLs, trees and trees with parent pointers, etc. For
structures outside of this class (e.g.\ skip-lists or trees with linked leaves),
our procedure is sound (the answer is positive if the entailment holds), but not
complete (the answer might be negative and the entailment could still hold). In
terms of program verification, such a lack of completeness in the entailment
prover can cause non-termination of the fixpoint, but will not provide unsound
results.

The method described in the paper belongs to the class of {\em
automata-theoretic} decision techniques: we translate an entailment problem
$\varphi \models \psi$ into a language inclusion problem $\lang{A_\varphi}
\subseteq \lang{A_\psi}$ for tree automata (TA)  $A_\varphi$ and $A_\psi$ that
(roughly speaking) encode the sets of models of $\varphi$ and $\psi$,
respectively. Yet, a na\"{\i}ve translation of the inductive definitions of SL
into TA encounters a {\em polymorphic representation} problem: the same
structure can be defined in several different ways, and TA simply mirroring the
definition will not report the entailment. For example, DLLs with selectors
$\mathtt{next}$ and $\mathtt{prev}$ for the next and previous nodes,
respectively, can be described by a~forward unfolding of the inductive
definition $\DLL(head,prev,tail,next) \equiv \exists x.~ head \mapsto (x,prev) *
\DLL(x,head,tail,next)$ as well as by a backward unfolding of the definition
$\DLL_{rev}(head, prev,tail,next) \equiv \exists x.~ tail \mapsto (next,x) *
\DLL_{rev}(head, prev,x,tail)$. Also, one can define a DLL starting with a node
in the middle and unfolding backward to the left of this node and forward to the
right: $\DLL_{mid}(head,prev,tail,next) \equiv \exists x,y,z ~.~ x \mapsto (y,z)
* \DLL(y,x,tail,next) * \DLL_{rev}(head, prev, z, x)$.  The entailments
$\DLL(\ax, \bx, \cx, \dx) \models \DLL_{rev}(\ax, \bx, \cx, \dx)$ and
$\DLL_{mid}(\ax, \bx, \cx, \dx) \models \DLL(\ax, \bx, \cx, \dx)$ hold, but a
na\"{\i}ve structural translation to TA might not detect this fact.  To bridge
this gap, we define a closure operation on TA, called {\em canonical rotation},
which automatically adds all possible representations of a given inductive
definition, encoded as a tree automaton.

Our reduction from SL to TA provides tight complexity bounds showing that
entailment in our fragment is EXPTIME-complete. Moreover, from a practical point
view, we implemented our method using the \textsc{Vata} \cite{libvata} tree
automata library, which leverages from recent advances in non-deterministic
language inclusion for TA \cite{anti-chains}, and obtained quite encouraging
experimental results.


\vspace*{2mm}\paragraph{Related work.} Given the large body of literature on
decidable logics for describing mutable data structures, we need to restrict
this section to the related work that focuses on SL \cite{reynolds02}. The first
(proof-theoretic) decidability result for SL on a restricted fragment defining
only SLLs was reported in \cite{berdine-calcagno-ohearn04}, which describe a
co-NP algorithm.  The full basic SL without recursive definitions, but with the
magic wand operator was found to be undecidable when interpreted {\em in any
memory model} \cite{Brotherston:2010}. A PTIME entailment procedure for SL with
list predicates is given in \cite{cook-haase-ouaknine-parkinson-worell11}. Their
method was extended to reason about nested and overlaid lists in \cite{enea13}.
More recently, entailments in an important SL fragment with hardcoded SLL/DLL
predicates were reduced to Satisfiability Modulo Theories (SMT) problems,
leveraging from recent advances in SMT technology \cite{Wies13}.

\enlargethispage{4mm}

Closer to our work on SL with user-provided {\em inductive
  definitions} is the fragment used in the tool \textsc{Sleek}, which
implements a semi-algorithmic entailment check, based on unfoldings
and unifications \cite{sleek}. Our previous work \cite{Iosif13} gives
a general decidability result for SL with inductive definitions
interpreted over graph-like structures, under several necessary
restrictions. The work \cite{fossacs14} provides a rather complete
picture of complexity for the entailment in various SL fragments with
inductive definitions, including EXPTIME-hardness of the decidable
fragment of \cite{Iosif13}, but provides no upper bound. The
EXPTIME-completness result in this paper provides an upper bound for
the fragment of {\em local} definitions, and strengthens the
EXPTIME-hard lower bound as well.

\vspace*{-2mm}\section{Definitions}\vspace*{-1.5mm}

The set of natural numbers is denoted by $\nat$. For a finite set $S$, we denote
by $\card{S}$ its cardinality. If $\vec{x} = \langle x_1, \ldots, x_n \rangle$
and $\vec{y} = \langle y_1, \ldots, y_m \rangle$ are tuples, $\vec{x} \cdot
\vec{y} = \langle x_1, \ldots, x_n, y_1, \ldots, y_m \rangle$ denotes their
concatenation, and $(\vec{x})_i = x_i$ denotes the $i$-th element of $\vec{x}$.
If $x$ is an element, we denote by $\vec{x}_{\neg x}$ the result of projecting
out all occurrences of $x$ from $\vec{x}$, and if $S$ is a set, we denote by
$\vec{x}_{\neg S}$ the result of projecting out all occurrences of some element
$s \in S$ from $\vec{x}$. The intersection $\vec{x} \cap S$ denotes the tuple in
which only elements of $S$ are maintained. The union (intersection) $\vec{x}
\cup \vec{y}$ ($\vec{x} \cap \vec{y}$) denotes the set of elements from either
(both) $\vec{x}$ and $\vec{y}$. For two tuples of variables $\vec{x} = \langle
x_1,\ldots,x_k \rangle$ and $\vec{y} = \langle y_1,\ldots,y_k \rangle$, of equal
length, we write $\vec{x}=\vec{y}$ for $\bigwedge_{i=1}^k x_i = y_i$. Moreover,
for a single variable $z$, we write $\vec{x}=z$ for $\bigwedge_{i=1}^k x_i = z$.

For a partial function $f : A \rightharpoonup B$, and $\bot\notin B$,
we denote $f(x)=\bot$ the fact that $f$ is undefined at some point $x
\in A$. The domain of $f$ is denoted $dom(f) = \{x \in A \mid f(x)
\neq \bot\}$, and the image of $f$ is denoted as $img(f) = \{ y\in B
\mid \exists x\in A ~.~ f(x)=y\}$. By $f : A \rightharpoonup_{fin} B$,
we denote any partial function whose domain is finite. Given two
partial functions $f,g$ defined on disjoint domains, we denote by $f
\oplus g$ their union.

\vspace*{-2mm}\subsection{Stores, Heaps, and States}\vspace*{-1mm}

We consider $Var = \{x,y,z,\ldots\}$ to be a countably infinite set of
{\em variables} and $\nil \in Var$ be a designated variable.
\ifLongVersion
We assume a total lexicographical ordering on the set
of variables $Var$, and for any set $S \subseteq Var$, we denote by
$minlex(S)$ the unique minimal element with respect to this ordering.
\fi
Let $Loc$ be a countably infinite set of locations, $null \in Loc$ be
a designated location, and $Sel = \{1,\ldots,\mathcal{S}\}$, for some
given $\mathcal{S} > 0$, be a finite set of natural numbers, called
{\em selectors} in the following.

\begin{definition}\label{state} A \emph{state} is a pair $\langle s, h \rangle$
where $s : Var \rightharpoonup Loc$ is a~ partial function mapping pointer
variables into locations such that $s(\nil)=null$, and $h : Loc
\rightharpoonup_{fin} Sel \rightharpoonup_{fin} Loc$ is a finite partial
function such that (i) $null \not\in dom(h)$ and (ii) for all $\ell\in dom(h)$
there exists $k \in Sel$ such that $(h(\ell))(k) \neq \bot$.\end{definition}

\enlargethispage{2mm}

Given a state $S=\langle s, h \rangle$, $s$ is called the \emph{store}
and $h$ the \emph{heap}. For any $k \in Sel$, we write $\ell
\arrow{k}{S} \ell'$ for $(h(\ell))(\sigma) = \ell'$. We call a triple
$\ell \arrow{k}{S} \ell'$ an {\em edge} of $S$. Sometimes we omit the
subscript when it is obvious from the context. Let $Img(h) =
\bigcup_{\ell \in Loc} img(h(\ell))$ be the set of locations which are
destinations of some edge in $h$. A location $\ell \in Loc$ is said to
be \emph{allocated} in $\langle s, h \rangle$ if $\ell \in dom(h)$
(i.e.\ it is the source of an edge), and \emph{dangling} in $\langle
s, h \rangle$ if $\ell \in [img(s) \cup Img(h)] \setminus dom(h)$,
i.e.\ it is referenced by a store variable, or reachable from an
allocated location in the heap, but it is not allocated in the heap
itself. The set $loc(S) = img(s) \cup dom(h) \cup Img(h)$ is the set
of all locations either allocated or referenced in the state $S$.

For any two states $S_1 = \langle s_1,h_1 \rangle$ and $S_2 = \langle
s_2,h_2 \rangle$, such that (i) $s_1$ and $s_2$ agree on the
evaluation of common variables ($\forall x \in dom(s_1) ~\cap~
dom(s_2) ~.~ s_1(x) = s_2(x)$), and (ii)~$h_1$ and $h_2$ have disjoint
domains ($dom(h_1) ~\cap~ dom(h_2) = \emptyset$), we denote by $S_1
\uplus S_2 = \langle s_1 \cup s_2, h_1 \oplus h_2 \rangle$ the {\em
  disjoint union} of $S_1$ and $S_2$. The disjoint union is undefined
if one of the above conditions does not hold.

\ifLongVersion
For a state $S = \langle s, h \rangle$ and a location $\ell \in
dom(h)$, the {\em neighbourhood of $\ell$ in $S$} is a state denoted
as $S_{\langle\ell\rangle} = \langle s_\ell, h_\ell \rangle$, where:
\begin{itemize}
\item $h_\ell = \{\langle \ell, \lambda k ~.~ \mbox{if}~ \ell
  \arrow{k}{S} \ell' ~\mbox{then}~ \ell' ~\mbox{else}~ \bot \rangle\}$
\item $s_\ell(x) = ~\mbox{if}~ s(x) \in dom(h_\ell) \cup img(h_\ell)
  ~\mbox{then}~ s(x) ~\mbox{else}~ \bot$
\end{itemize}
Intuitively, the neighbourhood of an allocated location $\ell$ is the
state in which only $\ell$ is allocated and all other locations
$\ell'$ for which there is an edge $\ell \arrow{k}{S} \ell'$ are
dangling.
\fi

\vspace*{-2mm}\subsection{Trees}\vspace*{-1mm}

Let $\Sigma$ be a countable alphabet, and $\nat^*$ be the set of
sequences of natural numbers. Let $\epsilon \in \nat^*$ denote the
empty sequence, and $p.q$ denote the concatenation of two sequences
$p,q\in \nat^*$. We say that $p$ is a {\em prefix} of $q$ if $q =
p.q'$, for some $q'\in\nat^*$. 

\ifLongVersion
The total order on $\nat$ extends to a total lexicographical order on
$\nat^*$, and for any set $S \subseteq \nat^*$, we denote by
$minlex(S)$\footnote{The distinction with the lexicographical order on
  $Var$ is clear from the type of elements in $S$.} the
lexicographically minimal element of $S$.
\fi

A \emph{tree} $t$ over $\Sigma$ is a~ finite partial function $t :
\nat^* \rightharpoonup_{fin} \Sigma$ such that $dom(t)$ is a finite
prefix-closed subset of $\nat^*$ and, for each $p \in dom(t)$ and $i
\in \nat$, we have $t(p.i)\neq\bot$ only if $t(p.j) \neq \bot$, for
all $0 \leq j < i$. The sequences $p \in dom(t)$ are called {\em
  positions} in the following. Given two positions $p, q \in dom(t)$,
we say that $q$ is the $i$-th successor (child) of $p$ if $q = p.i$,
for $i\in\nat$.
%
%

We denote by $\mathcal{D}(t) = \{-1,0,\ldots,N\}$ the {\em direction
  alphabet} of $t$, where $N = \max\{i \in \nat ~|~ \exists p \in
\nat^* ~.~ p.i \in dom(t)\}$, and we let $\mathcal{D}_+(t) =
\mathcal{D}(t) \setminus \{-1\}$. By convention, we have $(p.i).(-1) =
p$, for all $p \in \nat^*$ and $i \in \mathcal{D}_+(t)$. Given a tree
$t$ and a position $p\in dom(t)$, we define the {\em arity} of the
position $p$ as $\#_t(p)=\max\{d \in \mathcal{D}_+(t) \mid p.d \in
dom(t)\}+1$, and the {\em subtree of $t$ rooted at $p$} as
$\subtree{t}{p}(q) = t(p.q)$, for all $q \in \nat^*$.

\ifLongVersion 

A {\em path} in $t$, from $p_1$ to $p_k$, is a sequence $p_1,p_2,\dots,p_k\in
dom(t)$ of pairwise distinct positions, such that, for all $1 \leq i < k$, there
exist $d_i \in \mathcal{D}(t)$ such that $p_{i+1}=p_i.d_i$. Notice that a path
in the tree can also link sibling nodes, not just ancestors to their descendants
or vice versa.  However, a path may not visit the same tree position twice.

Given a position $p\in dom(t)$, we denote by $\subtree{t}{p}$ the subtree of $t$
starting at position $p$, i.e.\ $t_{\mid p}(q) = t(p.q)$, for all $q \in
\nat^*$.  

\fi

\begin{definition}\label{Rotation}
Given two trees $t_1, t_2 : \nat^* \rightharpoonup_{fin} \Sigma$, we
say that $t_2$ is a {\em rotation} of $t_1$ denoted by $t_1 \sim_r
t_2$ if and only if $r: dom(t_1) \rightarrow dom(t_2)$ is a bijective
function such that: $\forall p \in dom(t_1) \forall d \in
\mathcal{D}_+(t_1) ~:~ p.d \in dom(t_1) \Rightarrow \exists e \in
\mathcal{D}(t_2) ~.~ r(p.d) = r(p).e$. We write $t_1 \sim t_2$ if
there exists a function $r : dom(t_1) \rightarrow dom(t_2)$ such that
$t_1 \sim_r t_2$.
\end{definition}

\begin{wrapfigure}[6]{r}{60mm}

  \vspace*{-9mm}

  \epsfig{file=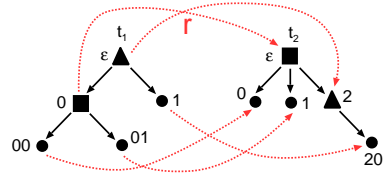, width=50mm}

  \vspace*{-4mm}
 
  \caption{An example of a rotation.}

  \vspace*{-2mm}
 
  \label{fig:exRotation}
 
\end{wrapfigure}

\enlargethispage{4mm}

An example of a rotation $r$ of a tree $t_1$ to a tree $t_2$ such that
$r(\varepsilon) = 2$, $r(0)=\varepsilon$, $r(1)=20$, $r(00)=0$, and $r(01)=1$ is
shown in Fig.~\ref{fig:exRotation}. Note
that, e.g., for $p = \varepsilon \in dom(t_1)$ and $d = 0 \in
\mathcal{D}_+(t_1)$ where $p.d = \varepsilon.0 \in dom(t_1)$, we get $e=-1 \in
\mathcal{D}(t_2)$ and $r(\varepsilon.0) = 2.(-1) = \varepsilon$.

\ifLongVersion

\begin{proposition}\label{Rotation:equivalence} The relation $\sim$ is an
equivalence relation.\end{proposition}

\begin{proof} The relation $\sim$ is clearly reflexive as one can choose $r$ as
the identity function. To prove that $\sim$ is transitive, let $t_1, t_2, t_3 :
\nat^* \rightharpoonup_{fin} \Sigma$ be trees such that $t_1 \sim_{r_1} t_2$ and
$t_2 \sim_{r_2} t_3$, and let $r_1 : dom(t_1) \rightarrow dom(t_2)$ and $r_2 :
dom(t_2) \rightarrow dom(t_3)$ be the bijective functions from Def.
\ref{Rotation}, respectively. Then, for all $p \in dom(t_1)$ and $d \in
\mathcal{D}_+(t_1)$ such that $p.d \in dom(t_1)$, there exists $e_1 \in
\mathcal{D}(t_2)$ such that $r_1(p.d)=r_1(p).e_1$. We distinguish two cases:
\begin{enumerate}

  \item If $e_1 \in \mathcal{D}_+(t_2)$, then there exists $e_2 \in
  \mathcal{D}(t_3)$ such that $(r_2 \circ r_1)(p.d) = r_2(r_1(p).e_1)
  = (r_2 \circ r_1)(p).e_2$.  

  \item Otherwise, if $e_1 = -1$, then $r_1(p) = r_1(p.d).e'_1$, for some $e'_1
  \in \mathcal{D}_+(t_2)$. Then there exists $e_2 \in \mathcal{D}(t_3)$ such
  that $(r_2 \circ r_1)(p) = r_2(r_1(p.d).e'_1) = (r_2 \circ r_1)(p.d).e_2$, and
  consequently, $(r_2 \circ r_1)(p.d) = (r_2 \circ r_1)(p).(-1)$.

\end{enumerate}

To show that $\sim$ is symmetric, let $t_1 \sim_r t_2$ be two trees, and let $p
\in dom(t_2)$ and $d \in \mathcal{D}_+(t_2)$ such that $p.d \in dom(t_2)$. We
prove that there exists $e \in \mathcal{D}(t_1)$ such that $r^{-1}(p.d) =
r^{-1}(p).e$, where $r$ is the bijective function from Def. \ref{Rotation}. By
contradiction, suppose that there is no such direction, and since $r^{-1}(p),
r^{-1}(p.d) \in dom(t_1)$, there exists a path $r^{-1}(p) = p_1, \ldots, p_i,
\ldots, p_k = r^{-1}(p.d)$ of length $k > 2$ in $t_1$ such that:\begin{itemize}

  \item $p_{j+1} = p_j.(-1)$, for all $1 \leq j < i$,

  \item $p_{j+1} = p_j.d_j$, for all $i \leq j < k$ and $d_i,\ldots,d_{k-1} \in
  \mathcal{D}_+(t_1)$

\end{itemize} where $p_1, \ldots, p_k$ are distinct positions in $dom(t_1)$.
Then $r(p_1), \ldots, r(p_k)$ are distinct positions in $dom(t_2)$ such
that:\begin{itemize}

  \item $r(p_j) = r(p_{j+1}).e_j$, for all $1 \leq j < i$ and some $e_j \in
  \mathcal{D}(t_2)$,

  \item $r(p_{j+1}) = r(p_j).e_j$, for all $i \leq j < k$ and some $e_j \in 
  \mathcal{D}(t_2)$.

\end{itemize} So there exists a path $p=r(p_1), \ldots, r(p_k)=p.d$ of length $k
> 2$ in $t_2$, which contradicts with the fact that $t_2$ is a tree.
\qed\end{proof}

A rotation function $r$ between two trees $t \sim_r u$ is said to \emph{revert}
two positions $p, q \in dom(t)$, if $q$ is a prefix of $p$ and $r(p)$ is a
prefix of $r(q)$. We show next that the only reversions in $dom(t)$ due to $r$
appear on the (unique) path from the root $\epsilon$ to $r^{-1}(\epsilon)$.

\begin{lemma}\label{rotation-reversion} Let $t,u : \nat^* \rightharpoonup_{fin}
\Sigma$ be two trees, $r : dom(t) \rightarrow dom(u)$ be a bijective function
such that $t \sim_r u$, and $p \in dom(t)$ be a position such that $r(p) =
\epsilon$ is the root of $u$. Then, for all $q \in dom(t)$ and all $0 \leq d <
\#_t(q)$, $r(q.d) = r(q).(-1)$ iff $q.d$ is a prefix of $p$.\end{lemma}

\begin{proof}

``$\Rightarrow$'' Let $q \in dom(t)$ be an arbitrary position such that $r(q.d)
= r(q).(-1)$, for some $0 \leq d < \#_t(q)$. Suppose, by contradiction, that
$q.d$ is not a prefix of $p$. There are two cases:\begin{enumerate}

  \item $p$ is a strict prefix of $q.d$, i.e. there exists a sequence
    $p = p_0, \ldots, p_k = q$, for some $k > 0$, such that, for all
    $0 \leq i < k$, there exists $0 \leq j_i < \#_t(p_i)$ such that
    $p_{i+1} = p_i.j_i$. Then, there exists $\ell \in \mathcal{D}(u)$
    such that $r(q) = r(p_{k-1}.j_{k-1}) = r(p_{k-1}).\ell$. If $\ell
    \geq 0$, we have $r(p_{k-1}) = r(q).(-1) = r(q.d) =
    r(p_{k-1}.j_{k-1}.d)$. Since both $j_{k-1}, d \geq 0$, $p_{k-1}$
    and $p_{k-1}.j_{k-1}.d$ are distinct nodes from $dom(t)$, and we
    reach a contradiction with the fact that $r$ is bijective. Hence,
    $\ell = -1$ is the only possibility. Applying the same argument
    inductively on $p_0, \ldots, p_k$, we find that $r(p)$ is either
    equal or is a descendant of $r(q)$, which contradicts with the
    fact that $r(p)=\epsilon$ is the root of $u$.
    
  \item $p$ is not a prefix of $q.d$, and there exists a common prefix
    $p_0$ of both $p$ and $q.d$. Let $p_0$ be the maximal such prefix.
    Since $p_0$ is a prefix of $q.d$, there exists a sequence $p_0,
    \ldots, p_k = q$, for some $k > 0$, such that, for all $0 \leq i <
    k$ there exists $0 \leq j_i < \#_t(p_i)$ such that $p_{i+1} =
    p_i.j_i$. By the argument of the previous case, we have that $r(q)
    = r(p_k), r(p_{k-1}), \ldots, r(p_0)$ is a strictly descending
    path in $u$, i.e.\ $r(p_i)$ is a child of $r(p_{i+1})$, for all $0
    \leq i < k$. Since $p_0$ is a prefix of $p$, there exists another
    non-trivial sequence $p_0 = p'_0, \ldots, p'_m = p$, for some $m >
    0$, such that, for all $0 \leq i < m$ there exists $0 \leq j'_i <
    \#_t(p'_i)$ and $p'_{i+1} = p'_i.j'_i$. Moreover, since $p_0$ is
    the maximal prefix of $p$ and $q.d$, we have $\{p_1, \ldots, p_k\}
    \cap \{p'_1, \ldots, p'_m \} = \emptyset$. Then $r(p'_1) =
    r(p_0.j'_0) = r(p_0).\ell$, for some $\ell \in \mathcal{D}(u)$. If
    $\ell = -1$, then necessarily $r(p'_1) = r(p_{k-1})$, which
    contradicts with the fact that $r$ is a bijection, since $p'_1
    \neq p_{k-1}$. Then the only possibility is $\ell \geq 0$. By
    applying the same argument inductively on $p'_0, \ldots, p'_m =
    p$, we obtain that $r(p)$ is a strict descendant of $r(p_0)$,
    which contradicts with $r(p)=\epsilon$ being the root of $u$.

\end{enumerate} Since both cases above lead to contradictions, the only
possibility is that $q.d$ is a prefix of $p$. 

\noindent ``$\Leftarrow$'' If $q.d$ is a prefix of $p$, there exists a
non-trivial sequence $q=q_0, q.d=q_1, \ldots, q_k=p$, for some $k >
0$, such that, for all $0 \leq i < k$, there exists $0 \leq j_i <
\#_t(q_i)$ such that $q_{i+1} = q_i.j_i$. Then, $\epsilon = r(p) =
r(q_k) = r(q_{k-1}.j_{k-1}) = r(q_{k-1}).(-1)$. Reasoning inductively,
we obtain that, for all $0 < i \leq k$, $r(q_i) = r(q_{i-1}).(-1)$,
hence $r(q.d) = r(q_1) = r(q_0).(-1) = r(q).(-1)$.  \qed\end{proof}

A rotation function $r$ between two trees $t \sim_r u$ is said to
\emph{revert} two positions $p, q \in dom(t)$, if $q$ is a prefix of
$p$ and $r(p)$ is a prefix of $r(q)$. One can see
\ifLongVersion\else(cf.  Appendix~\ref{app:trees})\fi that the only
reversions in $dom(t)$ due to $r$ appear on the (unique) path from the
root $\epsilon$ to $r^{-1}(\epsilon)$.

\fi

\begin{definition}\label{SpanningTree} Given a state $S = \langle s, h \rangle$,
a {\em spanning tree} of $S$ is a bijective tree $t: \nat^*
\rightarrow dom(h)$ such that $\forall p \in dom(t) \forall d \in
\mathcal{D}_+(t) ~.~ p.d\in dom(t) \Rightarrow \exists k \in Sel ~.~
t(p) \arrow{k}{S} t(p.d)$. An edge $\ell \arrow{k}{S} \ell'$ is said
to be {\em local} with respect to a spanning tree $t$ iff there exist
$p \in dom(t)$ and $d \in \mathcal{D}(t) \cup \{\epsilon\}$ such that
$t(p)=\ell$ and $t(p.d)=\ell'$. Moreover, $t$ is a~\emph{local
  spanning tree} of $S$ if $t$ is a spanning tree of $S$, and $S$ has
only local edges wrt. $t$.
\end{definition}

\begin{wrapfigure}[11]{r}{40mm}

  \vspace*{-10mm}

  \epsfig{file=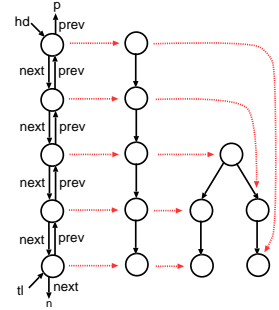, width=35mm}

  \vspace*{-4mm}
 
  \caption{Spanning trees.}

  \vspace*{-2mm}
 
  \label{fig:exSpanningTrees}
 
\end{wrapfigure}

An example of a doubly-linked list and two of its spanning trees is shown in
Fig.~\ref{fig:exSpanningTrees}. Note that both the $\mathtt{next}$ and
$\mathtt{prev}$ edges are local in both cases (they get mapped either to the $0$
or $-1$ direction) and so both of the spanning trees are local.

Notice that a spanning tree covers exactly the set of allocated locations in a
state. Moreover, if a state $S$ has a spanning tree $t$, then every location
$\ell \in loc(S)$, not necessarily allocated, is reachable from $t(\epsilon)$ by
a chain of edges. In other words, the state does not contain garbage nodes.

\begin{lemma}\label{spanning-tree-rotation} Let $S=\langle s, h \rangle$ be a
state and $t$ its spanning tree. If all edges of $S$ are local with respect to
$t$, then for each spanning tree $t'$ of $S$ we have $t \sim_{t'^{-1} \circ t}
t'$, and, moreover, all edges of $S$ are local with respect to $t'$ as well.
\end{lemma}

\ifLongVersion

\begin{proof} Let $p \in dom(t)$ be an arbitrary position and $d \in
\mathcal{D}_+(t)$ be an arbitrary non-negative direction such that $p.d \in
dom(t)$. If $t$ is a spanning tree, then $t$ is bijective, and there exists a
selector $s \in Sel$ such that $t(p) \arrow{s}{} t(p.d)$ is an edge in $S$.
Since $t'$ is a spanning tree, it is also bijective, hence there exist $p',p''
\in dom(t')$ such that $t'(p')=t(p)$ and $t'(p'') = t(p.d)$. 

By contradiction, suppose that $p'' \neq p'.e$ for all $e \in \mathcal{D}(t')$.
Since $p',p'' \in dom(t')$, there exists a path $p' = p_1, \ldots, p_k$ such
that $k > 2$, and, for all $1 \leq i < k$, we have $p_{i+1} = p_i.e_i$, for some
$e_1, \ldots, e_{k-1} \in \mathcal{D}(t')$. Because $t'$ is a spanning tree of
$S$, there exist selectors $s_1, \ldots, s_{k-1} \in Sel$ such that one of the
following holds, for all $1 \leq i < k$:\begin{itemize}

  \item $t'(p_i) \arrow{s_i}{} t'(p_{i+1})$ or

  \item $t'(p_{i+1}) \arrow{s_i}{} t'(p_i)$.

\end{itemize} Since all the above edges of $S$ are local wrt. $t$, there exists
a path $p = (t^{-1} \circ t')(p_1), \ldots,$ $(t^{-1} \circ t')(p_k) = p.d$ of
pairwise distinct positions in $t$, for $k > 2$, which contradicts with the fact
that $t$ is a tree. In conclusion, for all $p \in dom(t)$ and $d \in
\mathcal{D}_+(t)$, such that $p.d \in dom(t)$ there exists $e \in
\mathcal{D}(t')$ such that $(t'^{-1} \circ t)(p.d) = (t'^{-1} \circ t)(p).e$,
and since $t'^{-1} \circ t$ is a bijective mapping, we have $t \sim_{t'^{-1}
\circ t} t'$.

Finally, we are left with proving that all edges are local with
respect to $t'$. Let $\ell \arrow{s}{} \ell'$ be an arbitrary edge of
$S$, for some $s \in Sel$. The edge is local with respect to $t$, thus
there exists $p \in dom(t)$ and $d \in \mathcal{D}(t) \cup
\{\epsilon\}$ such that $t(p) = \ell$ and $t(p.d) = \ell'$. We
distinguish three cases:\begin{enumerate}

  \item If $d = \epsilon$, then $p = p.d$, and trivially $\ell =
  t'((t'^{-1} \circ t)(p)) = t'((t'^{-1} \circ t)(p.d)) = t'((t'^{-1}
  \circ t)(p).d) = \ell'$. 

  \item Otherwise, if $d \in \mathcal{D}_+(t)$, taking into account that $t
  \sim_{t'^{-1} \circ t} t'$, by the first part of this lemma, we have
  $(t'^{-1}\circ t)(p.d) = (t'^{-1}\circ t)(p).e$, for some $e \in
  \mathcal{D}(t')$. But $t'((t'^{-1}\circ t)(p)) = t(p) = \ell$ and
  $t'((t'^{-1}\circ t)(p.d)) = t(p.d) = \ell'$, so $\ell \arrow{s}{} \ell'$ is
  local wrt. $t'$.

  \item Otherwise, if $d = -1$, there exists $q \in dom(t)$ and $d' \in
  \mathcal{D}_+(t)$ such that $p = q.d'$. Since $t \sim_{t'^{-1} \circ
    t} t'$, by the first part of this lemma, we have $(t'^{-1}\circ
  t)(q.d') = (t'^{-1}\circ t)(q).e$, for some $e \in
  \mathcal{D}(t')$. But $t'(t'^{-1}\circ t)(q.d')=t(p)=\ell$ and
  $t'((t'^{-1}\circ t)(q))=t(q)=t(p.d)=\ell'$, so $\ell \arrow{s}{}
  \ell'$ is local with respect to $t'$.

\end{enumerate}
\qed\end{proof} 

\fi

As a consequence, if all edges of a state are local with respect to
some spanning tree, they are also local with respect to any other
spanning tree, hence we will simply say that they are local. A state
is said to be {\em local} if it has only local edges.

\subsection{Separation Logic}\vspace*{-1mm}

The syntax of {\em basic formulae} of Separation Logic (SL) is given
below:\vspace*{-1mm}
\[\begin{array}{lcl}
  \alpha & \in & Var \setminus \{\nil\};~ x ~\in~ Var; \\  
  \Pi & ::= & \alpha = x ~|~ \Pi_1 \wedge \Pi_2 \\ 
  \Sigma & ::= & \emp ~|~ \alpha \mapsto (x_1,\ldots,x_n) ~|~ \Sigma_1 * \Sigma_2 
  ~\mbox{, for some}~ n > 0 \\ 
  \varphi & ::= & \Sigma \wedge \Pi ~|~ \exists x ~.~ \varphi_1\vspace*{-1mm}
\end{array}\]
A formula of the form \mbox{$\bigwedge_{i=1}^n \alpha_i = x_i$}
defined by the $\Pi$ nonterminal in the syntax above is said to be
\emph{pure}. A formula of the form $\bigstar_{i=1}^k \alpha_i \mapsto
(x_{i,1},\ldots,x_{i,n})$ defined by the $\Sigma$ nonterminal in the
syntax above is said to be \emph{spatial}. The atomic proposition
$\emp$ denotes the empty spatial conjunction. A variable $x$ is said
to be {\em free} in $\varphi$ if it does not occur under the scope of
any existential quantifier.  We denote by $FV(\varphi)$ the set of
free variables, and by $AP(\varphi)$ the set of atomic propositions of
$\varphi$.

In the following, we shall use two equality relations. The {\em syntactic
equality}, denoted $\alpha \equiv \beta$, means that $\alpha$ and $\beta$ are the
same syntactic object (formula, variable, tuples of variables, etc.). On the
other hand, by writing $x =_\Pi y$, for two variables $x,y \in Var$ and a pure
formula $\Pi$, we mean that the equality of the values of $x$ and $y$ is implied
by $\Pi$.

A {\em substitution} is an injective partial function $\sigma : Var
\rightharpoonup_{fin} Var$. Given a basic formula $\varphi$ and a
substitution $\sigma$, we denote by $\varphi[\sigma]$ the result of
simultaneously replacing each variable (not necessarily free) $x$ that
occurs in $\varphi$, by $\sigma(x)$. For instance, if $\sigma(x)=y$,
$\sigma(y)=z$ and $\sigma(z)=t$, then $(\exists x,y ~.~ x \mapsto
(y,z) \wedge z = x)[\sigma] \equiv \exists y,z ~.~ y \mapsto (z,t)
~\wedge~ t=y$.

The semantics of a basic formula $\varphi$ is given by the relation $S
\models \varphi$, where $S=\langle s, h \rangle$ is a state such that
$FV(\varphi) \subseteq dom(s)$, and $\varphi$ is a basic SL
formula. The definition of $\models$ is by induction on the structure
of $\varphi$:\vspace*{-1mm}
\[\begin{array}{lcl}
  S \models \emp & \iff & dom(h) = \emptyset \\
  S \models \alpha \mapsto (x_1,\ldots,x_n) & \iff & 
  s = \{(\alpha,\ell_0), (x_1, \ell_1), \ldots, (x_n, \ell_n)\} ~\mbox{and}~ \\ 
  && h = \{\langle \ell_0, \lambda i ~.~ 
  \mbox{if}~ 1 \leq i \leq n ~\mbox{then}~ \ell_i ~\mbox{else}~ \bot \rangle\} \\
  && ~\mbox{for some $\ell_0, \ell_1, \ldots, \ell_n \in Loc$} \\
  S \models \varphi_1 * \varphi_2 & \iff & 
  S_1 \models \varphi_1 ~\mbox{and}~ S_2 \models \varphi_2 
  ~\mbox{where $S_1 \uplus S_2 = S$} \\
  S \models \exists x ~.~ \varphi & \iff & 
  \langle s[x \leftarrow \ell], h \rangle \models \varphi ~\mbox{for some $\ell \in Loc$}
\end{array}\]

\vspace*{-1mm}The semantics of $=$ and $\wedge$ is classical in first order
logic. Note that we adopt here the \emph{strict semantics}, in which a
points-to relation $\alpha \mapsto (x_1,\ldots,x_n)$ holds in a state
consisting of a single cell pointed to by $\alpha$, with exactly $n$
outgoing edges $s(\alpha) \arrow{k}{S} s(x_k)$, $1 \leq k \leq n$,
towards either the single allocated location (if $s(x_k)=s(\alpha)$),
or dangling locations (if $s(x_k) \neq s(\alpha)$). The empty heap is
specified by $\emp$.

\ifLongVersion

A variable $x \in FV(\Sigma)$ is said to be {\em allocated} in a basic
spatial formula $\Sigma$ if it occurs on the right hand side of a
points-to atomic proposition $x \mapsto (y_0, \ldots, y_{k-1})$ of
$\Sigma$. If $\Sigma$ is satisfiable, then clearly each free variable
$x \in FV(\Sigma)$ is allocated at most once. For a basic
quantifier-free SL formula $\varphi \equiv \Sigma \wedge \Pi$ and two
variables $x,y \in FV(\varphi)$, we say that $y$ is $\varphi$-{\em
  reachable} from $x$ in iff there exists a sequence $x =_\Pi
\alpha_0, \ldots, \alpha_m =_\Pi y$, for some $m \geq 0$, such that,
for each $0 \leq i < m$, $\alpha_i \mapsto (\beta_{i,1}, \ldots,
\beta_{i,p_i})$ is a points-to proposition in $\Sigma$, and
$\beta_{i,s} =_\Pi \alpha_{i+1}$, for some $1 \leq s \leq p_i$. A
variable $x \in FV(\Sigma)$ is said to be a {\em root} of $\Sigma$ if
every variable $y \in FV(\Sigma)$ is reachable from $x$.


\paragraph{Remark.} Notice that there is no explicit disequality between
variables in the basic fragment of SL. This is in part justified by the fact
that disequality can be partially defined using the following implication:
$\alpha \mapsto (x_1,\ldots,x_n) * \beta \mapsto(y_1,\ldots,y_m) \Rightarrow
\alpha \neq \beta$. This implication enforces disequality only between allocated
variables. However, in practice, most SL specifications do not impose
disequality constraints on dangling variables.

\fi

\enlargethispage{6mm}

\vspace*{-2mm}\subsection{Inductive Definitions}\vspace*{-1mm}
\label{sec:inductive-definitions}

A system of \emph{inductive definitions}\footnote{The name {\em inductive}
suggests that every structure described by the system is finite.} (inductive
system) $\mathcal{P}$ is a set of rules of the form:\vspace*{-2mm}
\begin{equation}\label{recursive-definitions}
\begin{array}{rcl}
P_1(x_{1,1},\ldots,x_{1,n_1}) & \equiv & \mid_{j=1}^{m_1}
R_{1,j}(x_{1,1},\ldots,x_{1,n_1}) 
\\ & \ldots & \\ 
P_k(x_{k,1},\ldots,x_{k,n_k}) & \equiv & \mid_{j=1}^{m_k}
R_{k,j}(x_{k,1},\ldots, x_{k,n_k}) \vspace*{-1mm}
\end{array}
\end{equation}
where $\{P_1, \ldots, P_k\}$ is a set of {\em predicates}, $x_{i,1},
\ldots, x_{i,n_i}$ are called {\em formal parameters}, and the
formulae $R_{i,j}$ are called the {\em rules} of $P_i$. Concretely,
each rule is of the form:\vspace*{-1mm} $$R_{i,j}(\vec{x}) \equiv \exists \vec{z}
~.~ \Sigma * P_{i_1}(\vec{y}_1) * \ldots * P_{i_m}(\vec{y}_m) ~\wedge~
\Pi$$ where $\vec{x} \cap \vec{z} = \emptyset$, and all of the
following hold:\vspace*{-1mm}\begin{enumerate}

  \item $\head(R_{i,j}) \defequiv \Sigma$ is a non-empty spatial
    formula (i.e.\ $\Sigma \not\equiv \emp$), and $FV(\Sigma)
    \subseteq \vec{x} \cup \vec{z}$,

  \item $\tail(R_{i,j}) \defequiv \langle P_{i_1}(\vec{y}_1), \ldots,
    P_{i_m}(\vec{y}_m) \rangle$ is an ordered sequence of {\em
      predicate occurrences}, where $\vec{y}_1 \cup \ldots \cup
    \vec{y}_m \subseteq \vec{x} \cup \vec{z}$,

  \item $\Pi$ is a pure formula, and $FV(\Pi) \subseteq \vec{z} \cup
    \vec{x}$.  In the following, we restrict the pure part of each
    rule such that, for all formal parameters $\beta \in \vec{x}$, we
    allow only equalities of the form $\alpha =_\Pi \beta$, where
    $\alpha$ is allocated in $\Sigma$. This restriction is of
    technical nature (see Section \ref{sec:canonization}). It is
    possible to lift it, but only at the expense of an exponential
    blowup in the size of the resulting tree automaton.

  \item For all $1 \leq r,s \leq m$, if $x_{i,k} \in \vec{y}_r$,
    $x_{i,l} \in \vec{y}_s$, and $x_{i,k} =_\Pi x_{i,l}$, for some $1
    \leq k,l \leq n_i$, then $r=s$. In other words, a formal parameter
    of a rule cannot be passed to two or more subsequent occurrences
    of predicates in that rule. This technical restriction can be
    lifted at the cost of introducing expensive tests for double
    allocation as was done in the translation of the inductive
    definitions to Monadic Second-Order Logic on graphs, reported in
    \cite{Iosif13}.
\end{enumerate}
The {\em size} of a rule $R$ is denoted by $\len{R}$ and is defined
inductively as follows:\vspace*{-1mm}
\[\begin{array}{ccccc}
\len{\alpha = x} = 1 & ~ & \len{\emp} = 1 & ~ & \len{\alpha \mapsto (x_1, \ldots,
  x_n)} = n + 1 \\
\len{\varphi \bullet \psi} = \len{\varphi} + \len{\psi} & ~ & \len{\exists x ~.~ 
  \varphi} = \len{\varphi} + 1 & ~ & 
\len{P(x_1, \ldots, x_n)} = n 
\end{array}\]
where $\alpha \in Var \setminus \{\nil\}$, $x, x_1, \ldots, x_n \in
Var$, and $\bullet \in \{*, \wedge\}$. The size of an inductive system
$\mathcal{P} = \{P_i \equiv |_{j=1}^{m_i} R_{i,j}\}_{i=1}^n$ is
defined as $\len{\mathcal{P}} = \sum_{i=1}^n \sum_{j=1}^{m_i}
\len{R_{i,j}}$.

\enlargethispage{6mm}

\begin{example}\label{ex:DLL} To illustrate the use inductive definitions (with
the above restrictions), we first show how to define a predicate
$\DLL(hd,p,tl,n)$ describing doubly-linked lists of length at least one. As
depicted on the left of Fig.~\ref{fig:exSpanningTrees}, the formal parameter
$hd$ points to the first allocated node of such a list, $p$ to the node pointed
by the $prev$ selector of $hd$, $tl$ to the last node of the list (possibly
equal to $hd$), and $n$ to the node pointed by the $next$ selector from $tl$.
This predicate can be defined as follows:\vspace*{-1mm} $$\DLL(hd,p,tl,n) \equiv
hd \mapsto (n,p) ~\wedge~ hd=tl ~\mid~ \exists x.~ hd \mapsto (x,p) *
\DLL(x,hd,tl,n)\vspace*{-1.5\baselineskip}$$ ~ 

\end{example}


\ifLongVersion
\begin{figure}[htb]
\begin{center}
\else
\begin{wrapfigure}[7]{r}{50mm}
\fi

  \vspace*{-8mm}

  \epsfig{file=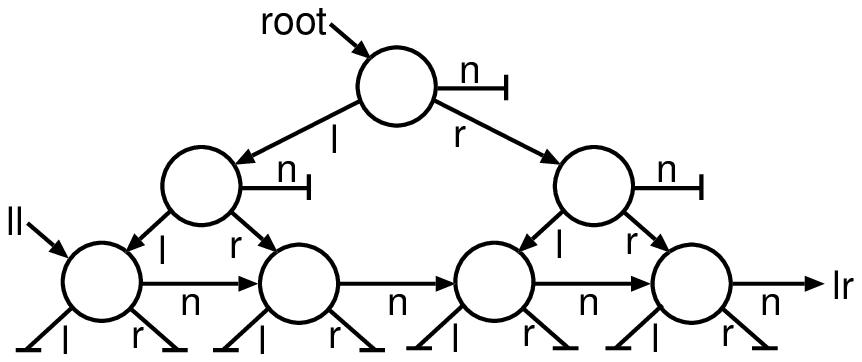, width=50mm}

  \vspace*{-5mm}
 
  \caption{A tree with linked leaves.}

  \vspace*{-2mm}
 
  \label{fig:exTLL}
 
\ifLongVersion
\end{center}
\end{figure}
\else
\end{wrapfigure}
\fi
\begin{example}\label{ex:TLL} Another example is the predicate
$\TLL(root,ll,lr)$ describing binary trees with linked leaves, whose root is
pointed by the $root$ formal parameter, the left-most leaf is pointed to by
$ll$, and the right-most leaf points to $lr$ (cf.
Fig.~\ref{fig:exTLL}):\vspace{-1mm} \[\begin{array}{rcl} \TLL(root,ll,lr) &
\equiv & root \mapsto (\nil,\nil,lr) ~\wedge~ root=ll \\ && \mid~ \exists x, y,
z.~ root \mapsto (x,y,\nil) * \TLL(x,ll,z) * \TLL(y,z,lr)\end{array}\]

\vspace*{-1.3\baselineskip} \ 
\end{example}


\ifLongVersion

\paragraph{Remark.} The first restriction above ($\head(R_{i,j}) \not\equiv
\emp$) currently does not allow defining data structures of zero size (such as
empty lists, trees, etc.) and, more importantly, certain types of structure
concatenation. For instance, considering the $\DLL$ predicate from the above
example, the following predicate cannot be defined in our
framework\footnote{However, defining concatenation of trees with linked leaves
is possible when the root is allocated: $\mathtt{JOIN\_TLL}(root,ll,lr) \equiv
\exists x,y,z ~.~ root \mapsto (x,y,\nil) * \TLL(x,ll,z) * \TLL(y,z,lr)$.}:
\[\mathtt{JOIN\_DLL}(hd,p,tl,n) \equiv \exists x,y ~.~ \DLL(hd,p,x,y) *
\DLL(x,y,tl,n)\] An extension to handle such definitions is possible,
by introducing a special ``concatenation'' operation at the level of
tree automata, which is considered for further
investigation. \qed\vspace*{\baselineskip}

\fi


\begin{definition}\label{connected-rule} Given a system $\mathcal{P} = \{P_i
\equiv |_{j=1}^{m_i} R_{i,j}\}_{i=1}^n$ of inductive definitions, a rule
$R_{i,j}(x_{i,1},\ldots,x_{i,k}) \equiv \exists \vec{z} ~.~ \Sigma *
P_{i_1}(\vec{y}_1) * \ldots * P_{i_m}(\vec{y}_m) \wedge \Pi$ is {\em connected}
if and only if there exists a formal parameter $x_{i,\ell}$, $1 \leq \ell \leq
k$ such that: \begin{itemize}
 
  \item $x_{i,\ell}$ is a root of $\Sigma$, and 
  
  \item for each $j = 1, \ldots, m$, there exists $0 \leq s < \len{\vec{y}_j}$
  such that $(\vec{y}_j)_s$ is $\Sigma \wedge \Pi$-reachable from $x_{i,\ell}$,
  and $x_{i_j,s}$ is a root of the head of each rule of $P_{i_j}$.

\end{itemize}
\end{definition}An inductive system is said to be {\em connected} if all its
rules are connected. In the rest of this section, we consider only connected
systems. This condition is necessary, since later on it is proved that
entailment between predicates of disconnected systems is undecidable, in general
(Thm.  \ref{not-connected-undecidable} in Sec. \ref{sec:complexity}). Notice
that the $\DLL$ and $\TLL$ systems from Examples~\ref{ex:DLL} and \ref{ex:TLL}
are both connected.

\subsection{Tree Automata}\vspace*{-1mm}

A (finite, non-deterministic, bottom-up) \emph{tree automaton}
(abbreviated as TA in the following) is a quadruple $A = \langle Q,
\Sigma,\Delta, F \rangle$, where $\Sigma$ is a finite alphabet, $Q$ is
a finite set of {\em states}, $F \subseteq Q$ is a set of {\em final
  states}, $\Sigma$ is an alphabet, and $\Delta$ is a set of
transition rules of the form $\sigma(q_1, \ldots, q_n) \rightarrow q$,
for $\sigma \in \Sigma$, and $q,q_1,\ldots,q_n \in Q$. Given a tree
automaton $A = \langle Q, \Sigma, \Delta, F \rangle$, for each rule
$\rho = (\sigma(q_1,\ldots,q_n) \arrow{}{} q)$, we define its size as
$\len{\rho} = n+1$. The size of the tree automaton is $\len{A} =
\sum_{\rho\in\Delta} \len{\rho}$. A \emph{run} of $A$ over a tree $t :
\nat^* \rightharpoonup_{fin} \Sigma$ is a function $\pi : dom(t)
\rightarrow Q$ such that, for each node $p \in dom(t)$, where $q =
\pi(p)$, if $q_i = \pi(p.i)$ for $1 \leq i \leq n$, then $\Delta$ has
a~rule $(t(p))(q_1, \ldots, q_n) \rightarrow q$. We write
$t\stackrel{\pi}{\Longrightarrow} q$ to denote that $\pi$ is a run of
$A$ over $t$ such that $\pi(\epsilon)=q$. We use $t\Longrightarrow q$
to denote that $t\stackrel{\pi}{\Longrightarrow} q$ for some run
$\pi$. The \emph{language} of a $A$ is defined as $\lang{A} = \{t \mid
\exists q \in F,~ t \Longrightarrow q\}$.

\vspace*{-2mm}\section{Tiles and Semantics of Inductive
Definitions}\vspace*{-1mm}

A {\em tile} is a tuple $T=\langle \varphi, \vec{x}_{-1}, \vec{x}_0, \ldots,
\vec{x}_{d-1} \rangle$, for some $d \geq 0$, where $\varphi$ is a basic SL
formula and each $\vec{x}_i$ is a tuple of pairwise distinct variables, called a
{\em port}, such that $\vec{x}_i \cap \vec{x}_j = \emptyset$, for all $-1 \leq i
< j < d$, and $\vec{x}_{-1} \cup \vec{x}_0 \cup \ldots \cup \vec{x}_{d-1}
\subseteq FV(\varphi)$. The set of all tiles is denoted by $\mathcal{T}$. The
variables from $\vec{x}_{-1}$ are said to be {\em incoming}, the ones from
$\vec{x}_0, \ldots, \vec{x}_{d-1}$, are said to be {\em outgoing}, and the ones
from $\prm(T) = FV(\varphi) \setminus (\vec{x}_{-1} \cup \vec{x}_0 \cup \ldots
\cup \vec{x}_{d-1})$ are called {\em parameters}. The {\em arity} of a tile
$T=\langle \varphi, \vec{x}_{-1}, \vec{x}_0, \ldots, \vec{x}_{d-1} \rangle$ is
the number of outgoing ports, denoted by $\#(T)=d$. In the following, we denote
$\form(T) \equiv \varphi$ and $\port_i(T) \equiv \vec{x}_i$, for all $i = -1, 0,
\ldots, d-1$. 

Given two tiles $T_1 = \langle \varphi_1, \vec{x}^1_{-1}, \vec{x}^1_0,
\ldots, \vec{x}^1_{d-1} \rangle$ and $T_2 = \langle \varphi_2,
\vec{x}^2_{-1}, \vec{x}^2_0, \ldots, \vec{x}^2_{e-1} \rangle$ such
that $FV(\varphi_1) \cap FV(\varphi_2) = \emptyset$, for some $0 \leq i
\leq d-1$ such that $\len{\vec{x}^1_i}=\len{\vec{x}^2_{-1}}$, we
define their {\em $i$-composition}:
\ifLongVersion
\[\begin{array}{rcl}
T_1 \circledast_i T_2 & = & \langle \psi, \vec{x}^1_{-1}, 
\vec{x}^1_0, \ldots \vec{x}^1_{i-1}, \vec{x}^2_0, \ldots, \vec{x}^2_{e-1}, 
\vec{x}^1_{i+1}, \ldots, \vec{x}^1_{d-1} \rangle \\
\mbox{where}~ \psi & \equiv & \exists \vec{x}^1_i \exists \vec{x}^2_{-1} ~.~ 
\varphi_1 * \varphi_2 \wedge \vec{x}^1_i = \vec{x}^2_{-1}.
\end{array}\]
\fi
For a tree position $q \in \nat^*$ and a tile $T = \langle \varphi,
\vec{x}_{-1}, \vec{x}_0, \ldots,$ $\vec{x}_{d-1} \rangle$, we denote by
$T^q$ the tile obtained by renaming every free variable $x \in
\vec{x}_{-1} \cup \vec{x}_0 \cup \ldots \cup \vec{x}_{d-1}$ by
$x^q$. Note that the parameters $x \in \prm(T)$ are not changed by
this renaming.

\enlargethispage{4mm}

\begin{definition}\label{tiled-tree}
A {\em tiled tree} is a tree $t : \nat^* \rightharpoonup_{fin}
\mathcal{T}$ such that, for all positions $p \in dom(t)$, the
following holds:\vspace*{-1mm}
\begin{itemize}
\item $\#_t(p) = \#(t(p))$, i.e.\ the arity of $p$ equals the arity of
  its label in $t$, and
\item for all $0 \leq i < \#_t(p)$, the $d$-composition $t(p)^p
  ~\circledast_i~ t(p.i)^{p.i}$ is defined.
\end{itemize}
\end{definition}
\vspace{-1mm}A tiled tree $t$ corresponds to a tile defined inductively, for any $p
\in dom(t)$, as:\vspace{-1mm}
\[\Phi(t,p) = t(p)^p \circledast_0 \Phi(t,p.0) \circledast_1 \Phi(t,p.1) 
~\ldots~ \circledast_{\#(p)-1} \Phi(t,p.(\#_t(p)-1)).\vspace*{-1mm}\]
The tile $\Phi(t,\epsilon)$ is said to be the {\em characteristic tile} of $t$
and it is denoted in the following as $\Phi(t)$. It can be easily shown, by
induction on the structure of the tiled tree $t$, that $\Phi(t) = \langle \psi,
\vec{x}_{-1} \rangle$, for an SL formula $\psi$ and an incoming port
$\vec{x}_{-1}$. In other words, $\Phi(t)$ has no outgoing ports. In this case,
we write $S \models \Phi(t)$ for $S \models \psi$.

Given an inductive system $\mathcal{P} = \{P_i \equiv |_{j=1}^{m_i}
R_{i,j}\}_{i=1}^n$ of the form (\ref{recursive-definitions}), for each
rule $R_{i,j}(\vec{x}) \equiv \exists \vec{u} ~.~ \Sigma *
P_{i_0}(\vec{y}_0) * \ldots * P_{i_{n-1}}(\vec{y}_{d-1}) ~\wedge~
\Pi$, where $\vec{x} \cap \vec{u} = \emptyset$, we define a tile
$T_{i,j} = \langle \varphi, \vec{x}, \vec{z}_0, \ldots, \vec{z}_{d-1}
\rangle$, where $\vec{z}_i$ are disjoint from $\vec{y}_j$, for all $0
\leq i,j < d$, and $\varphi = \exists \vec{u} ~.~ \Sigma \wedge \Pi
~\wedge~ \bigwedge_{i=0}^{d-1} \vec{z}_i = \vec{y}_i$.
 

\begin{definition}\label{unfolding-tree}
  Let $\mathcal{P} = \big\{P_i ~\equiv~ \mid_{j=1}^{m_i}
  R_{i,j}\big\}_{i=1}^{n}$ be an inductive system
  (\ref{recursive-definitions}). An \emph{unfolding tree} of
  $\mathcal{P}$ is a tiled tree $t : \nat^* \rightharpoonup_{fin}
  \mathcal{T}$ such that:\vspace*{-2mm}
  \begin{itemize}
    \item $t(\epsilon) = T_{i,j}$, for some $1 \leq j \leq m_i$, and
    \item for every position $p \in dom(t)$, if $t(p)= T_{i,j}$ and
      $\tail(R_{i,j}) = \langle P_{i_0}, \ldots, P_{i_{d-1}} \rangle$,
      where $d = \#(T_{i,j})$, then for all $0 \leq j < d$, we have
      $t(p.j) = T_{i_j,k}$, for some $1 \leq k \leq m_{i_j}$.
  \end{itemize}
  \vspace*{-2mm}If $t(\epsilon) = T_{i,j}$, for some $1 \leq j \leq m_i$, we say
  that $t$ is $i$-rooted.
\end{definition}
We denote by $\mathcal{T}_i(\mathcal{P})$ the set of $i$-rooted
unfolding trees of $\mathcal{P}$. The semantics of a predicate
$P_i(\vec{x}) \in \mathcal{P}$ is defined as follows:
$S \models P_i(\vec{x}) \iff S \models \Phi(t)$, for some $t \in
  \mathcal{T}_i(\mathcal{P})$.
Given an inductive system $\mathcal{P}$ and two predicates $P_i(x_1, \ldots,
x_n)$ and $P_j(y_1, \ldots, y_n)$ of $\mathcal{P}$, with the same number of
formal parameters $n$, and a tuple of variables $\vec{x}$, where $\len{\vec{x}}
= n$, the {\em entailment problem} is defined as follows:
$P_i(\vec{x}) \models_{\mathcal{P}} P_j(\vec{x}) ~:~ \forall S ~.~ S \models
P_i(\vec{x}) \Rightarrow S \models P_j(\vec{x}).$ A {\em rooted}
inductive system $\langle \mathcal{P}, P_i \rangle$ is an inductive system
$\mathcal{P}$ with a designated predicate $P_i \in \mathcal{P}$. Two rooted
systems $\langle \mathcal{P}, P_i \rangle$ and $\langle \mathcal{Q}, Q_j
\rangle$ are said to be {\em equivalent} if and only if $P_i
\models_{\mathcal{P} \cup \mathcal{Q}} Q_j$ and $Q_j \models_{\mathcal{P} \cup
\mathcal{Q}} P_i$.

\begin{definition}\label{local-recursive-system}
A rooted inductive system $\langle \{P_1, \ldots, P_n\}, P_i \rangle$
is said to be {\em local} if and only if for each unfolding tree $t
\in \mathcal{T}_i(\mathcal{P})$ and each state $S \models \Phi(t)$,
$S$ is local.
\end{definition}
The {\em locality problem} asks, given an inductive system
$\mathcal{P} = \{P_1, \ldots, P_n\}$, and an index $i = 1,\ldots,n$,
whether the rooted system $\langle \mathcal{P}, P_i \rangle$ is
local. Our method for deciding entailment problems of the form $P_i
\models_{\mathcal{P}} P_j$ is shown to be sound and complete provided
that both $\langle \mathcal{P}, P_i \rangle$ and $\langle \mathcal{P},
P_j \rangle$ are local. Otherwise, if one system is not local, our
method is sound, i.e.\ the algorithm returns ``yes'' only if the
entailment holds. We describe further a canonization procedure
(Sec. \ref{sec:canonization}) which performs a sufficient locality
test on the system, prior to the encoding of the inductive system as a
tree automaton.

\enlargethispage{5mm}

\vspace*{-2mm}\subsection{Canonical Tiles}\vspace*{-1mm}\label{sec:canonical-tiles}

This section defines a class of \emph{canonically tiled trees} (or,
to be short, \emph{canonical trees}) intended to reduce the number of ways in
which a given local state can be encoded. Moreover, it is shown that despite a
local state can still be described by several different canonical trees, such
trees must be in a rotation relation called a \emph{canonical rotation}. This
fact is the basis of the completeness argument of our method. By defining a
closure of tree automata (accepting canonical trees of states described by a
given inductive predicate) under canonical rotations, we ensure that any
entailment between two local inductive predicates can be reduced to a language
inclusion problem. The case of non-local states is subsequently dealt with in a
sound (yet, in general, not complete) way in the next subsection (Sec.
\ref{sec:quasi-canonical-tiles}).

A tile $T = \langle \varphi, \vec{x}_{-1}, \vec{x}_0, \ldots, \vec{x}_{d-1}
\rangle$ is said to be a {\em singleton} if $\varphi$ is of one of the
forms:\vspace*{-2mm}
\begin{enumerate}
\item $\exists z ~.~ z \mapsto (y_0,\ldots,y_{m-1}) \wedge \Pi$ or
\item $z \mapsto (y_0,\ldots,y_{m-1}) \wedge \Pi$ and $z \in \prm(T)$
\end{enumerate}
\vspace*{-2mm}and the following holds:\vspace*{-2mm}
\begin{itemize}

  \item For all $-1 \leq i < d$, $\vec{x}_i \cap (\{z,\nil\} \cup \prm(T)) =
  \emptyset$, i.e.\ neither the incoming nor the outgoing tuples of variables
  contain $z$, $\nil$, or parameters.

  \item For each $0 \leq j < m$, exactly one of the following holds: either (i)
  $y_j \equiv z$, (ii) $y_j \equiv \nil$, (iii) $y_j \in \prm(T)$, or (iv) there
  exists a unique tuple of variables $\vec{x}_i$, $-1 \leq i < d$, such that
  $y_j$ occurs in $\vec{x}_i$.

  \item For each outgoing tuple of variables $\vec{x}_i$, $0 \leq i < d$, there
  exists a variable $y_j \not\in \{z,\nil\} \cup \prm(T)$, for some $0 \leq j <
  m$, such that $y_j$ occurs in $\vec{x}_i$.

\end{itemize} 
We denote the spatial formula of a singleton tile as $(\exists z)~ z
\mapsto (y_0, \ldots, y_{m-1})$ in order to account for the optional
quantification of the allocated variable $z$. A singleton tile $T =
\langle (\exists z)~ z \mapsto (y_0,\ldots,y_{m-1}) \wedge \Pi,~
\vec{x}_{-1}, \vec{x}_0, \ldots, \vec{x}_{d-1} \rangle$ is said to be
    {\em canonical} if, moreover, for all $-1 \leq i < d$, $\vec{x}_i$
    can be factorized as $\vec{x}_i \equiv \vec{x}^{fw}_i \cdot
    \vec{x}^{bw}_i$ such that: 
\begin{enumerate}

\item $\vec{x}^{bw}_{-1} = \langle y_{h_0}, \ldots, y_{h_k} \rangle$,
  for some ordered sequence $0 \leq h_0 < \ldots < h_k < m$, i.e.\ the
  incoming tuple is ordered by the selectors referencing its elements.

\item For all $0 \leq i < d$, $\vec{x}^{fw}_i \equiv \langle y_{j_0},
  \ldots, y_{j_{k_i}} \rangle$, for some ordered sequence $0 \leq j_0
  < \ldots < j_k < m$, i.e.\ each outgoing tuple is ordered by the
  selectors referencing its elements.

\item For all $0 \leq i, j < d$, if $(\vec{x}_i^{fw})_0 \equiv y_p$
  and $(\vec{x}_j^{fw})_0 \equiv y_q$, for some $0 \leq p < q < m$,
  then $i < j$, i.e.\ outgoing tuples are ordered by the selector
  referencing the first element.

\item $(\vec{x}_{-1}^{fw} \cup \vec{x}_0^{bw} \cup \ldots \cup
  \vec{x}_{d-1}^{bw}) \cap \{y_0, \ldots, y_{m-1}\} = \emptyset$ and
  $\Pi \equiv \vec{x}_{-1}^{fw} = z ~\wedge~ \bigwedge_{i=0}^{d-1}
  \vec{x}_i^{bw} = z$.
\end{enumerate}

Given a canonical tile $T = \langle (\exists z) ~ z \mapsto (y_0,\ldots,y_{m-1})
\wedge \Pi,~ \vec{x}_{-1}, \vec{x}_0, \ldots, \vec{x}_{d-1} \rangle$, the
decomposition of each tuple $\vec{x}_i$ into $\vec{x}^{fw}_i$ and
$\vec{x}^{bw}_i$ is unique. This is because $\vec{x}^{fw}_i$ contains all
referenced variables $y_j$ that occur in an output port $\vec{x}_i$ (since
$\vec{x}^{bw}_i \cap \{y_0, \ldots, y_{m-1}\} = \emptyset$) in the same order in
which they are referenced in $z \mapsto (y_0, \ldots, y_{m-1})$. The same holds
for the input port $\vec{x}_{-1}$, with the roles of $\vec{x}_{-1}^{fw}$ and
$\vec{x}_{-1}^{bw}$ swapped. As each outgoing port must contain at least one
such variable, by the definition of singleton tiles, we also have that
$\vec{x}^{fw}_i \neq \emptyset$, for all $0 \leq i < d$. We denote further by
$\port^{fw}_i(T)$ and $\port^{bw}_i(T)$ the tuples $\vec{x}^{fw}_i$ and
$\vec{x}^{bw}_i$, resp., for all $i = -1, \ldots, d-1$. The set of canonical
tiles is denoted as $\mathcal{T}^c$.

\begin{definition}\label{Canonically-tiled}
A tiled tree $t : \nat^* \rightharpoonup_{fin} \mathcal{T}^c$ is said
to be {\em canonical} if and only if, for any $p \in dom(t)$ and each
$0 \leq i < \#_t(p)$, we have $\len{\port^{fw}_i(t(p))} =
\len{\port^{fw}_{-1}(t(p.i))}$ and $\len{\port^{bw}_i(t(p))} =
\len{\port^{bw}_{-1}(t(p.i))}$.
\end{definition}

\begin{example}[cont. of Ex.~\ref{ex:DLL}]\label{ex:DLLtree}To
illustrate the notion of canonical trees, Fig.~\ref{fig:DLLtiles}
shows two canonical trees for a~given DLL. The tiles are depicted as
big rectangles containing the appropriate basic formula as well the
input and output ports. In all ports, the first variable is in the
forward and the second in the backward part.
\end{example}

\begin{figure}[t]
  \begin{center}
 
    \epsfig{file=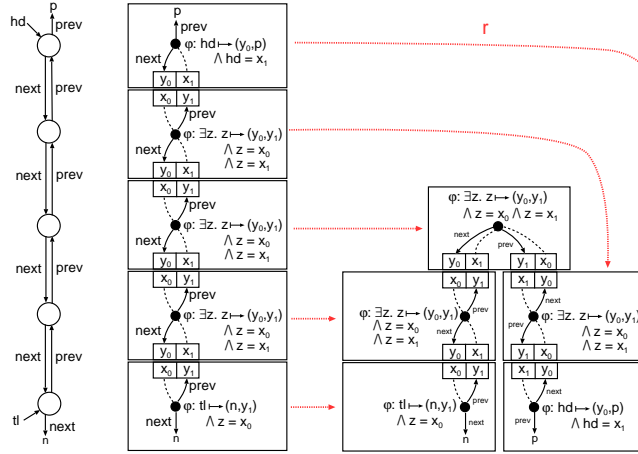, width=90mm}

    \vspace{-2mm}
 
    \caption{An example of a DLL with two of its canonical trees (related by a
    canonical rotation $r$).}

    \vspace{-6mm}
 
    \label{fig:DLLtiles}
 
 \end{center}
\end{figure}

\ifLongVersion
We will now work towards the notion of canonical rotations. For that, we first
give one technical definition. Given a canonical tree $t : \nat^*
\rightharpoonup_{fin} \mathcal{T}^c$, and a state $S = \langle s, h \rangle$,
let $u : dom(t) \rightarrow dom(h)$ be an arbitrary tree labeled with allocated
locations from $S$. For each position $p \in dom(t)$ its {\em explicit
neighbourhood} with respect to $t$ and $u$ is the state $S_{\langle t, u, p
\rangle} = \langle s_p, h_p \rangle$ defined as follows:\begin{itemize}

  \item $h_p = \overline{h}$, and
  
  \item \begin{minipage}{11.5cm}
  \vspace*{-3mm}
  \[s_p(x) = \left\{\begin{array}{ll}
  u(p) & \mbox{if $x \in \port_{-1}^{fw}(t(p)) \cup
    \port_0^{bw}(t(p)) \cup \ldots \cup \port_{\#_t(p)-1}^{bw}(t(p))$} \\ 
  u(p.i) & \mbox{if $x \in \port_i^{fw}(t(p))$, for some $0 \leq
    i < \#_t(p)$} \\ 
  u(p.(-1)) & \mbox{if $p \neq \epsilon$ and $x \in
    \port_{-1}^{bw}(t(p))$} \\ 
  \overline{s}(x) & \mbox{otherwise}
  \end{array}\right.\]
  \end{minipage} \\
  where $S_{\langle u(p) \rangle} = \langle \overline{s}, \overline{h}
  \rangle$ is the neighbourhood of the location $u(p)$ in $S$.

\end{itemize}
\fi

\enlargethispage{6mm}

An important property of canonical trees is that each state that is a
model of the characteristic tile $\Phi(t)$ of a canonical tree $t$
(i.e.\ $S \models \Phi(t)$) can be uniquely described by a local
spanning tree $u : dom(t) \rightarrow Loc$, which has the same
structure as $t$, i.e.\ $dom(u) = dom(t)$.  An immediate consequence
is that any two models of $\Phi(t)$ differ only by a~renaming of the
allocated locations, i.e.\ they are identical up to isomorphism.
\ifLongVersion
\begin{lemma}\label{canonical-model}
Let $t : \nat^* \rightharpoonup_{fin} \mathcal{T}^c$ be a canonical
tree, and let $S = \langle s, h \rangle$ be a state. Then $S \models
\Phi(t)$ iff there exists a local spanning tree $u : dom(t)
\rightarrow dom(h)$ such that, for all $p \in dom(t)$:
\begin{enumerate}
\item\label{cm:1} $\len{\port_i^{fw}(t(p))} = \card{\{k \in Sel ~|~ u(p) 
  \arrow{k}{S} u(p.i)\}}$, for all $0 \leq i < \#_t(p)$.
\item\label{cm:2} $\len{\port_{-1}^{bw}(t(p))} = \card{\{k \in Sel ~|~
  u(p) \arrow{k}{S} u(p.(-1))\}}$ if $p \neq \epsilon$.
\item\label{cm:3} $S_{\langle t, u, p \rangle} \models \form(t(p))$.
\end{enumerate}
\end{lemma}
\begin{proof}``$\Rightarrow$'' By induction on the structure of $t$. For
  the base case $\#_t(\epsilon)=0$, i.e.\ $dom(t)=\{\epsilon\}$, we
  have $t(\epsilon) = \langle \form(t(\epsilon),
  \port_{-1}(t(\epsilon)) \rangle$ and therefore $S \models \Phi(t)$
  if and only if $S \models \form(t(\epsilon))$. Since $t(\epsilon)$
  is a canonical tile, only one location is allocated in $dom(h)$,
  i.e.\ $dom(h) = \{\ell\}$, for some $\ell \in Loc$. We define $u =
  \{(\epsilon, \ell)\}$.  It is immediate to check that $u$ is a local
  spanning tree of $S$, and that $S_{\langle t, u, \epsilon \rangle} =
  S \models \form(t(\epsilon))$ -- points (\ref{cm:1}) and
  (\ref{cm:2}) are vacuously true. 

  For the induction step $\#_t(\epsilon) = d > 0$, we have:
  $$\Phi(t,\epsilon) = t(\epsilon)^\epsilon \circledast_0 \Phi(t,0)
  \ldots \circledast_{d-1} \Phi(t,d-1)$$ Since $S \models
  \Phi(t,\epsilon)$, there exist states $S_i = \langle s_i, h_i
  \rangle$, for each $-1 \leq i < d$, such that $S_{-1} \models
  t(\epsilon)^\epsilon$ and $S_i \models \Phi(t,i)$, for all $0 \leq i
  < d$, and moreover, $S = S_{-1} \uplus S_0 \uplus \ldots \uplus
  S_{d-1}$. By the induction hypothesis, for each $0 \leq i < d$ there
  exists a local spanning tree $u_i : dom(\subtree{t}{i}) \rightarrow
  dom(h_i)$, meeting conditions (\ref{cm:1}), (\ref{cm:2}) and
  (\ref{cm:3}). Since $S_{-1} \models t(\epsilon)^\epsilon$, it must
  be that $dom(h_{-1}) = \{\ell\}$ for some location $\ell \in
  Loc$. We define $u$ as follows:
  \[\begin{array}{rcl}
  u(\epsilon) & = & \ell \\
  u(i.q) & = & u_i(q) ~\mbox{for all $q \in dom(u_i)$}
  \end{array}\]
  To prove that $u$ is a spanning tree of $S$, let $p \in dom(u)$ be a
  position, and $0 \leq i < \#_u(p)$ be a direction. We distinguish
  two cases:
  \begin{itemize}
  \item if $p = \epsilon$, then $\port_i^{fw}(t(p)) \neq \emptyset$,
    hence there exists an edge $\ell \arrow{k}{S_{-1}}
    u_i(\epsilon)$. Then $u(p) \arrow{k}{S} u(p.i)$ as well.
  \item if $p = j.q$, and $q \in dom(u_j)$, for some $0 \leq j < d$,
    by the induction hypothesis, there exists $k \in Sel$, such that
    $u(p) = u_j(q) \arrow{k}{S_j} u_j(q.i) = u(p.i)$, i.e.\ $u(p)
    \arrow{k}{S} u(p.i)$.
  \end{itemize}
  To prove that $u$ is a local spanning tree of $S$, let $\kappa
  \arrow{k}{S} \kappa'$ be an edge, for some $\kappa, \kappa' \in
  dom(h)$. Since $dom(h)$ is a disjoint union of $dom(h_{-1}) =
  \{\ell\}, dom(h_0), \ldots, dom(h_{d-1})$, we distinguish several
  cases:
  \begin{itemize}
  \item if $\kappa \in dom(h_{-1})$ and $\kappa' \in dom(h_i)$, for
    some $0 \leq i < d$, then $u(p)=\kappa=\ell$ and $u(i) =
    u_i(\epsilon) = \kappa'$ is the only possibility -- due to the
    strict semantics of SL, it is not possible to define an edge
    between $\ell$ and an location $\kappa' = u_i(q)$, unless $q =
    \epsilon$.
  \item if $\kappa, \kappa' \in dom(h_i)$, for some $0 \leq i < d$,
    then, by the induction hypothesis, there exists $q \in dom(u_i)$
    and $0 \leq j < \#_{u_i}(q)$ such that $\kappa = u_i(q) = u(i.q)$
    and $\kappa' = u_i(q.j) = u(i.q.j)$.
  \item if $\kappa \in dom(h_i)$ and $\kappa' \in dom(h_j)$, for some
    $0 \leq i,j < d$, then we reach a contradiction with the strict
    semantics of SL -- since there is no equality between output ports
    in $t(\epsilon)$, it is not possible to define an edge between a
    location from $dom(h_i)$ and $dom(h_j)$,
  \end{itemize}
  The proofs of points (\ref{cm:1}), (\ref{cm:2}) and (\ref{cm:3}) are
  by the strict semantics of SL. 

  \noindent ``$\Leftarrow$'' By induction on the structure of $t$. In
  the base case $\#_t(\epsilon) = 0$, i.e.\ $dom(t) = \{\epsilon\}$,
  then $t(\epsilon) = \langle \form(t(\epsilon),
  \port_{-1}(t(\epsilon)) \rangle$. We have $S = S_{\langle
    t,u,\epsilon \rangle} \models \form(t(\epsilon))$, by point
  (\ref{cm:3}), hence $S \models \Phi(t)$.

  For the induction step $\#_t(\epsilon) = d > 0$, let $u_i =
  \subtree{u}{i}$, for all $0 \leq i < d$. Since $u$ is a bijective
  function, we have $img(u_i) \cap img(u_j) = \emptyset$, for all $0
  \leq i < j < d$. For all $-1 \leq i < d$, we define $S_i = \langle
  s_i, h_i \rangle$, where:
  \[\begin{array}{rclr}
  h_i(\ell) & = & \left\{\begin{array}{ll}
  h(\ell) & \mbox{if $\ell \in img(u_i)$} \\
  \bot & \mbox{otherwise}
  \end{array}\right. & ~\mbox{for all $\ell \in Loc$, if $i \geq 0$} \\
  h_{-1}(\ell) & = & \left\{\begin{array}{ll}
  h(\ell) & \mbox{if $u(\epsilon) = \ell$} \\
  \bot & \mbox{otherwise}
  \end{array}\right. & ~\mbox{for all $\ell \in Loc$} \\
  s_i(x) & = & \left\{\begin{array}{ll}
  s(x) & \mbox{if $s(x) \in Img(h_i)$} \\
  \bot & \mbox{otherwise}
  \end{array}\right. & ~\mbox{for all $x \in Var$}
  \end{array}\]
  It is not hard to show that $S = S_{-1} \uplus S_0 \uplus \ldots
  S_{d-1}$. Since $u$ is a local spanning tree of $S$, $u_i$ is a
  local spanning tree for $S_i$, for all $0 \leq i < d$, and
  conditions (\ref{cm:1}), (\ref{cm:2}) and (\ref{cm:3}) hold for
  $S_i$ and $u_i$, respectively. By the induction hypothesis, we have
  $S_i \models \Phi(\subtree{t}{i}, \epsilon)$. By points
  (\ref{cm:1}), (\ref{cm:2}) and (\ref{cm:3}) moreover, we have that
  $S_{-1} \models \form(t(\epsilon))$. Hence $S \models \Phi(t)$. 
\qed\end{proof}
\fi
The following definition is a refinement of Def. \ref{Rotation}. The change in
the structure of the tree is mirrored by a change in the structure of the
canonical tiles labeling the tree.

\begin{definition}\label{Canonical-rotation}
  Given two canonical trees $t,u : \nat^* \rightharpoonup_{fin}
  \mathcal{T}^c$, we say that $u$ is a {\em canonical rotation} of
  $t$, denoted $t \sim_r^c u$, if and only if $r : dom(t) \rightarrow
  dom(u)$ is a bijective function, and for all $p \in dom(t)$, there
  exists a substitution $\sigma_p : Var \rightharpoonup_{fin} Var$
  such that $\form(t(p))[\sigma_p] \equiv \form(u(r(p)))$, and for all
  $0 \leq i < \#_t(p)$, there exists $j \in \mathcal{D}(u)$ such that
  $r(p.i) = r(p).j$, and:\eject  
  \vspace*{-10mm}\[\begin{array}{rcl}
  \port_i^{fw}(t(p))[\sigma_p] & \equiv & \mathbf{if}~ j \geq 0 ~\mathbf{then}~
  \port_j^{fw}(u(r(p))) ~\mathbf{else}~ \port_{-1}^{bw}(u(r(p))) \\
  \port_i^{bw}(t(p))[\sigma_p] & \equiv & \mathbf{if}~ j \geq 0 ~\mathbf{then}~
  \port_j^{bw}(u(r(p))) ~\mathbf{else}~ \port_{-1}^{fw}(u(r(p))) 
  \end{array}\]



\end{definition}

%
%

\begin{example}[cont. of Ex.~\ref{ex:DLLtree}]\label{ex:DLLrot}The
notion of canonical rotation is illustrated by the canonical rotation $r$
relating the two canonical trees of a DLL shown in Fig.~\ref{fig:DLLtiles}. In
its case, the variable substitutions are simply the identity in each node. Note,
in particular, that when the tile $0$ of the left tree (i.e., the second one
from the top) gets rotated to the tile $1$ of the right tree (i.e., the right
successor of the root), the input and output ports get swapped and so do their
forward and backward parts.
\end{example}

The following lemma is the key for proving completeness of our
entailment checking for local inductive systems: if a (local) state is
a model of the characteristic tiles of two different canonical tree,
then these trees must be related by canonical rotation.

\begin{lemma}\label{canonical-rotation-lemma}
  Let $t : \nat^* \rightharpoonup_{fin} \mathcal{T}^c$ be a canonical
  tree and $S = \langle s, h \rangle$ be a state such that $S \models
  \Phi(t)$. Then, for any canonical tree $u : \nat^*
  \rightharpoonup_{fin} \mathcal{T}^c$, we have $S \models \Phi(u)$ iff
  $t \sim^c u$.
\end{lemma}

\ifLongVersion
\begin{proof}If $S \models \Phi(t)$, there exists a local spanning tree $t_s :
dom(t) \rightarrow dom(h)$ of $S$, that meets the three conditions of Lemma
\ref{canonical-model}. 

\noindent``$\Rightarrow$'' If $S \models \Phi(u)$ there exists a local
spanning tree $u_s : dom(u) \rightarrow dom(h)$ of $S$, that meets the
three conditions of Lemma \ref{canonical-model}. Since $t_s$ and $u_s$
are spanning trees of $S$, by Lemma \ref{spanning-tree-rotation}, we
have $t_s \sim_{u_s^{-1} \circ t_s} u_s$. Let $r = u_s^{-1} \circ t_s$
from now on. Clearly, for any $p \in dom(t)$, we have $t_s(p) =
u_s(r(p))$. Since $dom(t)=dom(t_s)$ and $dom(u)=dom(u_s)$, we have $t
\sim_r u$ as well. 

We further need to prove the three points of
Def. \ref{Canonical-rotation} in order to show that $t \sim_r^c
u$. Let $p \in dom(t)$ be an arbitrary location. For the first two
points, let $0 \leq i < \#_t(p)$ be a direction. We have $t_s(p.i) =
u_s(r(p.i)) = u_s(r(p).j)$, for some $j \in \mathcal{D}(u)$. We
distinguish two cases:
\begin{itemize}
\item if $j \geq 0$, we compute:
  \[\begin{array}{rcl}
  \len{\port_i^{fw}(t(p))} & = & \card{\{k \in Sel ~|~ t_s(p) \arrow{k}{S} t_s(p.i)\}} \\
  & = & \card{\{k \in Sel ~|~ u_s(r(p)) \arrow{k}{S} u_s(r(p).j)\}} \\
  & = & \len{\port_j^{fw}(u(r(p)))} \\
  \\
  \len{\port_i^{bw}(t(p))} & = & \len{\port_{-1}^{bw}(t(p.i))} ~\mbox{(by Def. \ref{Canonically-tiled})} \\
    & = & \card{\{k \in Sel ~|~ t_s(p.i) \arrow{k}{S} t_s(p)\}} \\
    & = & \card{\{k \in Sel ~|~ u_s(r(p).j) \arrow{k}{S} u_s(r(p))\}} \\
    & = & \len{\port_j^{bw}(u(r(p)))}
  \end{array}\]
\item if $j = -1$, we compute:
  \[\begin{array}{rcl}
  \len{\port_i^{fw}(t(p))} & = & \card{\{k \in Sel ~|~ u_s(r(p)) \arrow{k}{S} u_s(r(p).(-1))\}} \\
  & = & \len{\port_{-1}^{bw}(u(r(p)))} \\
  \\
  \len{\port_i^{bw}(t(p))} & = & \card{\{k \in Sel ~|~ u_s(r(p).(-1)) \arrow{k}{S} u_s(r(p))\}} \\
  & = & \card{\{k \in Sel ~|~ u_s(r(p).(-1)) \arrow{k}{S} u_s(r(p).(-1).j')\}} \\
  && ~\mbox{for some $j' \geq 0$ such that $r(p)=r(p).(-1).j'$} \\
  & = & \len{\port_{j'}^{fw}(u(r(p).(-1)))} \\
  & = & \len{\port_{-1}^{fw}(u(r(p)))} ~\mbox{(by Def. \ref{Canonically-tiled})}
  \end{array}\]
\end{itemize}
Since all variables are pairwise distinct in $\port_i(t(p))$ and
$\port_j(u(r(p)))$, respectively, one can define a substitution
$\sigma_p$ meeting the conditions of the first two points of
Def. \ref{Canonical-rotation}.

For the third point of Def. \ref{Canonical-rotation}, by Lemma
\ref{canonical-model}, we have that $S_{\langle t, t_s, p \rangle}
\models \form(t(p))$ and $S_{\langle u, u_s, r(p) \rangle} \models
\form(u(r(p)))$, where $S_{\langle t, t_s, p \rangle} = \langle s_t,
h_t \rangle$ is the explicit neighbourhood of $p$ w.r.t $t$ and $t_s$,
and $S_{\langle u, u_s, p \rangle} = \langle s_u, h_u \rangle$ is the
explicit neighbourhood of $r(p)$ w.r.t $u$ and $u_s$. Since $t(p)$ and
$u(r(p))$ are canonical tiles, we have:
\[\begin{array}{rcl}
\form(t(p)) & \equiv & (\exists z)~ z \mapsto (y_0, \ldots, y_{m-1}) \wedge \Pi_t \\
\form(u(r(p))) & \equiv & (\exists w)~ w \mapsto (v_0, \ldots, v_{n-1}) \wedge \Pi_u
\end{array}\]
where 
\[\begin{array}{rcl}
\Pi_t & \equiv & \port_{-1}^{fw}(t(p)) = z ~\wedge~ 
\bigwedge_{i=0}^{\#_t(p)-1}\port_i^{bw}(t(p)) = z \\ 
\Pi_u & \equiv & \port_{-1}^{fw}(u(r(p))) = w ~\wedge~ 
\bigwedge_{i=0}^{\#_u(r(p))-1}\port_i^{bw}(u(r(p))) = w
\end{array}\]
We extend the substitution $\sigma_p$ to $\sigma_p[z \leftarrow w]$ in
the following. It is not hard to check that $s_t \circ \sigma_p = s_u$
and $h_t = h_u$. Hence $S_{\langle u, u_s, p \rangle} \models
\form(t(p))[\sigma_p]$. Since both $\form(t(p))[\sigma_p]$ and
$\form(u(r(p)))$ have the same model, by the definition of the
(strict) semantics of SL, it follows that the numbers of edges are the
same, i.e.\ $m=n$, and either both $z$ and $w$ are quantified, or they
are free. Thus, we obtain that $\form(t(p))[\sigma_p] \equiv
\form(u(r(p)))$.

\noindent''$\Leftarrow$'' It $t \sim^c u$, there exists a bijective
function $r : dom(t) \rightarrow dom(u)$, meeting the conditions of
Def. \ref{Canonical-rotation}. Let $u_s = t_s \circ r^{-1}$ be a
tree. Since $t_s$ and $r$ are bijective, then also $u_s$ is bijective,
and $dom(u_s) = dom(u)$. To prove that $S \models \Phi(u)$, it is
enough to show that $u_s$ is a local spanning tree of $S$, meeting the
three conditions of Lemma \ref{canonical-model}.

To show that $u_s$ is a local spanning tree of $S$, let $p \in
dom(u_s)$ and $i \in \mathcal{D}_+(u_s)$ such that $p.i \in
dom(u_s)$. Since $t_s \sim_{r} u_s$, by
Prop. \ref{Rotation:equivalence}, we have $u_s \sim_{r^{-1}} t_s$,
hence there exists $j \in \mathcal{D}(t_s)$ such that $r^{-1}(p.i) =
r^{-1}(p).j$. Since $t_s$ is a spanning tree of $S$, there exists a
selector $k \in Sel$ such that: $$u_s(p) = t_s(r^{-1}(p)) \arrow{k}{S}
t_s(r^{-1}(p).j) = t_s(r^{-1}(p.i)) = u_s(p.i)$$ Hence $u_s$ is a
spanning tree of $S$. In order to prove its locality, let $\ell
\arrow{k}{} \ell'$ be an edge of $S$, for some $\ell, \ell \in loc(S)$
and $k \in Sel$. Since $t_s$ is a local spanning tree of $S$, there
exist $p \in dom(t_s)$ and $i \in \mathcal{D}_+(t_s)$ such that $\ell
= t_s(p) = u_s(r(p))$ and $\ell' = t_s(p.i) = u_s(r(p.i))$. Since $t_s
\sim_{r} u_s$, there exists $j \in \mathcal{D}(u_s)$ such that $r(p.i)
= r(p).j$. Hence we have $u_s(r(p)) = \ell \arrow{k}{} \ell' =
u_s(r(p).j)$, i.e.\ the edge is local w.r.t. $u_s$ as well.

To prove point (\ref{cm:1}) of Lemma \ref{canonical-model}, let $p \in
dom(u)$ be an arbitrary position, and $0 \leq i < \#_u(p)$ be a
direction. Since $u_s \sim_{r^{-1}} t_s$, we have $r^{-1}(p.i) =
r^{-1}(p).j$, for some $j \in \mathcal{D}(t)$. We compute:
\[\begin{array}{rcl}
\card{\{k \in Sel ~|~ u_s(p) \arrow{k}{S} u_s(p.i)\}} & = & 
\card{\{k \in Sel ~|~ t_s(r^{-1}(p)) \arrow{k}{S} t_s(r^{-1}(p.i) \}} \\
& = & \card{\{k \in Sel ~|~ t_s(r^{-1}(p)) \arrow{k}{S} t_s(r^{-1}(p).j) \}}
\end{array}\]
We distinguish two cases:
\begin{itemize}
\item if $j \geq 0$, we have, by Lemma \ref{canonical-model} applied
  to $t$ and $t_s$:
\[\begin{array}{rcl}
\card{\{k \in Sel ~|~ t_s(r^{-1}(p)) \arrow{k}{S} t_s(r^{-1}(p).j) \}}
& = & \len{\port_j^{fw}(t(r^{-1}(p)))} \\ 
& = & \len{\port_i^{fw}(u(p))}
\end{array}\]
The last equality is because $r(r^{-1}(p).j) = r(r^{-1}(p)).i$. 
\item if $j = -1$, there exists $d \in \mathcal{D}_+(t)$ such that
  $r^{-1}(p) = (r^{-1}(p).(-1)).d$. But then we have:
\[\begin{array}{rcl}
r((r^{-1}(p).(-1)).d) & = & r(r^{-1}(p)) = p = (p.i).(-1) \\
& = & r(r^{-1}(p).(-1)).(-1)
\end{array}\]
We compute, further:
\[\begin{array}{rcl}
\card{\{k \in Sel ~|~ t_s(r^{-1}(p)) \arrow{k}{S} t_s(r^{-1}(p).j) \}}
& = & \len{\port_{-1}^{bw}(t(r^{-1}(p)))} \\ & = &
\len{\port_d^{bw}(t(r^{-1}(p).(-1)))} ~\mbox{(by Def. \ref{Canonically-tiled})} \\ 
& = & \len{\port_{-1}^{fw}(u(p.i))} ~\mbox{(by Def. \ref{Canonical-rotation})} \\
& = & \len{\port_i^{fw}(u(p))} ~\mbox{(by Def. \ref{Canonically-tiled})}
\end{array}\]
\end{itemize}
For point (\ref{cm:2}) of Lemma \ref{canonical-model}, one applies a
symmetrical argument. To prove point (\ref{cm:3}) of Lemma
\ref{canonical-model}, let $S_{\langle t, t_s, r^{-1}(p) \rangle} =
\langle s_t, h_t \rangle$ be the explicit neighbourhood of $r^{-1}(p)$
w.r.t $t$ and $t_s$. It is not hard to check that $S_{\langle u, u_s,
  p \rangle} = \langle s_t \circ \sigma_p, h_t \rangle$, hence:
$$S_{\langle u, u_s, p \rangle} = \langle s_t \circ \sigma_p, h_t
\rangle \models \form(t(r^{-1}(p)))[\sigma_p] \equiv \form(u(p))$$ The
last equivalence is by Def. \ref{Canonical-rotation}.\qed\end{proof}
\fi

\enlargethispage{6mm}

\vspace*{-3mm}\subsection{Quasi-Canonical Tiles}\vspace*{-1mm}
\label{sec:quasi-canonical-tiles}

This section introduces a generalization of canonical trees, denoted
as \emph{quasi-canonical trees}, that are used to define states with
non-local edges. A singleton tile $T = \langle \varphi \wedge \Pi,
\vec{x}_{-1}, \vec{x}_0, \ldots, \vec{x}_{d-1} \rangle$, for some $d
\geq 0$, is said to be {\em quasi-canonical} if and only if $\vec{x}_i
\equiv \vec{x}_i^{fw} \cdot \vec{x}_i^{bw} \cdot \vec{x}_i^{eq}$, for
each $-1 \leq i < d$, $\Pi$ is pure formula, and all of the following
holds:\vspace*{-1mm}\begin{enumerate}

  \item $\langle \varphi,~ \vec{x}_{-1}^{fw} \cdot \vec{x}_{-1}^{bw},~
  \vec{x}_0^{fw} \cdot \vec{x}_0^{bw},~ \ldots,~ \vec{x}_{d-1}^{fw} \cdot
  \vec{x}_{d-1}^{bw} \rangle$ is a canonical tile.

  \item For each $0 \leq i < \len{\vec{x}_{-1}^{eq}}$, either
  $(\vec{x}_{-1}^{eq})_i \in FV(\varphi)$ or $(\vec{x}_{-1}^{eq})_i =_\Pi
  (\vec{x}^{eq}_k)_j$ for some unique indices $0 \leq k < d$ and $0 \leq j <
  \len{\vec{x}_k^{fw}}$.

  \item For each $0 \leq k < d$ and each $0 \leq j < \len{\vec{x}_k^{eq}}$,
  either $(\vec{x}^{eq}_k)_j \in FV(\varphi)$ or exactly one of the following
  holds: (i) $(\vec{x}^{eq}_k)_j =_\Pi (\vec{x}_{-1}^{eq})_i$ for some unique
  index $0 \leq i < \len{\vec{x}_{-1}^{eq}}$. (ii)~$(\vec{x}^{eq}_k)_j =_\Pi
  (\vec{x}_r^{eq})_s$ for some unique indices $0 \leq r < d$ and $0 \leq s <
  \len{\vec{x}_r^{eq}}$.

  \item For any two variables $x,y \in \vec{x}_{-1}^{eq} ~\cup~
  \bigcup_{i=0}^{d-1} \vec{x}_i^{eq}$, $x =_\Pi y$ only in one of the cases
  above.

\end{enumerate}

\eject 

\ifLongVersion
Notice that in a quasi-canonical tile $T = \langle (\exists z)~ z
\mapsto (y_0, \ldots, y_{m-1}) \wedge \Pi, \vec{x}_{-1},~ \vec{x}_0,~
\ldots, ~\vec{x}_{d-1} \rangle$, the free variables from the set
$\{y_0,\ldots,y_{m-1}\} \setminus \{z,\nil\}$ are allowed to occur in
some tuple $\vec{x}_i^{eq}$, for $-1 \leq i < d$, as well. These
variables can be used to describe non-local edges of states because
equality constraints between $\vec{x}_{-1}^{eq}$ and $\vec{x}_i^{eq}$,
for $0 \leq i < d$, may span several tiles in a tree. In this case,
the factorization of each port $\vec{x}_i$, $-1 \leq i < d$, into
$\vec{x}_i^{fw}$, $\vec{x}_i^{bw}$, and $\vec{x}_i^{eq}$ is not
uniquely determined by the distribution of the referenced variables
$y_0, \ldots, y_{m-1}$. In practice (Sec.\ref{sec:canonization}) we
use a conservative heuristic to define this factorization. 
\fi
We denote in the following $\port_i^{eq}(T) \equiv \vec{x}_i^{eq}$,
for all $-1 \leq i < d$. The set of quasi-canonical tiles is denoted
by $\mathcal{T}^{qc}$.  The next definition of quasi-canonical trees
extends Def.  \ref{Canonically-tiled} to the case of quasi-canonical
tiles.

\begin{definition}\label{Quasi-canonically-tiled}
A tiled tree $t : \nat^* \rightharpoonup_{fin} \mathcal{T}^{qc}$ is {\em
  quasi-canonical} if and only if, for any $p \in dom(t)$ and each $0
\leq i < \#_t(p)$, such that $t(p) = \langle \varphi, \vec{x}_{-1},
\vec{x}_0, \ldots, \vec{x}_{d-1} \rangle$, $t(p.i) = \langle \psi,
\vec{y}_{-1}, \vec{y}_0, \ldots, \vec{y}_{e-1} \rangle$, $\vec{x}_i
\equiv \vec{x}^{fw}_i \cdot \vec{x}^{bw}_i \cdot \vec{x}^{eq}_i$ and
$\vec{y}_{-1} \equiv \vec{y}_{-1}^{fw} \cdot \vec{y}_{-1}^{bw} \cdot
\vec{y}_{-1}^{eq}$, we have $\len{\vec{x}^{fw}_i} =
\len{\vec{y}_{-1}^{fw}}$, $\len{\vec{x}^{bw}_i} =
\len{\vec{y}_{-1}^{bw}}$, and $\len{\vec{x}^{eq}_i} =
\len{\vec{y}_{-1}^{eq}}$.
\end{definition}

\begin{example}[cont. of Ex.~\ref{ex:TLL}]\label{ex:TLLtree}For an
illustration of the notion of canonical trees, see Fig.~\ref{fig:TLLtiles},
which shows a canonical tree for the TLL from Fig.~\ref{fig:exTLL}. The figure
uses the same notation as Fig.~\ref{fig:DLLtiles}. In all the ports, the first
variable is in the forward part, the backward part is empty, and the rest is the
equality part.
\end{example}

\begin{figure}[t]
  \begin{center}
 
    \epsfig{file=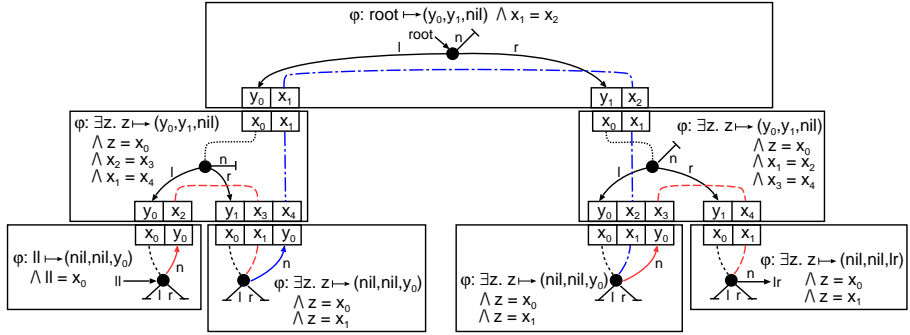, width=120mm}

    \vspace{-2mm}
 
    \caption{An example of a quasi-canonically tiled tree for the tree with
    linked leaves from Fig.~\ref{fig:exTLL}.}

    \vspace{-6mm}
 
    \label{fig:TLLtiles}
 
 \end{center}
\end{figure}

\ifLongVersion
For a quasi-canonical tree $t : \nat^* \rightharpoonup_{fin}
\mathcal{T}^{qc}$, we define its {\em canonical projection} as $t^c :
dom(t) \rightarrow \mathcal{T}^c$, such that, for all $p \in dom(t)$,
we have:
\[\begin{array}{rcl}
t^c(p) & = & \langle \psi, \port_{-1}^{fw}(t(p)) \cdot
\port_{-1}^{bw}(t(p)), \ldots, \port_{\#_t(p)-1}^{fw}(t(p)) \cdot
\port_{\#_t(p)-1}^{bw}(t(p)) \rangle \\
&& \mbox{where}~ \psi \iff \exists \port_{-1}^{eq}(t(p)) \ldots \exists
\port_{\#_t(p)-1}^{eq}(t(p)) ~.~ \form(t(p))
\end{array}\]
  Intuitively, the canonical projection of a tree removes all equality variables
  from every tile in the tree. Since these variables are involved only in
  equality constraints, the above quantifiers $\exists \port_i^{eq}(t(p))$ can
  be eliminated. We can now define the notion of quasi-canonical rotation for
  quasi-canonical trees as an extension of Def. \ref{Canonical-rotation}.
\fi

\begin{definition}\label{Quasi-canonical-rotation} Given two quasi-canonical
trees $t,u : \nat^* \rightharpoonup_{fin} \mathcal{T}^{qc}$, we say that $u$ is
a~\emph{quasi-canonical rotation} of $t$, denoted $t \sim_r^{qc} u$, if and only
if $r : dom(t) \rightarrow dom(u)$ is a~bijective function such that $t^c
\sim^c_r u^c$ and, for all $p \in dom(t)$, there exists a substitution $\sigma_p
: Var \rightharpoonup_{fin} Var$ such that, for all $0 \leq i < \#_t(p)$ and all
$-1 \leq j < \#_t(p)$, where $r(p.i) = r(p).j$, we have
$\port_i^{eq}(t(p))[\sigma] \equiv \port_j^{eq}(u(r(p)))$.\end{definition}

The increase in expressivity (i.e.\ the possibility of defining non-local edges)
comes at the cost of a loss of completeness. The following lemma generalizes the
necessity direction ($\Leftarrow$) of Lemma \ref{canonical-rotation-lemma} for
quasi-canonical tiles. Notice that the sufficiency ($\Rightarrow$) direction
does not hold in general.

\begin{lemma}\label{quasi-canonical-rotation-lemma} Let $t,u : \nat^*
\rightharpoonup_{fin} \mathcal{T}^{qc}$ be quasi-canonical trees such that $t
\sim^{qc} u$. For all states $S$, if $S \models \Phi(t)$, then $S \models
\Phi(u)$.\end{lemma}
\ifLongVersion
\begin{proof} Let $S$ be a state such that $S \models \Phi(t)$, and $r : dom(t)
\rightarrow dom(u)$ be a bijective function such that $t \sim^{qc}_r u$. Then $S
\models \Phi(t^c)$, and since we have $t^c \sim^c_r u^c$ by Def.
\ref{Quasi-canonical-rotation}, then also $S \models \Phi(u^c)$, by Lemma
\ref{canonical-rotation-lemma}. Also, for every $p \in dom(t)$, let $\sigma_p :
Var \rightharpoonup_{fin} Var$ be the substitution from Def.
\ref{Canonical-rotation}.

Assume that $S \not\models \Phi(u)$. Then there exists a non-trivial
path $p_0, \ldots, p_k \in dom(u)$, for some $k > 0$, and some
variables $x_0 \in FV(\form(u(p_0))), \ldots, x_k \in
FV(\form(u(p_k)))$, such that $x_0$ and $x_k$ are allocated in
$\form(u(p_0))$ and $\form(u(p_k))$, respectively, and for all $0 \leq
i < k$, we have $x_i =_{\Pi_u} x_{i+1}$, where $\Pi_u$ is the pure
part of $\form(\Phi(u))$. But then $r^{-1}(p_0), \ldots, r^{-1}(p_k)
\in dom(t)$ is a path in $t$, and the variables
$\sigma_{p_0}^{-1}(x_0)$ and $\sigma_{p_k}^{-1}(x_k)$ are allocated in
$t(r^{-1}(p_0))$ and $t(r^{-1}(p_k))$, respectively. Moreover, we have
$\sigma_{p_i}^{-1}(x_i) =_{\Pi_t} \sigma_{p_{i+1}}^{-1}(x_{i+1})$,
where $\Pi_t$ is the pure part of $\form(\Phi(t))$. Hence $S
\not\models \Phi(t)$, contradiction.\qed\end{proof}
\fi

\enlargethispage{3mm}

\vspace*{-2.5mm}\section{From Inductive Definitions to Tree
Automata}\vspace*{-2mm}\label{sec:FromDefToTA}

This section describes the reduction of an entailment problem, for
given inductive systems, to a language inclusion problem between two
tree automata.  This reduction is the basis of our method for
entailments checking. The complexity of the reduction algorithm is
polynomial in the size of the input system.

\vspace*{-2.5mm}\subsection{Canonization}\vspace*{-1mm}\label{sec:canonization}

Given a rooted inductive system $\langle \mathcal{P}, P_i \rangle$, the
canonization procedure described in this section will produce an equivalent
rooted inductive system $\langle \mathcal{Q}, Q_j \rangle$, such that each
unfolding tree $t \in \mathcal{T}_j(\mathcal{Q})$ is canonical. For the latter
inductive system, we define a~TA recognizing its set of unfolding trees. For the
rest of this section, let $\mathcal{P} = \{P_i \equiv |_{j=1}^{m_i}
R_{i,j}\}_{i=1}^n$ be an inductive system of the form
(\ref{recursive-definitions}), where each rule is of the form $R_{i,j}(\vec{x})
\equiv \exists \vec{z} ~.~ \Sigma * P_{i_1}(\vec{y}_1) * \ldots *
P_{i_m}(\vec{y}_m) \wedge \Pi$, $\vec{y}_1 \cup \ldots \cup \vec{y}_m \cup
FV(\Pi) \subseteq \vec{x} \cup \vec{z}$. We recall that $\Sigma$ is a non-empty
spatial formula, and the pure formula $\Pi$ can only specify equalities between
(i) two existentially quantified variables from $\vec{z}$, (ii) an existentially
quantified variable from $\vec{z}$ and a formal parameter from $\vec{x}$, or
(iii) between an allocated parameter from $\vec{x}$ and some other formal
parameter in $\vec{x}$. Notice that equality between two different allocated
variables leads to unsatisfiable rules, therefore it will cause a pre-processing
error.

\vspace*{-4mm}\subsubsection{Elimination of equalities.}

The first pre-processing step consists in eliminating equalities
involving existentially quantified variables from the pure part $\Pi$
of a rule. We use the equivalence relation $=_\Pi$ induced by $\Pi$ on
the variables in the rule. For each variable $x \in \vec{x} \cup
\vec{z}$, its equivalence class $[x]_{\Pi}$ either (i) contains at
most one formal parameter, or (ii) contains one or more formal
parameters, out of which exactly one is allocated. We consider this
parameter to be the representative of the equivalence class, and
replace each occurrence of a variable in a rule by its
representative. The result is an equivalent system of the same size as
the original. 

\ifLongVersion
Let $\alloc(\Sigma)$ denote the set of allocated variables in
$\Sigma$. For each equivalence class $[x]_{\Pi}$, where $x \in \vec{x}
\cup \vec{z}$, we define its {\em representative} to be either one of
the following:
\begin{itemize}

\item the unique formal parameter $x \in [x]_{\Pi}$, if $[x]_{\Pi}
  \cap \vec{x} \neq \emptyset$ and $[x]_{\Pi} \cap \vec{x} \cap
  \alloc(\Sigma) = \emptyset$,

\item the unique allocated parameter $x \in [x]_{\Pi} \cap \vec{x}
  \cap \alloc(\Sigma)$, if $[x]_{\Pi} \cap \vec{x} \cap \alloc(\Sigma)
  \neq \emptyset$,

\item the lexicographically minimal element $minlex([x]_{\Pi})$ of
  $[x]_{\Pi}$, if $[x]_{\Pi} \cap \vec{x} = \emptyset$.

\end{itemize} 
For a tuple of variables $\vec{y} = \langle y_1, \ldots, y_\ell
\rangle$, let $[\vec{y}]_\Pi$ denote the tuple of representatives of
the equivalence classes $[y_1]_\Pi, \ldots, [y_\ell]_\Pi$,
respectively. Then we define the rule: $$R^=_{i,j}(\vec{x}) \equiv
\Sigma^= * P_{i_1}([\vec{y}_1]_\Pi) * \ldots *
P_{i_m}([\vec{y}_m]_\Pi) \wedge \Pi^=$$ where $\Sigma^=$ is obtained
from $\Sigma$ by replacing each free variable $x \in FV(\Sigma)$ by
the representative of $[x]_\Pi$, and $\Pi^= \equiv \bigwedge \{ x = y
~|~ x,y \in \vec{x},~ x \in [y]_{\Pi} \cap \alloc(\Sigma)\}$ keeps
only the equalities between formal parameters, one of which is
allocated. It is not hard to check that any rooted system $\langle
\mathcal{P}[R^=_{i,j}/R_{i,j}], P_k \rangle$, obtained by replacing
$R_{i,j}$ with $R^=_{i,j}$, is equivalent to the original rooted
system $\langle \mathcal{P}, P_k \rangle$, for all $k = 1, \ldots, n$.
Moreover, the two systems have the same size.
\fi

\vspace*{-4mm}\subsubsection{Reduction to one points-to proposition per rule.}

The second step of the canonization procedure builds an inductive
system in which the head of each rule consists of exactly one
points-to atomic proposition of the form $x \mapsto (y_0, \ldots,
y_\ell)$. This step succeeds under the assumption that $\mathcal{P}$
is a connected system, and no rule has an empty head (otherwise the
system breaks the restrictions that we have introduced, and the
computation is aborted). The result is an equivalent connected system
$\mathcal{Q}$, whose size is increased by at most a linear factor,
i.e.\ $\len{\mathcal{Q}} = \mathcal{O}(\len{\mathcal{P}})$. Algorithm
\ref{alg:split} 
describes the splitting of rules into sets of rules with exactly one points-to
proposition.

\ifLongVersion
\begin{algorithm}[t]
\begin{algorithmic}[0]
  \State {\bf input} An inductive system $\mathcal{P} = \{P_i \equiv 
    |_{j=1}^{m_i} R_{i,j}\}_{i=1}^n$ 
  \State {\bf output} An inductive system $\mathcal{Q}$ with one points-to proposition 
    per rule
\end{algorithmic}
\begin{algorithmic}[1]
  \Function{splitSystem}{$\mathcal{P}$}
  \State{$\mathcal{Q} \leftarrow \emptyset$}
  \ForAll{$i = 1, \ldots, n$} \Comment{iterate over all predicates}
  \State{$\overline{P}_i \leftarrow \mathbf{empty\_predicate}$}
  \ForAll{$j = 1, \ldots, m_i$} \Comment{iterate over all rules of $P_i$}
  \State {\bf assume} $R_{i,j}(\vec{x}) \equiv \exists \vec{z} ~.~ \Sigma 
  * P_{i_1}(\vec{y}_1)   * \ldots * P_{i_k}(\vec{y}_k)$
  \State {\bf choose} $root \in \vec{x} \cap FV(\Sigma)$
  \State{$t \leftarrow \Call{DepthFirstTraverse}{\Sigma,root}$}\label{line:dfs}
  \If{$dom(t) = \emptyset$}
  \State {\bf error}(``empty rule $R_{i,j}$'')
  \EndIf
  \If{$\bigcup_{p \in dom(t)}t(p) \neq AP(\Sigma)$}
  \State {\bf error}(``disconnected rule $R_{i,j}$'')
  \EndIf
  \ForAll{$s = 1, \ldots, k$}
  \State{$X_s \leftarrow \{p \in dom(t) ~|~ 
    t(p) \equiv y \mapsto (z_1, \ldots, z_\ell),~ 
    \vec{y}_s \cap \{z_1, \ldots, z_\ell\} \neq \emptyset\}$}
  \If{$X_s = \emptyset$}
  \State {\bf error}(``disconnected rule $R_{i,j}$'')
  \EndIf
  \EndFor
  \ForAll{$p \in dom(t)$} \Comment{create fresh predicates}
  \State\label{line:minlex}{$P^p_{i,j} = \{R^p_{i,j} \equiv ~[p=\epsilon ~?~ \exists 
    \vec{z}]~ t(p) ~*~ \bigstar_{\scriptstyle{0 \leq d < \#t(p)}} P^{p.d}_{i,j}(\vec{x},
    \vec{z}) ~*~ 
    \bigstar_{\hspace*{-2mm}
      \begin{array}{c}
        \vspace*{-1mm}
        \scriptstyle{s = 1,\ldots,k} \\
        \scriptstyle{minlex(X_s) = p} 
        \vspace*{-2mm}
    \end{array}
    \hspace*{-2mm}} P_{i_s}(\vec{y}_s)\}$}
  \EndFor
  \State{$\overline{P}_i \leftarrow \overline{P}_i ~|~ R^\epsilon_{i,j}$} \Comment{create 
    root-level predicates}
  \State{$\mathcal{Q} \leftarrow \mathcal{Q} \cup \{P^p_{i,j} ~|~ p \in dom(t) \setminus 
    \{\epsilon\}\}$}
  \EndFor
  \State{$\mathcal{Q} \leftarrow \mathcal{Q} \cup \{\overline{P}_i\}$} 
    \label{line:initial-q}
  \EndFor
  \State {\bf return} $\Call{Cleanup}{\mathcal{Q}}$ \label{line:cleanup}
  \EndFunction
\end{algorithmic}
\caption{Reduction to one points-to proposition per
  rule} \label{alg:split}
\end{algorithm}

For each rule $R_{i,j}$, with $\head(R_{i,j}) \equiv \Sigma$, the procedure
\textsc{DepthFirstTraverse}$(\Sigma,x)$ (line \ref{line:dfs}) performs a
depth-first traversal of the spatial formula $\Sigma \equiv \bigstar_{i=1}^sx_i
\mapsto (y_{i,1}, \ldots,$ $y_{i,m_i})$ starting with a randomly
chosen free variable $x \in FV(\Sigma)$, and builds an injective depth-first
spanning tree $t : \nat^* \rightharpoonup_{fin} AP(\Sigma)$ of the
formula\footnote{It must be possible to build the spanning tree from any free
variable $x \in FV(\Sigma)$ unless the system is disconnected (in which case an
error is announced and the computation aborted).}.
Each position $p \in dom(t)$ is labeled by one
points-to atomic proposition from $\Sigma$, and each atomic proposition from
$\Sigma$ is found in the tree, i.e.\ $AP(\Sigma) = \bigcup_{p \in dom(t)} t(p)$.
Formally, for all positions $p \in dom(t)$, such that $t(p) \equiv x_i \mapsto
(y_{i,1}, \ldots, y_{i,m_i})$, we have:\begin{itemize}

  \item For all $0 \leq d < \#_t(p)$, $\alloc(t(p.d)) = \{y_{i,j}\}$, for some
  $1 \leq j \leq m_i$, i.e.\ the children of each position correspond to
  points-to formulae that allocate variables pointed to by the proposition of
  that position.

  \item For all $0 \leq d < e < \#_t(p)$, if $\alloc(t(p.d)) = \{y_{i,j}\}$ and
  $\alloc(t(p.e)) = \{y_{i,k}\}$, then $j < k$, i.e. the children of each node
  are ordered with respect to the selectors via which they are pointed to.

\end{itemize} The spanning tree $t$ is used to create a set of fresh predicates
$P^p_{i,j} \equiv R_{i,j}^p$, one for each position $p \in dom(t) \setminus
\{\epsilon\}$, and top rules $\overline{P}_i \equiv |_{j=1}^{m_i}
R_{i,j}^\epsilon$, which are the only rules in which existential
quantification is allowed. For each predicate occurrence $P_{i_1}(\vec{y}_1),
\ldots, P_{i_k}(\vec{y}_k)$, the sets $X_1, \ldots, X_k \subseteq dom(t)$
correspond to the positions $p$ where the actual parameters are referred to by
$t(p) \equiv y \mapsto (z_1, \ldots, z_\ell)$. We chose the lexicographically
minimal position from each set $X_1,\ldots,X_k$ (line \ref{line:minlex}) to
place the occurrences of $P_{i_1}, \ldots, P_{i_k}$, respectively.

Finally, the new inductive system $\mathcal{Q}$ is cleaned (line
\ref{line:cleanup}) by removing (i) all unused variables from the rules and from
the calls to the predicates in which they are declared, and (ii) moving
existential quantifiers inside the rules where they are used. For instance, in
the example below, the existential quantifier $\exists z$ has been moved from
$R_1$ (left) to $R_2$ (right), because $z$ is used in the points-to formula of
$R_2$:
\[\begin{array}{ccc}
\left[\begin{array}{rcl}
R_1(x) & \equiv & \exists y,z ~.~ x \mapsto y * R_2(y,z) \\
R_2(y,z) & \equiv & y \mapsto z
\end{array}\right]
& ~\arrow{\Call{Cleanup}{}}{}~ &
\left[\begin{array}{rcl}
R_1(x) & \equiv & \exists y ~.~ x \mapsto y * R_2(y) \\
R_2(y) & \equiv & \exists z ~.~ y \mapsto z
\end{array}\right]
\end{array}\]
The elimination of useless variables is done in reversed topological order,
always processing a predicate $P_i$ before $P_j$ only if $P_i$ occurs in a rule
of $P_j$, whereas the elimination of existential quantifiers is performed in
topological order, i.e.\ we process $P_i$ before $P_j$ only if $P_j$ occurs in a
rule of $P_i$. 

It can be easily checked that $\Call{DepthFirstTraverse}{\Sigma,root}$
takes time $\mathcal{O}(\len{\Sigma})$, and
$\Call{Cleanup}{\mathcal{Q}}$ takes time
$\mathcal{O}(\len{Q})$. Moreover, since the inductive system
$\mathcal{Q}$ is obtained (line \ref{line:initial-q}) in time
$\mathcal{O}(\len{\mathcal{P}})$, it must be the case that
$\len{\mathcal{Q}}=\mathcal{O}(\len{\mathcal{P}})$. Thus, the entire
Algorithm \ref{alg:split} takes time $\mathcal{O}(\len{\mathcal{P}})$.
It is not hard to check that the result of Algorithm \ref{alg:split}
is an inductive system which is equivalent to the input,
i.e.\ $\langle \mathcal{P}, P_i \rangle$ and $\langle
\Call{splitSystem}{\mathcal{P}}, \overline{P}_i \rangle$ are
equivalent, for all $i = 1, \ldots, n$.
\fi

\enlargethispage{4mm}

\vspace*{-4mm}\subsubsection{Parameter elimination.} 

The final pre-processing step, before conversion of a rooted inductive
system $\langle \mathcal{P}, P_i \rangle$ to a tree automaton, is the
specialisation of $\mathcal{P}$ with respect to the predicate
$P_i(x_{i,1}, \ldots, x_{i,n_i})$, and a tuple of actual parameters
$\overline{\alpha} = \langle \alpha_1, \ldots, \alpha_{n_i} \rangle$,
not occurring in the system. We say that a formal parameter $x_{i,k}$
of a rule $R_{i,j}(x_{i,1}, \ldots, x_{i,n_i}) \equiv \exists \vec{z}
~.~ \Sigma * P_{i_1}(\vec{y}_1) * \ldots * P_{i_m}(\vec{y}_m) \wedge
\Pi$ is {\em directly propagated} to some (unique) parameter of a
predicate occurrence $P_{i_j}$, for some $1 \leq j \leq m$, if and
only if $x_{i,k} \not\in FV(\Sigma)$ and $x_{i,k} \equiv
(\vec{y}_{i_j})_\ell$, for some $0 \leq \ell < \len{\vec{y}_{i_j}}$,
i.e.\ $x_{i,k}$ is neither allocated nor pointed to by the head of the
rule, before being passed on to $P_{i_j}$. We denote direct
propagation of parameters by the relation $x_{i,k} \leadsto
x_{i_j,\ell}$, where $x_{i_j,\ell}$ is the formal parameter of
$P_{i_j}$ which is mapped to the occurrence of
$(\vec{y}_{i_j})_\ell$. We say that $x_{i,k}$ is {\em propagated} to
$x_{r,s}$ if $X_{i,k} \leadsto^* x_{r,s}$, where $\leadsto^*$ denotes
the reflexive and transitive closure of the $\leadsto$
relation. Algorithm \ref{alg:eliminate} describes the elimination of all variables
that are propagated from the formal parameter tuple of a given
predicate.
\ifLongVersion
Notice that all eliminated variables are collected in a global set
$\mathtt{Parameters}$, which will be used later, by Algorithm
\ref{alg:sl2ta} to produce a tree automaton.

The tracking/elimination of a given formal parameter is implemented by
a recursive function $\Call{trackEliminate}{P_r,x_r,\mathtt{del}}$,
where $P_r$ is the current predicate, $x_{r,s}$ is the currently
tracked formal parameter of $P_r$ (i.e.\ $\mathtt{Tracked} \leadsto^*
x_{r,s}$ is an invariant for every call to $\Call{trackEliminate}{}$),
and $\mathtt{del}=\true$ if and only if $x_{r,s}$ is to be removed
from the definition of $P_r$. If the currently tracked parameter
$x_{r,s}$ is either allocated in a rule of $P_r$ or it is not
propagated further, then every occurrence of $x_{r,s}$ is replaced
with $\mathtt{Tracked}$ (line \ref{line:tracked}). Otherwise, if the
parameter is propagated to $P_j$ as $x_{j,\ell}$ (line
\ref{line:prop}) and it is referenced by the current rule (line
\ref{line:ref}), then it will not be removed any longer from the
system (line \ref{line:del:false}). In this case, we keep tracking it
only to place the global variable $\mathtt{Tracked}$ in the right
place (line \ref{line:tracked}). Otherwise, if the parameter is not
referenced, it will be removed completely from the rule (line
\ref{line:noref}). Finally, the old predicates in $\mathcal{Q} \cap
\mathcal{P}$, which have become unreachable from the $P_i$, due to the
insertion of the new ones ($Q_r$), are removed from the system by a
call to the $\Call{Trim}{}$ procedure.
\fi
The running time of the algorithm is linear in the size of
$\mathcal{P}$ (i.e.\ each rule is visited at most once) and the size
of the output system is reduced with respect to the input,
i.e.\ $\len{\mathcal{Q}} \leq \len{\mathcal{P}}$.

\begin{example}[cont. of Ex.~\ref{ex:DLL}]\label{ex:DLLelim}As an example of
parameter elimination, let us take the predicate $\DLL$ introduced in
Sect.~\ref{sec:inductive-definitions}, called as $\DLL(\ax, \bx, \cx,
\dx)$. After the parameter elimination and renaming the newly created
predicates, we have a call $Q_1$ (without parameters) of the following
inductive system:\vspace*{-2mm}
\[\begin{array}{rcl}
Q_1() & \equiv & \ax \mapsto (\dx,\bx) ~\wedge~ \ax = \cx \mid
\exists x.~ \ax \mapsto (x, \bx) * Q_2(x,\ax) \\
Q_2(hd, p) & \equiv & hd \mapsto (\dx, p) ~\wedge~ hd= \cx  \mid
\exists x.~ hd \mapsto (x,p) * Q_2(x,hd)
\end{array}\]

\vspace*{-\baselineskip}
\end{example}

\ifLongVersion
\begin{algorithm}[t]
\begin{algorithmic}[0]
  \State {\bf input} A rooted system $\langle \mathcal{P}, P_i
  \rangle$, where $\mathcal{P} = \{P_i(x_{i,1}, \ldots, x_{i,n_i})
  \equiv |_{j=1}^{m_i} R_{i,j}\}_{i=1}^n$, and $P_i(\vec{x})$

  \State {\bf output} A rooted system $\langle \mathcal{Q}, Q_i
  \rangle$, where $Q_i$ has empty formal parameter tuple

  \State {\bf global} $\mathcal{Q} \leftarrow \mathcal{P}$,
  $\mathtt{Visited} \leftarrow \emptyset$, $\mathtt{Parameters}
  \leftarrow \emptyset$, $\mathtt{Tracked}$
\end{algorithmic}
\begin{algorithmic}[1]
  \Function{eliminateParameters}{$P_i,\overline{\alpha}$} \Comment{$\overline{\alpha}$ is the tuple of 
    actual parameters}

  \ForAll{$k = 1, \ldots, n_i$} \Comment{iterate through all parameters of $P_i$}
  \State{$\mathtt{Tracked} \leftarrow \overline{\alpha}_k$}
  \State{$\Call{trackEliminate}{P_i, x_{i,k}, \true}$}
  \State{$\mathtt{Parameters} \leftarrow \mathtt{Parameters} \cup \{\mathtt{Tracked}\}$}
  \EndFor
  \State{\bf return} $\Call{Trim}{\mathcal{Q},P_i}$
  \EndFunction
\end{algorithmic}
\begin{algorithmic}[1]
  \Function{trackEliminate}{$P_r, x_{r,s}, \mathtt{del}$} 
  \Comment{parameter $x_{r,s}$ of $P_r$, boolean $\mathtt{del}$}

  \State{$\vec{x}_{new} \leftarrow \mathbf{if}~ \mathtt{del}
    ~\mathbf{then}~ \vec{x}_{\neg x_{r,s}} ~\mathbf{else}~ \vec{x}$}
  \Comment{initialize a new formal parameter tuple}

  \State{$Q_r(\vec{x}_{new}) \leftarrow \mathbf{empty\_predicate}$}
  \State{$\mathcal{Q} \leftarrow \mathcal{Q} \cup \{Q_r\}$}
  \Comment{create a new predicate name}

  \ForAll{$q = 1, \ldots, m_r$} \Comment{iterate through the rules of $P_r$}
  \State{{\bf assume} $R_{r,q}(\vec{x}) \equiv \exists \vec{z} ~.~ 
    \alpha \mapsto (\vec{y}) * P_{i_1}(\vec{u}_1) * \ldots * P_{i_m}(\vec{u}_m)$}

  \If{$x_{r,s} \equiv \alpha$ {\bf or} $x_{r,s} \not\in \bigcup_{j=1}^{m}\vec{u}_j$} 
  \Comment{$x_{r,s}$ is allocated in $R_{r,q}$ or not propagated}

  \State{$R_{new}(\vec{x}_{new}) \leftarrow R_{r,q}[\mathtt{Tracked}/x_{r,s}]$} 
  \label{line:tracked}
  \Comment{replace $x_{r,s}$ in $R_{r,q}$ by the global $\mathtt{Tracked}$}

  \Else \Comment{check if the current tracked parameter $x_{r,s}$ is passed to $P_j$ on position $\ell$}

  \If{$\exists j,\ell ~.~ 1 \leq j \leq m ~\wedge~ 0 \leq \ell <
    \len{\vec{y}_j} ~\wedge~ x_{r,s} \equiv (\vec{u}_j)_\ell$}
  \label{line:prop}
  
  \If{the choice of $j$ and $\ell$ is not unique} 
  \State {\bf error} (``branching propagation for parameter $x_{r,s}$'') 
  \EndIf

  \If{$x_{r,s} \in \vec{y}$} \Comment{the tracked parameter is referenced before being passed}\label{line:ref}

  \State $R_{new}(\vec{x}_{new}) \leftarrow ~\mathbf{if}~ \mathtt{del}
  ~\mathbf{then}~ \exists x_{r,s} ~.~ R_{r,q} ~\mathbf{else}~ R_{r,q}$

  \State{$\mathtt{del} \leftarrow \false$}\label{line:del:false}
  \Else \Comment{the tracked parameter is passed without being referenced}
  \If{$\mathtt{del}$}
  \State{$R_{new}(\vec{x}_{new}) \leftarrow ~\exists \vec{z} ~.~ 
    \alpha \mapsto (\vec{y}) * P_{i_1}(\vec{u}_1) * \ldots * Q_{i_j}(\vec{u}_{\neg x_{r,s}}) * 
    \ldots * P_{i_m}(\vec{u}_m)$}
  \label{line:noref}
  \Else
  \State $R_{new}(\vec{x}_{new}) \leftarrow R_{r,q}$
  \EndIf
  \EndIf  

  \If{$x_{j,\ell} \not\in \mathtt{Visited}$} \Comment{continue tracking parameter $x_{j,\ell}$ of $P_j$}
  \State $\mathtt{Visited} \leftarrow \mathtt{Visited} \cup \{x_{j,\ell}\}$
  \State \Call{trackEliminate}{$P_j,x_{j,\ell},\mathtt{del}$} 
  \EndIf

  \EndIf
  \EndIf
  \State{$Q_r \equiv Q_r ~|~ R_{new}$}
  \EndFor
  \EndFunction
\end{algorithmic}
\caption{Elimination of propagated formal parameters} \label{alg:eliminate}
\end{algorithm}
\fi

\vspace*{-2mm}\subsection{Conversion to tree automata}\vspace*{-1mm}

The following algorithm takes as input a rooted inductive system $\langle
\mathcal{P}, P_i \rangle$ such that $P_i$ has no formal parameters, the head of
each rule in $\mathcal{P}$ is of the form $\exists\vec{z} ~.~ \alpha \mapsto
(y_0,\ldots,y_{m-1}) \wedge \Pi$, and for each $x \in FV(\Pi)$, we have (i)
$\alpha \in [x]_{\Pi}$ and (ii) $[x]_{\Pi} \cap \{y_0, \ldots, y_{m-1}\} =
\emptyset$. It is easy to see that these conditions are ensured by the first two
steps of pre-processing (i.e.\ the elimination of equalities and the reduction
to one points-to proposition per rule). The first step of the conversion
computes several sets of parameters that are called {\em signatures} and defined
as follows:\vspace*{-1.5mm}\[\begin{array}{rcl}

  \sig_j^{fw} & = & \{x_{j,k} ~|~ x_{j,k} ~\mbox{is allocated in every rule of
  $P_j$}\} ~\cap \\

  && \{x_{j,k} ~|~ \mbox{there is an edge to $(\vec{y})_k$ in every occurrence
  $P_j(\vec{y})$ of $\mathcal{P}$}\} \\

  \sig_j^{bw} & = & \{x_{j,k} ~|~ \mbox{there is an edge to $x_{j,k}$ in every
  rule of $P_j$}\} ~\cap \\

  && \{x_{j,k} ~|~ ~\mbox{$(\vec{y})_k$ is allocated for every occurrence
  $P_j(\vec{y})$ of $\mathcal{P}$}\} \\

  \sig_j^{eq} & = & \{x_{j,1}, \ldots, x_{j,n_j}\} \setminus (\sig_j^{fw} \cup
  \sig_j^{bw})\vspace*{-1.5mm}

\end{array}\] where $\langle x_{j,1}, \ldots, x_{j,n_j} \rangle$ is the tuple of
formal parameters of the predicate $P_j$, in $\mathcal{P}$, for all
$j=1,\ldots,n$. The signatures of the system are computed via
Algorithm \ref{alg:sl2ta-signatures} .

The result of the algorithm can be used to implement a sufficient test
of locality: a given system is local if $sig_i^{eq}=\emptyset$, for
each $i$. This simple test turned out to be powerful enough for all
the examples that we considered in our experiments in
Section~\ref{sec:experiments}.\vspace*{-2mm}

\enlargethispage{4mm}

\begin{example}[cont. of Ex.~\ref{ex:DLLelim}]\label{ex:DLLsign} The signatures
for the system in Example~\ref{ex:DLLelim} are:
\ifLongVersion
\begin{itemize}
  \item $sig_1^{fw}=sig_1^{bw}=sig_1^{eq}=\emptyset$
  \item $sig_2^{fw}=\{0\}, sig_2^{bw}=\{1\}, sig_2^{eq}=\emptyset$
\end{itemize}
\fi
The fact that, for each $i$, we have $sig_i^{eq}=\emptyset$
informs us that the system is local.\vspace*{-2mm}
\end{example}

\ifLongVersion
\begin{algorithm}[t]
\begin{algorithmic}[0]
\State {\bf input} A rooted system $\langle \mathcal{P}, P_i \rangle$,
where $\mathcal{P} = \{P_i \equiv |_{j=1}^{m_i} R_{i,j}\}_{i=1}^n$
\State {\bf global} $\sig_1^{fw}, \sig_1^{bw}, \sig_1^{eq}, 
\ldots, \sig_n^{fw}, \sig_n^{bw}, \sig_n^{eq}$
\end{algorithmic}
\begin{algorithmic}[1]
\Function{fwEdge}{$\mathtt{index},\mathtt{pred_{no}},\mathcal{P}$}
\ForAll{$i = 1, \ldots, n$} \Comment{iterate through all predicates}
\ForAll{$j = 1, \ldots, m_i$} \Comment{iterate through all rules}
\State {\bf assume} $R_{i,j}(\vec{x}) \equiv \exists \vec{z} ~.~ \alpha 
\mapsto \vec{y}	* P_{i_0}(\vec{x}_0)   * \ldots * P_{i_k}(\vec{x}_k) \wedge \Pi$
\State $\mathtt{res} \leftarrow \true$
\ForAll{$\ell = 0, \ldots, k$} \Comment{iterate through all predicate occurrences}
\If{$i_\ell=\mathtt{pred_{no}}$}
\If{$(\vec{x}_\ell)_{\mathtt{index}} \not\in \vec{y}$} 
\Comment{parameter not referenced in occurrence of $P_{\mathtt{pred}_{no}}$}
\State $\mathtt{res} \leftarrow \false$
\EndIf
\EndIf
\EndFor
\EndFor
\EndFor
\State {\bf return} $\mathtt{res}$
\EndFunction
\end{algorithmic}
\begin{algorithmic}[1]
\Function{bwEdge}{$\mathtt{index},\mathtt{pred_{no}},\mathcal{P}$}
\ForAll{$i = 1, \ldots, n$} \Comment{iterate through all predicates}
\ForAll{$j = 1, \ldots, m_i$} \Comment{iterate through all rules}
\State {\bf assume} $R_{i,j}(\vec{x}) \equiv \exists \vec{z} ~.~ \alpha \mapsto
\vec{y}	* P_{i_0}(\vec{x}_0)   * \ldots * P_{i_k}(\vec{x}_k) \wedge \Pi$
\State $\mathtt{res} \leftarrow \true$
\ForAll{$\ell = 0, \ldots, k$} \Comment{iterate through all predicate occurrences}
\If{$i_\ell=\mathtt{pred_{no}}$}
\If{$(\vec{x}_\ell)_{\mathtt{index}} \neq_{\Pi} \alpha$} 
\Comment{parameter not allocated in occurrence of $P_{\mathtt{pred}_{no}}$}
\State $\mathtt{res} \leftarrow \false$
\EndIf
\EndIf
\EndFor
\EndFor
\EndFor
\State {\bf return} $\mathtt{res}$
\EndFunction
\end{algorithmic}
\begin{algorithmic}[1]
\Function{allocated}{$\mathtt{index}$,$\mathtt{pred_{no}}$,$\mathcal{P}$}
\State $\mathtt{res} \leftarrow \true$
\ForAll{$j = 1, \ldots, m_{\mathtt{pred_{no}}}$} \Comment{iterate through all rules of $P_{\mathtt{pred_{no}}}$}
\State {\bf assume} $R_{\mathtt{pred_{no}},j}(\vec{x}) \equiv \exists \vec{z} ~.~ \alpha \mapsto
\vec{y}	* P_{i_0}(\vec{x}_0)   * \ldots * P_{i_k}(\vec{x}_k) \wedge \Pi$
\If{$(\vec{x})_{\mathtt{index}} \neq_{\Pi} \alpha$}
\Comment{formal parameter $\mathtt{index}$ not allocated in $P_{\mathtt{pred_{no}}}$}
\State $\mathtt{res} \leftarrow \false$
\EndIf
\EndFor
\State {\bf return} $\mathtt{res}$
\EndFunction
\end{algorithmic}
\begin{algorithmic}[1]
\Function{referenced}{$\mathtt{index}$,$\mathtt{pred_{no}}$,$\mathcal{P}$}
\State $\mathtt{res} \leftarrow \true$
\ForAll{$j = 1, \ldots, m_{\mathtt{pred_{no}}}$} \Comment{iterate through all rules of $P_{\mathtt{pred_{no}}}$}
\State {\bf assume} $R_{\mathtt{pred_{no}},j}(\vec{x}) \equiv \exists \vec{z} ~.~ \alpha \mapsto
\vec{y}	* P_{i_0}(\vec{x}_0)   * \ldots * P_{i_k}(\vec{x}_k) \wedge \Pi$
\If{$(\vec{x})_{\mathtt{index}} \not\in \vec{y}$}
\Comment{formal parameter $\mathtt{index}$ not referenced in $P_{\mathtt{pred_{no}}}$}
\State $\mathtt{res} \leftarrow \false$
\EndIf
\EndFor
\State {\bf return} $\mathtt{res}$
\EndFunction
\end{algorithmic}
\begin{algorithmic}[1]
\Function{computeSignatures}{$\mathcal{P}$}
\ForAll{$i = 1, \ldots, n$} \Comment{iterate through all predicates of $\mathcal{P}$}
\State {\bf assume} $P_i(\vec{x})$ \Comment{$P_i$ has formal parameters $\vec{x}$}
\ForAll{$p=0, \ldots, \len{\vec{x}} - 1$} \Comment{iterate through all formal parameters of $P_i$}
\State $\sig_{i}^{fw}\leftarrow \emptyset, \sig_{i}^{bw}\leftarrow \emptyset,
\sig_{i}^{eq}\leftarrow \emptyset$
\If {$\Call{fwEdge}{i,p,\mathcal{P}}$ {\bf and} $\Call{allocated}{i,p,\mathcal{P}}$}
\State $\sig_{i}^{fw} \leftarrow \sig_{i}^{fw} \cup \{p\}$
\Else
\If{$\Call{bwEdge}{i,p,\mathcal{P}}$ {\bf and} $\Call{referenced}{i,p,\mathcal{P}}$}
\State $\sig_{i}^{bw} \leftarrow \sig_{i}^{bw} \cup \{p\}$
\Else
\State $\sig_{i}^{eq} \leftarrow \sig_{i}^{eq} \cup \{p\}$
\EndIf
\EndIf
\EndFor
\EndFor
\EndFunction
\end{algorithmic}
\caption{Signature computation} \label{alg:sl2ta-signatures}
\end{algorithm}
\fi

\ifLongVersion
Next, we define a normal form for the quasi-canonical tiles, used as alphabet
symbols in the rules of the tree automaton. Given a quasi-canonical tile $T =
\langle \varphi, \vec{x}_{-1}, \vec{x}_0, \ldots,$ $\vec{x}_{k-1} \rangle$, for
some $k \geq 0$, where $\varphi \equiv (\exists z)~ z \mapsto (y_0, \ldots,
y_{m-1}) \wedge \Pi$, for some $m > 0$, let $\xi_0, \ldots, \xi_m$ be unique
variable namesi. Let $\sigma$ be the substitution defined as $\sigma(z) =
\xi_0$, and $\sigma(y_i)=\xi_{i+1}$, for all $0 \leq i < m$. For all variables
$\vec{x}_{-1} \cup \vec{x}_0 \cup \ldots \cup \vec{x}_{k-1} = \{\alpha_1,
\ldots, \alpha_n\}$, we chose unique names $\xi_{m+1}, \ldots, \xi_{m+n}$,
according to the order in which they appear in the ports $\vec{x}_{-1},
\vec{x}_0, \ldots, \vec{x}_{k-1}$, respectively, and we extend the substitution
further by defining $\sigma(\alpha_i)=\xi_{m+i}$, for all $i = 1, \ldots, n$.
We define the {\em normal tile} $\overline{T} = \langle \varphi[\sigma],
\vec{x}_{-1}[\sigma], \vec{x}_0[\sigma], \ldots,$ $\vec{x}_{k-1}[\sigma]
\rangle$. For a quasi-canonical tree $t : \nat^* \rightharpoonup_{fin}
\mathcal{T}^{qc}$, we define the {\em normal tree} $\overline{t}$ as
$\overline{t}(p)=\overline{t(p)}$, for all $p \in dom(t)$.

The $\Call{sl2ta}{}$ function (Algorithm \ref{alg:sl2ta}) builds a
tree automaton $A$, for a rooted system $\langle \mathcal{P}, P_i
\rangle$. For each rule in the system, the algorithm creates a
quasi-canonical tile, where the input and output ports $\vec{x}_i$ are
factorized as $\vec{x}_i^{fw} \cdot \vec{x}_i^{bw} \cdot
\vec{x}_i^{eq}$, according to the precomputed signatures. The backward
part of the input port $\vec{x}_{-1}^{bw}$ and the forward parts of
the output ports $\vec{x}_i^{fw}$, for $i \geq 0$, are sorted
according to the order of incoming selector edges from the single
points-to formula $\alpha \mapsto (\overline{\beta})$, by the
$\Call{selectorSort}{}$ function (lines \ref{line:selsort1} and
\ref{line:selsort2}). The output ports $x_i$, $i \geq 0$, are sorted
within the tile, according to the order of the selector edges pointing
to $(\vec{x}_i^{fw})_0$, for each $i=0,1,\ldots$ (function
$\Call{sortTile}{}$, line \ref{line:tilesort}). Finally, each
predicate name $P_j$ is associated a state $q_j$ (line
\ref{line:states}), and for each inductive rule, the algorithm creates
a transition rule in the tree automaton (line \ref{line:transition}).
The final state corresponds to the root of the system (line
\ref{line:final}). 
\fi 
The following lemma summarizes the tree automata construction.
\begin{lemma}\label{sl-ta} 
Given a rooted inductive system $\langle \mathcal{P}, P_i(\vec{x})
\rangle$ where $\mathcal{P} = \big\{P_i ~\equiv~ \mid_{j=1}^{m_i}
R_{i,j}\big\}_{i=1}^{n}$ and $\vec{x} = \langle
x_{i,1},\ldots,x_{i,n_i} \rangle$, and a vector $\overline{\alpha} =
\langle \alpha_1, \ldots, \alpha_{n_i} \rangle$ of variables not used
in $\mathcal{P}$. Then, for every state $S$, we have $S \models
P_i(\overline{\alpha})$ if and only if there exists $t \in \mathcal{L}(A)$
such that $S \models \Phi(t)$ where $A = \Call{sl2ta}{\mathcal{P}, i,
  \overline{\alpha}}$. Moreover, $\len{A} =
\mathcal{O}(\len{\mathcal{P}})$.
\end{lemma}

\begin{example}[cont. of Ex.~\ref{ex:DLLsign}]\label{ex:DLLtoTA}
The automaton corresponding to the $\DLL$ system called by $\DLL(\ax,
\bx, \cx, \dx)$ is $A=\langle \Sigma, \{q_1, q_2\}, \Delta, \{q_1\}
\rangle$ where:\vspace*{-1.5mm}{\small
\[\begin{array}{rcl}
\Delta & = & \left\{\begin{array}{lcl}
\overline{\langle \ax \mapsto (\bx, \dx) \wedge \ax = \cx, \emptyset \rangle} ()  \rightarrow
q_1 \ \ \ \ \ 
\overline{\langle \ax \mapsto (x, \bx), \emptyset, (x, \ax) \rangle} (q_2) & \rightarrow & q_1 \\
\overline{\langle \exists hd'. hd' \mapsto (\dx, p) \wedge hd = \cx \wedge hd' = hd, (hd, p) \rangle} () & \rightarrow & q_2 \\
\overline{\langle \exists hd'. hd' \mapsto(x, p) \wedge hd' = hd, (hd, p), (x, hd) \rangle} (q_2) & \rightarrow & q_2
\end{array}\right\}
\end{array}\]} 
\end{example}

\ifLongVersion
\begin{algorithm}[t]
\begin{algorithmic}[0]
  \State {\bf input} A rooted system $\langle \mathcal{P},
  P_{\mathtt{index}} \rangle$, $\mathcal{P} = \{P_i \equiv
  |_{j=1}^{m_i} R_{i,j}\}_{i=1}^n$, and actual parameters $\vec{u}$ of a call of
  $P_{\mathtt{index}}$

  \State {\bf output} A tree automaton $A = \langle \Sigma, Q, \Delta,
  F \rangle$
\end{algorithmic}
\begin{algorithmic}[1]
  \Function{sl2ta}{$\mathcal{P}, \mathtt{index}, \vec{u}$} 
  \State $\mathcal{P} \leftarrow \Call{splitSystem}{\mathcal{P}}$
  \State $\mathcal{P} \leftarrow \Call{eliminateParameters}{P_{\mathtt{index}},\vec{u}}$
  \State $\Call{computeSignatures}{\mathcal{P}}$
  \State $\Delta \leftarrow \emptyset, \Sigma \leftarrow \emptyset$
  \ForAll{$i = 1, \ldots, n$} \Comment{iterate through all predicates of $\mathcal{P}$}
  \ForAll{$j = 1, \ldots, m_i$} \Comment{iterate through all rules of $P_i$}
  \State {\bf assume} $R_{i,j}(\vec{x}) \equiv \exists \vec{z} ~.~ \alpha \mapsto
  \overline{\beta} * P_{i_0}(\vec{z}_0)   * \ldots * P_{i_k}(\vec{z}_k) \wedge \Pi$

  \State{$\phi \leftarrow ~\mbox{\bf if}~ \alpha\in\mathtt{Parameters}
    ~\mbox{\bf then}~ \alpha \mapsto (\overline{\beta}) \wedge \Pi
    ~\mbox{\bf else}~ \exists \alpha' ~.~ \alpha' \mapsto
    (\overline{\beta}) \wedge \Pi \wedge \alpha=\alpha'$}
  
  \State{$\vec{x}_{-1}^{fw} \leftarrow \langle\rangle, 
    \vec{x}_{-1}^{bw} \leftarrow \langle\rangle,
    \vec{x}_{-1}^{eq}\leftarrow \langle\rangle$}
  \Comment{initialize input ports}

  \ForAll{$\ell = 0,\ldots,\len{\vec{x}}-1$} 
  \Comment{iterate through the formal parameters of $R_{i,j}$}
  
  \If{$\ell \in \sig_i^{fw}$}
  \State{$\vec{x}_{-1}^{fw} \leftarrow \vec{x}_{-1}^{fw} \cdot \langle (\vec{x})_\ell \rangle$}
  \EndIf
  
  \If{$\ell \in \sig_i^{bw}$}
  \State $\vec{x}_{-1}^{bw} \leftarrow \vec{x}_{-1}^{bw} \cdot \langle (\vec{x})_\ell \rangle$
  \EndIf
  
  \If{$\ell \in \sig_{i}^{eq}$}
  \State $\vec{x}_{-1}^{eq} \leftarrow \vec{x}_{-1}^{eq} \cdot \langle (\vec{x})_\ell \rangle$
  \EndIf
  
  \EndFor
  
  \State $\Call{selectorSort}{\vec{x}_{-1}^{bw},\overline{\beta}}$\label{line:selsort1}
  \ForAll{$\ell=0, \ldots, k$} \Comment{iterate through predicate occurrences in $R_{i,j}$}
  \State{$\vec{x}_\ell^{fw} \leftarrow \langle\rangle, 
    \vec{x}_\ell^{bw} \leftarrow \langle\rangle,
    \vec{x}_\ell^{eq} \leftarrow \langle\rangle$}
  \Comment{initialize output ports}
  \ForAll{$r=0,\dots,\len{\vec{z}_\ell}-1$} 
  \Comment{iterate through variables of occurrence $P_{i_\ell}(\vec{z}_\ell)$}
    \If{$r \in \sig_{i_\ell}^{fw}$}
    \State $\vec{x}_{\ell}^{fw} \leftarrow \vec{x}_{\ell}^{fw} \cdot \langle (\vec{z}_\ell)_r$
    \EndIf

    \If{$r \in \sig_{i_\ell}^{bw}$}
    \State $\vec{x}_{\ell}^{bw} \leftarrow \vec{x}_{\ell}^{bw} \cdot \langle (\vec{z}_\ell)_r \rangle$
    \EndIf
     	
    \If{$r \in \sig_{i_\ell}^{eq}$}
    \State $\vec{x}_{\ell}^{eq} \leftarrow \vec{x}_{\ell}^{eq} \cdot \langle (\vec{z}_\ell)_r \rangle$
    \EndIf
    
    \EndFor
    \State $\Call{selectorSort}{\vec{x}_{\ell}^{fw},\overline{\beta}}$\label{line:selsort2}
    \EndFor
	
    \State $\mathtt{T} \leftarrow \langle \phi,\vec{x}_{-1},\vec{x}_0,
    \ldots, \vec{x}_k\rangle$ \Comment{create a new tile}

    \State $\mathtt{lhs} \leftarrow \langle q_{i_0}, \ldots, q_{i_k}
    \rangle$ \Comment{create the left hand side of the transition rule
      for $R_{i,j}$}

    \State $(\mathtt{T_{new}},\mathtt{lhs_{new}}) \leftarrow \Call{sortTile}{\mathtt{T},\mathtt{lhs}}$
    \label{line:tilesort}
    \State $\Sigma \leftarrow \Sigma\cup \{\overline{\mathtt{T_{new}}}\}$
    \Comment{insert normalized tile in the alphabet}
    \State $Q \leftarrow Q \cup \{q_i, q_{i_0}, \ldots, q_{i_k}\}$
    \label{line:states}
    \State $\Delta \leftarrow \Delta \cup \{ \mathtt{\overline{T_{new}}}(\mathtt{lhs_{new}}) \rightarrow q_i \}$
    \Comment{build transition using normalized tile}
    \label{line:transition}
    \EndFor
    \EndFor
    \State {$F=\{ q_{\mathtt{index}} \}$}
    \label{line:final}
    \State {{\bf return} $A=\langle Q,\Sigma,\Delta,F \rangle$}
    \EndFunction   
\end{algorithmic}
\caption{Converting rooted inductive systems to tree automata} \label{alg:sl2ta}
\end{algorithm}
\fi

\subsection{Rotation of Tree Automata}\vspace*{-1mm}\label{sec:canonical-rotation-ta}

This section describes the algorithm that produces the closure of a
quasi-canonical tree automaton (i.e.\ a tree automaton recognizing
quasi-canonical trees only) under rotation. The result of the
algorithm is a tree automaton $A^r$ that recognizes all trees $u :
\nat^* \rightharpoonup_{fin} \mathcal{T}^{qc}$ such that $t \sim^{qc}
u$, for some tree $t$ recognized by $A = \langle Q, \Sigma, \Delta, F
\rangle$. Algorithm \ref{alg:rot-closure} describes the rotation closure.
\ifLongVersion
The result of Algorithm \ref{alg:rot-closure} (function
$\Call{sl2ta}{}$) is a language-theoretic union of $A$ and automata
$A_\rho$, one for each rule $\rho$ of $A$. The idea behind the
construction of $A_\rho = \langle Q_\rho, \Sigma, \Delta_\rho,
\{q^f_\rho\} \rangle$ can be understood by considering a tree $t \in
\lang{A}$, a run $\pi : dom(t) \rightarrow Q$, and a position $p \in
dom(t)$, which is labeled with the right hand side of the rule $\rho =
T(q_1, \ldots, q_k) \arrow{}{} q$ of $A$. Then $\lang{A_\rho}$ will
contain the rotated tree $u$, i.e.\ $t \sim^{qc}_r u$, where the
significant position $p$ is mapped into the root of $u$ by the
rotation function $r$, i.e.\ $r(p) = \epsilon$. To this end, we
introduce a new rule $T_{new}(q_0, \ldots, q_j, q^{rev}, q_{j+1},
\ldots, q_n) \arrow{}{} q^f_\rho$, where the tile $T_{new}$ mirrors
the change in the structure of $T$ at position $p$, and $q^{rev} \in
Q_\rho$ is a fresh state corresponding to $q$. The construction of
$A_\rho$ continues recursively, by considering every rule of $A$ that
has $q$ on the left hand side: $U(q'_1, \ldots, q, \ldots, q'_\ell)
\arrow{}{} s$. This rule is changed by swapping the roles of $q$ and
$q'$ and producing a rule $U_{new}(q'_1, \ldots, s^{rev}, \ldots
q'_\ell) \arrow{}{} q^{rev}$, where $U_{new}$ mirrors the change in
the structure of $U$. Intuitively, the states $\{q^{rev} | q \in Q\}$
mark the unique path from the root of $u$ to $r(\epsilon) \in
dom(u)$. The recursion stops when either (i) $s$ is a final state of
$A$, (ii) the tile $U$ does not specify a forward edge in the
direction marked by $q$, or (iii) all states of $A$ have been
visited. 
\fi
The following lemma proves the correctness of Algorithm
\ref{alg:rot-closure}.
\ifLongVersion
\begin{lemma}\label{canonical-rotation-ta}
  Let $A = \langle Q, \mathcal{T}^{qc}, \Delta, F \rangle$ be a tree
  automaton. Then $\lang{A^r} = \{\overline{u} ~|~ u : \nat^*
  \rightharpoonup_{fin} \mathcal{T}^{qc},~ \exists t \in \lang{A} ~.~
  u \sim^{qc} t\}$. Moreover, the size of $A^r$ is of the order of
  $\mathcal{O}(\len{A}^2)$.
\end{lemma}
\fi
\ifLongVersion
\begin{proof}
``$\subseteq$'' Let $u \in \lang{A^r}$ be a normal quasi-canonical
  tree. Then $u \in \lang{A} \cup \bigcup_{\rho \in \Delta}
  \lang{A_\rho}$ (lines \ref{line:init} and \ref{line:union} in
  Alg. \ref{alg:rot-closure}). If $u \in \lang{A}$, then we choose
  $t=u$ and trivially $t \sim^{qc} u$. Otherwise $u \in
  \lang{A_\rho}$, where $A_\rho = \langle Q_\rho, \mathcal{T}^{qc},
  \Delta_\rho, \{q_\rho^f\} \rangle$, for some $\rho \in \Delta$ (line
  \ref{line:A-rho}). Let $\pi : dom(u) \rightarrow Q_\rho$ be the run
  of $A_\rho$ on $u$. Also let $p_0 \in dom(u)$ be the maximal
  (i.e.\ of maximal length) position such that $u(p_0) = q^{rev}$, for
  some $q^{rev} \in Q^{rev}$ (line \ref{line:Q-rev}). Let
  $\subtree{u}{p_0.i}$ be the subtrees of $u$ rooted at the children
  of $p_0$, for all $i=0, \ldots, \#_u(p_0)-1$. Since $p_0$ is the
  maximal position of $\pi$ to be labeled by some state in $Q^{rev}$,
  it is easy to see that $\subtree{\pi}{p_0.i}$ are labeled by states
  in $Q$, hence $\subtree{\pi}{p_0.i}$ are runs of $A$ over
  $\subtree{t}{p_0.i}$, for all $i=0, \ldots, \#_u(p_0)-1$ (the only
  rules of $A_\rho$ involving only states from $Q$ are the rules of
  $A$, cf. line \ref{line:copy}).

  We build a quasi-canonical tree $t : \nat^* \rightharpoonup_{fin}
  \mathcal{T}^{qc}$ and a run $\theta : dom(t) \rightarrow Q$ of $A$
  on $t$ top-down as follows. Let $p_0, p_1, \ldots, p_n=\epsilon$ be
  the path composed of the prefixes of $p_0$,
  i.e.\ $p_i=p_{i-1}.(-1)$, for all $i=1,\ldots,\len{p_0}=n$. We build
  $t$, $\theta$, and a path
  $\overline{p_0}=\epsilon,\overline{p_1}=j_0, \ldots,
  \overline{p_{n}} \in dom(t)$, by induction on this path.

  For the base case, let $u(p_0) = \langle \varphi, \vec{x}_{-1}^{fw}
  \cdot \vec{x}_{-1}^{bw} \cdot \vec{x}_{-1}^{eq}, \vec{x}_0, \ldots,
  \vec{x}_{\#_u(p)-1} \rangle$ be the tile which labels $u$ at
  position $p_0$. Let $t(\epsilon) = \langle \varphi, \emptyset,
  \vec{x}_0, \ldots, \vec{x}_{j_0-1}, \vec{x}_{-1}^{bw} \cdot
  \vec{x}_{-1}^{fw} \cdot \vec{x}_{-1}^{eq}, \vec{x}_{j_0}, \ldots,
  \vec{x}_{\#_u(p)-1} \rangle$, where $j_0+1$ is the canonical
  position of the output port $\vec{x}_{-1}^{bw} \cdot
  \vec{x}_{-1}^{fw} \vec{x}_{-1}^{eq}$, in $t(\epsilon)$, according to
  the order of selector edges in $\varphi$. Then $A_\rho$ has a
  rule: $$u(p_0)(\pi(p_0.0), \ldots, \pi(p_0.(\#_u(p)-1))) \arrow{}{}
  \pi(p)$$ such that $\{\pi(p_0.0), \ldots, \pi(p_0.(\#_u(p)-1)\} \cap
  Q^{rev} = \emptyset$. Since $p_0$ is the maximal position labeled
  with a state from $Q^{rev}$, this rule was generated at line
  \ref{line:last-rule} in Alg. \ref{alg:rot-closure}. It follows that
  $A$ has a rule $$t(\epsilon)(\pi(p_0.0), \ldots,
  \pi(p_0.(\#_u(p_0)-1))) \arrow{}{} q$$ for some final state $q \in
  F$ (line \ref{line:final-state}). Then let $\theta(\epsilon) =
  q$. We further define:
  \[
  \subtree{t}{i} = \left\{\begin{array}{ll}
  \subtree{u}{p_0.i} & \mbox{if $i = 0, \ldots, j_0-1$} \\
  \subtree{u}{p_0.(i+1)} & \mbox{if $i=j_0+1, \ldots, \#_u(p_0)-1$}
  \end{array}\right. \hspace*{5mm}
  \subtree{\theta}{i} = \left\{\begin{array}{ll} 
  \subtree{\pi}{p_0.i} & \mbox{if $i = 0, \ldots, j_0-1$} \\ 
  \subtree{\pi}{p_0.(i+1)} & \mbox{if $i=j_0+1, \ldots, \#_u(p_0)-1$}
  \end{array}\right.
  \]
  For the induction step, for each $0 < i \leq n$, we have
  $p_{i-1}=p_i.k_i$, for some $0 \leq k_i < \#_u(p_{i-1})$. We have
  $u(p_i) = \langle \varphi, \vec{x}_{-1}, \vec{x}_0, \ldots,
  \vec{x}_{\#_u(p_i)-1} \rangle$, and define:
  $$t(\overline{p_i}) = \langle \varphi, ~\vec{x}_{k_i}^{bw} \cdot
  \vec{x}_{k_i}^{fw} \cdot \vec{x}_{k_i}^{eq}, ~\vec{x}_0, ~\ldots,
  ~\vec{x}_{j_i-1}, ~\vec{x}_{-1}^{bw} \cdot \vec{x}_{-1}^{fw} \cdot
  \vec{x}_{-1}^{eq},~ \vec{x}_{j_i}, ~\ldots, ~\vec{x}_{\#_u(p_i)-1}
  \rangle$$ where $j_i+1$ is the canonical position of the port
  $\vec{x}_{-1}^{bw} \cdot \vec{x}_{-1}^{fw} \cdot \vec{x}_{-1}^{eq}$
  given by the the selector edges in $\varphi$. Moreover, $A_\rho$ has
  a rule: $$u(p_i)(\pi(p_i.0), \ldots, \pi(p_i.(\#_u(p_i)-1)))
  \arrow{}{} \pi(p_i)$$ where $\pi(p_i.k_i)=\pi(p_{i-1}) \in Q^{rev}$,
  $\pi(p_i) \in Q^{rev}$ if $0 \leq i < n$, and $\pi(p_n) =
  q^f_\rho$. Moreover, this rule was introduced at line
  \ref{line:rot-rule}, if $i < n$, or at line
  \ref{line:final-rot-rule}, if $i=n$. Let $i < n$ (the case $i = n$
  uses a similar argument). If $\pi(p_{i-1})=s^{rev}$ and
  $\pi(p_i)=q^{rev}$, then $A$ must have a
  rule: $$t(\overline{p_i})(\pi(p_i.0), \ldots, \pi(p_i.(j_i-1)), q,
  \pi(p_i.j_i), \ldots, \pi(p_i.(\#_u(p_i)-1))) \arrow{}{} s$$ Let
  $\overline{p_{i+1}}=\overline{p_i}.j_i$ and
  $\theta(\overline{p_{i+1}})=q$. We further define:
  \[
  \subtree{t}{\overline{p_i}.\ell} = \left\{\begin{array}{ll}
    \subtree{u}{p_i.\ell} & \mbox{if $\ell = 0, \ldots, j_i - 1$} \\
    \subtree{u}{p_i.(\ell+1)} & \mbox{if $\ell = j_i+1, \ldots, \#_u(p_i) - 1$}
  \end{array}\right. \hspace*{5mm}
  \subtree{\theta}{\overline{p_i}.\ell} = \left\{\begin{array}{ll} 
  \subtree{\pi}{p_i.\ell} & \mbox{if $\ell = 0, \ldots, j_i - 1$} \\ 
  \subtree{\pi}{p_i.(\ell+1)} & \mbox{if $i = j_i+1, \ldots, \#_u(p_i)-1$}
  \end{array}\right.
  \]
  We define the following rotation function $r : dom(t) \rightarrow
  dom(u)$. For each $i = 0, \ldots, n$, we have
  $r(\overline{p_i})=p_i$, and:
  \[
  r(\overline{p_i}.\ell) = \left\{\begin{array}{ll}
  r(p_i.\ell) & \mbox{if $\ell = 0, \ldots, j_i - 1$} \\
  r(p_i.(\ell+1)) & \mbox{if $\ell = j_i, \ldots, \#_u(p_i) - 1$}
  \end{array}\right.
  \]
  It is easy to check that indeed $t \sim^{qc}_r u$, and that $\theta$ is
  a run of $A$ over $t$. 

\noindent''$\supseteq$'' Let $u : \nat^* \rightharpoonup_{fin}
\mathcal{T}^{qc}$ be a quasi-canonical tree such that $t \sim^{qc}_r
u$, for some $t \in \lang{A}$ and a bijective function $r : dom(t)
\rightarrow{}{} dom(u)$. Let $\pi : dom(t) \rightarrow Q$ be a run of
$A$ over $t$. We build an accepting run $\theta$ of $A^r$ over the
normal tree $\overline{u}$. Let $p_0 \in dom(t)$ such that
$r(p_0) = \epsilon$ is the root of $u$. If $p_0 = \epsilon$, it is
easy to show that $dom(u) = dom(t)$ and $r(p) = p$, for all $p \in
dom(t)$, because both $t$ and $u$ are quasi-canonical trees, and the
order of children is given by the order of selector edges in the tiles
labeling the trees. Moreover, the normal form of these tiles is
identical, i.e.\ $\overline{u}(p)=t(p)$, for all $p \in dom(t)$, hence
$u \in \lang{A} \subseteq \lang{A}^r$ (cf. line \ref{line:init}). 

Otherise, if $p_0 \neq \epsilon$, we consider the sequence of prefixes
of $p_0$, defined as $p_i = p_{i-1}.(-1)$, for all $i = 1, \ldots, n =
\len{p_0}$. Applying Lemma \ref{rotation-reversion} to $p_0, \ldots,
p_n$ successively, we obtain a path $\epsilon = r(p_0), r(p_1),
\ldots, r(p_n) \in dom(u)$, such that $r(p_{i+1}) = r(p_i).d_i$, for
all $i=0,\ldots,n-1$, and some positive directions $d_0, \ldots,
d_{n-1} \in \mathcal{D}_+(u)$.

We build the run $\theta$ by induction on this path. For the base
case, let $t(p_0) = q_0$. Then $A$ has a rule $\rho = (t(p_0)(q_1,
\ldots, q_{\#_t(p_0)}) \arrow{}{} q_0) \in \Delta$. By the
construction of $A^r$, cf. line \ref{line:final-rot-rule}, $A_\rho$
has a rule $$(\overline{u}(\epsilon))(q_1, \ldots,
q_{d_0},q_0^{rev},q_{d_0+1}, \ldots, q_{\#_t(p_0)}) \arrow{}{}
q^f_\rho.$$ We define $\theta(\epsilon) = q^f_\rho$ and:
\[
\subtree{\theta}{i} = \left\{\begin{array}{ll} 
\subtree{\pi}{p_0.i} & \mbox{if $i = 0, \ldots, d_0-1$} \\ 
\subtree{\pi}{p_0.(i+1)} & \mbox{if $i=d_0+1, \ldots, \#_t(p_0)-1$}
\end{array}\right.
\]
For the induction step, we denote $q_i = \pi(p_i)$, for all $0 < i
\leq n$. Then $A$ has a rule $$(t(p_i))(\pi(i.0), \ldots,
\pi(i.(k_i-1)), q_i, \pi(i.(k_i+1)), \ldots, \pi(i.(\#_t(p_i)-1)))
\arrow{}{} q_{i+1},$$ for each $0 < i \leq n$. By the construction of
$A^r$, cf. line \ref{line:rot-rule}, $A_\rho$ has a rule:
$$(\overline{u(r(p_i))})(\pi(i.0), \ldots, \pi(i.(d_i-1)),
~q_{i+1}^{rev},~ \pi(i.(d_i+1)), \ldots, \pi(i.(\#_t(p_i)-1))) \arrow{}{}
q_i^{rev}.$$
We define $\theta(r(p_i))=q_i^{rev}$ and:
\[
\subtree{\theta}{r(p_i).\ell} = \left\{\begin{array}{ll} 
\subtree{\pi}{r(p_i).\ell} & \mbox{if $\ell = 0, \ldots, d_i-1$} \\ 
\subtree{\pi}{r(p_i).(\ell+1)} & \mbox{if $i=d_i+1, \ldots, \#_t(p_i)-1$}
\end{array}\right.
\]
It is not difficult to prove that $\theta$ is a run of $A_\rho$, and,
moreover, since $\theta(\epsilon)=q^f_\rho$, it is an accepting run,
hence $\overline{u} \in \lang{A_\rho} \subseteq A^r$ (cf. line
\ref{line:union}).

Concerning the size of $A^r$, notice that $\len{A^r} \leq \len{A} + \sum_{\rho
\in \Delta}\len{A_\rho}$ where $A_\rho = \langle Q_\rho, \Sigma, \Delta_\rho,$
$\{q^f_\rho\} \rangle$. We have that $\card{\Delta_\rho} \leq \card{\Delta}$,
and for each rule $\tau \in \Delta_\rho \setminus \Delta$ created at lines
\ref{line:last-rule}, \ref{line:final-rot-rule}, or \ref{line:rot-rule}, there
exists a rule $\nu \in \Delta$ such that $\len{\tau} \leq \len{\nu}+1$.
Moreover, $\card{\Delta_\rho \setminus \Delta} \leq \card{Q}$ since we introduce
a new rule for each state in the set $\mathtt{visited} \subseteq Q$.  Hence
$\len{A_\rho} = \len{A} + \sum_{\tau \in \Delta_\rho \setminus \Delta}
\len{\tau} \leq \len{A} + \len{A} + \card{Q} \leq 3\len{A}$. Hence $\len{A^r}
\leq 3 \len{A}^2 = \mathcal{O}(\len{A}^2)$. \qed\end{proof}
\fi

\ifLongVersion
\begin{algorithm}[t]
\begin{algorithmic}[0]
  \State {\bf input} A quasi-canonical tree automaton $A= \langle Q, \Sigma,
  \Delta, F \rangle$

  \State {\bf output} A tree automaton $A^r$, where $\lang{A^r} = \{u
  : \nat^* \rightharpoonup_{fin} \mathcal{T}^{qc} ~|~ \exists t \in
  \lang{A} ~.~ u \sim^{qc} t\}$
\end{algorithmic}
\begin{algorithmic}[1]
  \State $A^r = \langle Q_r, \Sigma, \Delta_r, F_r \rangle \leftarrow A$ 
  \label{line:init}
  \Comment{make a copy of $A$ into $A^r$}

  \ForAll{$\rho \in \Delta$} \Comment{iterate through the rules of $A$}

  \State{{\bf assume} $\rho \equiv T(q_0, \ldots, q_k) \rightarrow q$ {\bf and}
    $T \equiv \langle \varphi, \vec{x}_{-1},\vec{x}_0,\dots, \vec{x}_k \rangle$}
  \Comment{the chosen rule $\rho$ will recognize the root of the rotated tree}
  

  \If{$\vec{x}_{-1} \neq \emptyset$ {\bf or} $q \not\in F$}
  \State{{\bf assume} $\vec{x}_{-1}\equiv \vec{x}_{-1}^{fw}\cdot \vec{x}_{-1}^{bw}\cdot \vec{x}_{-1}^{eq}$}
  \Comment{factorize the input port}


  \If {$\vec{x}_{-1}^{bw} \neq \emptyset$} 

  \State\label{line:Q-rev} $Q^{rev} \leftarrow \{q^{rev} \mid q\in Q
  \}$ \Comment{ states $q^{rev}$ label the unique reversed path in
    $A^r$}
  
  \State\label{line:copy} $(Q_\rho,\Delta_\rho) \leftarrow (Q \cup
  Q^{rev} \cup \{q^f_\rho\},\Delta)$ \Comment{assuming $Q \cap Q^{rev}
    = \emptyset$, $q^f_\rho \not\in Q \cup Q^{rev}$}

  \State $p \leftarrow \Call{positionOf}{\vec{x}_{-1}^{bw},\varphi}$
  \Comment{find new output port for $\vec{x}_{-1}$ based on selectors of $\varphi$}

  \State\label{line:root-tile} $T_{new} \leftarrow \langle \varphi,
  \emptyset,\vec{x}_0,\ldots,\vec{x}_p,
  \vec{x}_{-1}^{bw}\cdot\vec{x}_{-1}^{fw}\cdot\vec{x}_{-1}^{eq},
  \vec{x}_{p+1},\ldots,\vec{x}_k \rangle$ \Comment{swap
    $\vec{x}_{-1}^{bw}$ with $\vec{x}_{-1}^{fw}$}

  \State\label{line:final-rot-rule} $\Delta_\rho \leftarrow \Delta_\rho \cup \{T_{new}(q_0,
  \ldots, q_p, q^{rev}, q_{p+1}, \ldots, q_k) \arrow{}{} q^f_\rho \}$

  \State $(\Delta_\rho,\_) \leftarrow \Call{rotateRule}{q,\Delta,\Delta_\rho,\emptyset,F}$
  \Comment{continue building $\Delta_\rho$ recursively}

  \State\label{line:A-rho} $A_\rho \leftarrow \langle Q_\rho, \Sigma, \Delta_\rho, \{q^f_\rho\} \rangle$
  \Comment{$A_\rho$ recognizes trees whose roots are labeled by $\rho$}

  \State $A^r \leftarrow A^r \cup A_\delta$ 
  \label{line:union}
  \Comment{incorporate $A_\delta$ into $A_r$}
  \EndIf
  
  \EndIf 
  \EndFor
  \State {\bf return} $A^r$
\end{algorithmic}
\begin{algorithmic}[1]
  \Function{rotateRule}{$q,\Delta,\Delta_{new},\mathtt{visited},F$} 
  \State $\mathtt{visited} \leftarrow \mathtt{visited} \cup \{q\}$ 
  \ForAll{$(U(s_0,\dots,s_\ell) \rightarrow s) \in \Delta$} 
  \Comment{iterate through the rules of $A$}

  \ForAll{$0 \leq j \leq \ell$ {\bf such that} $s_j=q$}
  \Comment{for all occurrences of $q$ on the rhs}

  \State{{\bf assume} $U = \langle \varphi, \vec{x}_{-1}, \vec{x}_0,
  \ldots, \vec{x}_j, \dots, \vec{x}_\ell\rangle$}

  \State{{\bf assume} $\vec{x}_{j} \equiv \vec{x}_{j}^{fw}\cdot \vec{x}_{j}^{bw}
    \cdot\vec{x}_{j}^{eq}$} \Comment{factorize the $\vec{x}_j$ output port}

  \If{$\vec{x}_{-1}=\emptyset$ {\bf and} $s\in F$}\label{line:final-state}
  \Comment{remove $\vec{x}_j$ from output and place it as input port}
  \State $U' \leftarrow \langle \varphi,\vec{x}_{j}^{bw}\cdot
  \vec{x}_{j}^{fw}\cdot\vec{x}_{j}^{eq},\vec{x}_0,\dots,
  \vec{x}_{j-1},\vec{x}_{j+1},\dots,\vec{x}_\ell\rangle$

  \State\label{line:last-rule} $\Delta_{new} \leftarrow \Delta_{new}
  \cup \{ U'(s_0,\dots,s_{j-1},s_{j+1},\dots,s_\ell) \arrow{}{}
  q^{rev} \}$

  \Else \Comment{swap $\vec{x}_{-1}$ with $\vec{x}_j$}
  \State $\vec{x}_{-1}\equiv \vec{x}_{-1}^{fw}\cdot \vec{x}_{-1}^{bw} \cdot\vec{x}_{-1}^{eq}$


  \If{$\vec{x}_{-1}^{bw} \neq \emptyset$}

  \State $\mathtt{ports} \leftarrow \langle
  \vec{x}_0,\dots,\vec{x}_{j-1},\vec{x}_{j+1}, \ldots,\vec{x}_\ell \rangle$

  \State $\mathtt{states} \leftarrow (s_0,\dots,s_{j-1},s_{j+1},\dots,s_\ell)$

  \State $p \leftarrow \Call{insertOutPort}{\vec{x}_{-1}^{bw} \cdot
    \vec{x}_{-1}^{fw}\cdot \vec{x}_{-1}^{eq},\mathtt{ports}, \varphi$}

  \State $\Call{insertLhsState}{s^{rev}, \mathtt{states}, p}$

  \State $U_{new} \leftarrow \langle \varphi,\vec{x}_{j}^{bw}\cdot
  \vec{x}_{j}^{fw}\cdot\vec{x}_{j}^{eq},\mathtt{ports}\rangle$
  \Comment{create rotated tile}

  \State\label{line:rot-rule} $\Delta_{new} \leftarrow \Delta_{new}
  \cup \{U_{new}(\mathtt{states}) \rightarrow q^{rev}\}$
  \Comment{create rotated rule}

  \If{$s \not\in \mathtt{visited}$}

  \State $(\Delta_{new},\mathtt{visited}) \leftarrow
  \Call{rotateRule}{s,\Delta,\Delta_{new},\mathtt{visited},F}$
  \EndIf 
  \EndIf 
  \EndIf
  \EndFor 
  \EndFor 
  \State {\bf return} $(\Delta_{new},\mathtt{visited})$
  \EndFunction
\end{algorithmic}
\caption{Rotation Closure of Quasi-canonical TA} \label{alg:rot-closure}
\end{algorithm}
\fi

The main result of this paper is given by the following theorem. The
entailment problem for inductive systems is reduced, in polynomial
time, to a language inclusion problem for tree automata. The inclusion
test is sound (if the answer is yes, the entailment holds), and
moreover, complete assuming that both systems are local. 

\begin{theorem}\label{sl-ta-entailment} Let $\mathcal{P} = \big\{P_i ~\equiv~
\mid_{j=1}^{m_i} R_{i,j}\big\}_{i=1}^{n}$ be a system of inductive definitions.
Then, for any two predicates $P_i(x_{i,1}, \ldots, x_{i,n_i})$ and $P_j(x_{j,1},
\ldots, x_{j,n_j})$ of $\mathcal{P}$ such that $n_i = n_j$ and for any variables
$\overline{\alpha} = \{\alpha_1, \ldots, \alpha_{n_i}\}$ not used in
$\mathcal{P}$, the following holds:\begin{itemize}

  \item ~[Soundness] $P_i(\overline{\alpha}) \models_{\mathcal{P}}
  P_j(\overline{\alpha})$ if $\lang{A_1} \subseteq \lang{A^r_2}$,

  \item ~[Completness] $P_i(\overline{\alpha}) \models_{\mathcal{P}}
  P_j(\overline{\alpha})$ only if $\lang{A_1} \subseteq \lang{A^r_2}$ provided
  that the rooted systems $\langle \mathcal{P}, P_i \rangle$ and $\langle
  \mathcal{P}, P_j \rangle$ are both local\vspace*{-1mm}

\end{itemize} where $A_1 = \Call{sl2ta}{\mathcal{P},i,\overline{\alpha}}$ and
$A_2 = \Call{sl2ta}{\mathcal{P},j,\overline{\alpha}}$.\end{theorem}
\ifLongVersion
\begin{proof} [Soundness] Let $S$ be a state such that $S \models
P_i(\overline{\alpha})$. By Lemma \ref{sl-ta}, there exists a~tree $t \in
\lang{A_1}$ such that $S \models \Phi(t)$. Since $t \in \lang{A_1} \subseteq
\lang{A_2^r}$, by Lemma \ref{canonical-rotation-ta}, there exists
a~quasi-canonical tree $u \in \lang{A_2}$ such that $t \sim^{qc} u$. By Lemma
\ref{quasi-canonical-rotation-lemma}, we have $S \models \Phi(u)$, and by Lemma
\ref{sl-ta} again, we obtain $S \models P_j(\overline{\alpha})$. Hence
$P_i(\overline{\alpha}) \models_{\mathcal{P}} P_j(\overline{\alpha})$.

[Completness] If $\langle \mathcal{P}, P_i \rangle$ and $\langle \mathcal{P},
P_j \rangle$ are local rooted systems, it follows that $A_1$ and $A_2$ recognize
only normal canonical trees. Let $t \in \lang{A_1}$ be a normal canonical tree.
It is not difficult to see that starting from $t$, one can build a state $S$ and
its spanning tree such that the right-hand side of Lemma~\ref{canonical-model}
holds. Then, by Lemma~\ref{canonical-model}, one gets that $S \models \Phi(t)$.
By Lemma \ref{sl-ta}, we have $S \models P_i(\overline{\alpha})$, and since
$P_i(\overline{\alpha}) \models_{\mathcal{P}} P_j(\overline{\alpha})$, we have
$S \models P_j(\overline{\alpha})$. By Lemma \ref{sl-ta} again, we obtain a
canonical tree $u \in \lang{A_2}$ such that $S \models \Phi(u)$. Since $S
\models \Phi(t)$ and $S \models \Phi(u)$, by Lemma
\ref{canonical-rotation-lemma}, we have $t \sim^c u$, and since $t$ is normal,
we obtain $t \in \lang{A^r_2}$, by Lemma~\ref{canonical-rotation-ta}. Hence
$\lang{A_1} \subseteq \lang{A^r_2}$.
\qed\end{proof}
\fi

\begin{example}[cont. of Ex.~\ref{ex:DLLtoTA}]\label{ex:TAclosure}
The automaton $A^r$ corresponding to the rotation closure of the automaton $A$ is
$A^r=\langle \Sigma, \{q_1,q_2,q_2^{rev},q_{fin}\}, \Delta,\{q_1,q_{fin}\}\rangle$ 
where:{\small
\[\begin{array}{rcl}
\Delta & = & \left\{\begin{array}{lcl}
\overline{\langle \ax \mapsto (\bx,\dx) \wedge \ax=\cx, \emptyset \rangle} ()  \rightarrow  q_1
\ \ \ \ \ 
\overline{\langle \ax \mapsto (x,\bx),\emptyset, (x,\ax) \rangle} (q_2) & \rightarrow & q_1 \\
\overline{\langle \exists hd'. hd' \mapsto (\dx,p) \wedge hd=\cx \wedge hd'=hd, (hd,p) \rangle} () & \rightarrow & q_2 \\
\overline{\langle \exists hd'. hd' \mapsto(x,p) \wedge hd'=hd, (hd,p), (x,hd) \rangle} (q_2) & \rightarrow & q_2 \\
\overline{\langle \exists hd'. hd' \mapsto (d,p) \wedge hd=c \wedge hd'=hd, \emptyset, (p,hd)
\rangle} (q_2^{rev}) & \rightarrow & q_{fin} \\
\overline{\langle \ax \mapsto (x,\bx), (\ax,x) \rangle} () & \rightarrow & q_2^{rev} \\
\overline{\langle \exists hd'. hd' \mapsto(x,p) \wedge hd'=hd, (hd,x), (p,hd) \rangle} (q_2^{rev}) &
\rightarrow & q_2^{rev}\\
\overline{\langle \exists hd'. hd' \mapsto(x,p) \wedge hd'=hd, \emptyset, (x,hd),(p,hd) \rangle} 
(q_2,q_2^{rev}) & \rightarrow & q_{fin}\\
\end{array}\right\}
\end{array}\]}  
\end{example}

\enlargethispage{5mm}

\vspace*{-3mm}\section{Complexity}\vspace*{-2mm}\label{sec:complexity}

In this section, we prove tight complexity bounds for the entailment problem in
the fragment of SL with inductive definitions under consideration (i.e.\ with
the restrictions defined in Sec. \ref{sec:inductive-definitions}). The first
result shows the need for {\em connected} rules in the system as allowing
unconnected rules leads to the undecidability of the entailment problem. As
a~remark, the general undecidability of entailments for SL with inductive
definitions has already been proved in \cite{fossacs14}. In addition to their
result, our proof stresses the fact that undecidability occurs due the lack of
connectivity within some rules.

\begin{theorem}\label{not-connected-undecidable} The entailment problem $P_i
\models_{\mathcal{P}} P_j$, for some $1 \leq i,j \leq n$, is undecidable for
inductive systems $\mathcal{P} = \{P_i \equiv |_{j=1}^{m_i} R_{i,j}\}_{i=1}^n$
that can have unconnected rules.\end{theorem}
\ifLongVersion
\begin{proof} A context-free grammar is a tuple $G=\langle \mathcal{X}, \Sigma,
\delta \rangle$ where $\mathcal{X}$ is a finite nonempty set of
\emph{nonterminals}, $\Sigma$ is a finite nonempty \emph{alphabet} such that
$\mathcal{X} \cap \Sigma = \emptyset$, and $\delta \subseteq \mathcal{X}\times
(\Sigma \cup \mathcal{X})^*$ is a finite set of \emph{productions}. Given two
strings $u,v \in (\Sigma \cup \mathcal{X})^*$, we define a \emph{step} $u
\Longrightarrow_G v$ if there exists a production $(X, w) \in \delta$ and some
words $y,z \in (\Sigma \cup \mathcal{X})^*$ such that $u=yXz$ and $v=ywz$. We
denote by $\Longrightarrow^*_G$ the reflexive and transitive closure of the
$\Longrightarrow_G$ relation. The {\em language} of a grammar $G$ generated by a
nonterminal $X \in \mathcal{X}$ is defined as $L_X(G) = \{w \in \Sigma^* ~|~ X
\Longrightarrow_G^* w\}$. It is known that the inclusion problem $L_X(G)
\subseteq L_Y(G)$, for some nonterminals $X,Y \in \mathcal{X}$, is undecidable
as originally proved in \cite{barhillel-perles-shamir61}. We reduce from this
problem to entailment within unconnected inductive systems. 

Let $\Sigma = \{\sigma_1, \ldots, \sigma_N\}$ be the alphabet of
$G$. We define the set of selectors $$Sel = \{1, \ldots, \lceil \log_2
N \rceil + 1\}$$ and, for each alphabet symbol $\sigma_K$, we define a
basic SL formula $\varphi_K(x,y) \equiv x \mapsto(\vec{\beta},y)$
where, for all $0 \leq i < \lceil \log_2 N \rceil$:
\[\vec{\beta}_i = \left\{\begin{array}{ll}
x & \mbox{if $1$ occurs on position $i$ in the binary encoding of $K$} \\
\nil & \mbox{if $0$ occurs on position $i$ in the binary encoding of $K$}
\end{array}\right.\]
A word $w = \sigma_{i_1} \cdot \ldots \cdot \sigma_{i_n} \in \Sigma^*$
is encoded by the formula $$\varphi_w(x,y) \equiv \exists x_1 \ldots
\exists x_{n-1} ~.~ \varphi_{i_1}(x,x_1) * \ldots
\varphi_{i_n}(x_{n-1},y).$$ We define a predicate name $P_X(x,y)$ for
each nonterminal $X \in \mathcal{X}$, and for each production $\pi
\equiv (X, w_1 \cdot X_{i_1} \cdot \ldots \cdot w_n \cdot X_{i_n}
\cdot w_{n+1}) \in \delta$ where $w_1,\ldots,w_n \in \Sigma^*$ are
words and $X_{i_1}, \ldots, X_{i_n} \in \mathcal{X}$ are
non-terminals, we have the rule: 
\[\begin{array}{rcl}
R_\pi(x,y) \equiv \exists x_1 \ldots \exists x_{2n} & . &
\varphi_{w_1}(x,x_1) * P_{X_{i_1}}(x_1,x_2) * \ldots * \\ &&
\varphi_{w_n}(x_{2n-2},x_{2n-1}) * P_{X_{i_n}}(x_{2n-1}, x_{2n}) *
\varphi_{w_{n+1}}(x_{2n},y).
\end{array}\]
Finally, each predicate is defined as
$P_X \equiv |_{(X,w) \in \delta} ~R_{(X,w)}$, and the inductive system
is $\mathcal{P}_G = \{P_X ~|~ X \in \mathcal{X}\}$. It is immediate to
check that $P_X(x,y) \models_{\mathcal{P}_G} P_Y(x,y)$ if and only if
$L_X(G) \subseteq L_Y(G)$. \qed\end{proof}
\fi

The second result of this section provides tight complexity bounds for the
entailment problem for connected systems with the restrictions defined in Sec.
\ref{sec:inductive-definitions}. We must point out that EXPTIME-hardness of
entailments in the fragment of \cite{Iosif13} was already proved in
\cite{fossacs14}. The result below is stronger as the fragment under
consideration is a restriction of the fragment from \cite{Iosif13}. In
particular, we do not allow branching parameter copying and do not consider
checking entailments between formulae with empty spatial heads like, e.g.,
$\exists z ~.~ \DLL(x,z) * \DLL(z,y) \models \DLL(x,y)$. On the other hand,
these restrictions allow us to set an EXPTIME upper bound on the entailment
problem for connected systems.

\begin{theorem}\label{entailment-exptime-complete} The entailment problem $P_i
\models_{\mathcal{P}} P_j$, for some $1 \leq i,j \leq n$, is EXPTIME-complete
for connected inductive systems $\mathcal{P} = \{P_i \equiv |_{j=1}^{m_i}
R_{i,j}\}_{i=1}^n$.\end{theorem}
\ifLongVersion
\begin{proof} We reduce from and to the inclusion problem for non-deterministic
(bottom-up) tree automata. It is known that the language inclusion problem for
non-deterministic tree automata is EXPTIME-complete (see, e.g.\ Corollary 1.7.9
in \cite{tata05}). 

For the hardness part, let $A_1 = \langle Q_1, \Sigma, \Delta_1, F_1
\rangle$ and $A_2 = \langle Q_2, \Sigma, \Delta_2, F_2 \rangle$ be two
tree automata over the same alphabet $\Sigma = \{\sigma_1, \ldots,
\sigma_N\}$. We assume w.l.o.g. that $Q_1 \cap Q_2 = \emptyset$.
We define the set of selectors $$Sel = \{1, \ldots, \lceil \log_2 N
\rceil + \max\{\#(\sigma) ~|~ \sigma \in \Sigma\}\}.$$ For each
alphabet symbol $\sigma_K$, $1 \leq K \leq N$ and each rule $\rho
\equiv (\sigma_K(q_1,\ldots,q_n) \arrow{}{} q) \in \Delta_1 \cup
\Delta_2$, we define a basic SL formula $\varphi_K^\rho(x,y_1, \ldots,
y_n) \equiv x \mapsto(\vec{\beta},y_1,\ldots,y_n)$ where, for all $0
\leq i < \lceil \log_2 N \rceil$:
\[\vec{\beta}_i = \left\{\begin{array}{ll}
x & \mbox{if $1$ occurs on position $i$ in the binary encoding of $K$} \\ 
\nil & \mbox{if $0$ occurs on position $i$ in the binary encoding of $K$}
\end{array}\right.\]

For each state $q \in Q_i$ of $A_i$, we consider a predicate name
$P_q(x)$ and an extra predicate name $P_i(x)$, for both $i=1,2$. For
each transition rule $\rho \equiv (\sigma_K(q_1,\ldots,q_n) \arrow{}{}
q)$, of either $A_1$ or $A_2$, we define the following connected
inductive rule:
$$R_\rho(x) \equiv \exists y_1 \ldots \exists y_n ~.~
\varphi_K^\rho(x,y_1,\ldots,y_n) * P_{q_1}(y_1) * \ldots *
P_{q_n}(y_n)$$ Then we define the predicates:
\[\begin{array}{rcl}
P_q & = & \{R_q ~|~ (\sigma(q_1,\ldots,q_n)
\arrow{}{} q) \in \Delta_1 \cup \Delta_2\} \\ 
P_i & = & \{R_q ~|~ (\sigma(q_1,\ldots,q_n) \arrow{}{} q) \in \Delta_i
~\mbox{and}~ q \in F_i\} ~\mbox{for both $i=1,2$}
\end{array}\]
Finally, we have $\mathcal{P} = \{P_q ~|~ (\sigma(q_1,\ldots,q_n)
\arrow{}{} q) \in \Delta_1 \cup \Delta_2\} \cup \{P_1,P_2\}$. It is
easy to check that $P_1(x) \models_{\mathcal{P}} P_2(x)$ if and only
if $\mathcal{L}(A_1) \subseteq \mathcal{L}(A_2)$. Hence the entailment
problem is EXPTIME-hard.

For the EXPTIME-completness part, Theorem \ref{sl-ta-entailment} shows
that any entailment problem $P_i \models_{\mathcal{P}} P_j$, for a
connected system $\mathcal{P}$, can be reduced, in polynomial time, to
a language inclusion problem between tree automata
$\mathcal{L}(A_{P_i}) \subseteq \mathcal{L}(A_{P_j}^r)$. Hence the
entailment problem is in EXPTIME. \qed\end{proof}
\fi

\vspace*{-3mm}\section{Experiments}\vspace*{-2.5mm}\label{sec:experiments}

We implemented a prototype tool \cite{slide} that takes as input two rooted
inductive systems $\langle \mathcal{P}_{lhs},P_{lhs} \rangle$ and $\langle
\mathcal{P}_{rhs}, P_{rhs} \rangle$ with eliminated parameters (i.e.\ the result
of Alg.~\ref{alg:eliminate}).
%
%
Table~\ref{TaExpRes} lists the entailment queries on which we tried
out our tool. The upper part of the table contains local inductive
systems whereas the bottom part contains non-local systems. Apart from
the $\DLL$ and $\TLL$ predicates from
Sect.~\ref{sec:inductive-definitions}, the considered entailment
queries contain the following predicates : $\DLL_{rev}$ that encodes a
DLL from the end, $\TREE_{pp}$ encoding trees with parent pointers,
$\TREE^{rev}_{pp}$ that encodes trees with parent pointers starting
from an arbitrary leaf, $\TLL_{pp}$ encoding TLLs with parent
pointers, and $\TLL^{rev}_{pp}$ which encodes TLLs with parent
pointers starting from their left most leaf.
\ifLongVersion 
These predicates are defined in Fig.~\ref{fig:indDefForExp}. 
\fi 
Columns $|A_{lhs}|$, $|A_{rhs}|$, and $|A^r_{rhs}|$ of Table~\ref{TaExpRes}
provide information about the number of states and transitions of the respective
automata. The tool answered all queries correctly (despite the incompleteness
for non-local systems), and the running times were all under 1 sec. on a
standard PC (Intel Core2 CPU, 3GHz, 4GB RAM).

\begin{table}[t]
\begin{center}

  \caption{Experimental results. The upper table contains local
    systems, while the lower table non-local ones. Sizes of initial TA
    (col. 3,4) and rotated TA (col. 5) are in numbers of
    states/transitions }\vspace*{-2mm}

{\fontsize{8}{9}\selectfont
\begin{tabular}{||c||c||c|c|c||}

  \hline

  Entailment $LHS\models RHS$ & Answer &$|A_{lhs}|$ & $|A_{rhs}|$ &
  $|A^r_{rhs}|$ \\

  \hline \hline
  
  $\DLL(a,\nil,c,\nil) \models \DLL_{rev}(a,\nil,c,\nil)$ & True & 2/4 & 2/4 & 5/8
  \\
  
  $\DLL_{rev}(a,\nil,c,\nil) \models \DLL(a,\nil,c,\nil)$ & True & 2/4 & 2/4 &
  5/8\\ 
  
  $\DLL_{mid}(a,\nil,c,\nil) \models \DLL(a,\nil,c,\nil)$ & True & 4/8 & 2/4 &
  5/8\\ 

  $\DLL(a,\nil,c,\nil) \models \DLL_{mid}(a,\nil,c,\nil)$ & True & 2/4 & 4/8 &
  12/18\\ 

  $ \exists x,n,b.~ x\mapsto (n,b) * \DLL_{rev}(a,\nil,b,x) *
  \DLL(n,x,c,\nil)\models \DLL(a,\nil,c,\nil)$ & True & 3/5 & 2/4 & 5/8 \\
  
  $ \DLL(a,\nil,c,\nil) \models \exists x,n,b.~ x\mapsto (n,b) *
  \DLL_{rev}(a,\nil,b,x) * \DLL(n,x,c,\nil)$ & False & 2/4 & 3/5 & 9/13\\
  
  $\exists y,a.~ x\mapsto (y,\nil) * y\mapsto(a,x) * \DLL(a,y,c,\nil) \models
  \DLL(x,\nil,c,\nil)$ & True & 3/4 & 2/4 & 5/8 \\ 
  
  $\DLL(x,\nil,c,\nil)\models \exists y,a.~ x\mapsto (\nil,y) * y\mapsto(a,x)*
  \DLL(a,y,c,\nil)$ & False & 2/4 & 3/4 & 8/10 \\
  
  $\TREE_{pp}(a,\nil) \models \TREE^{rev}_{pp}(a,\nil)$ & True & 2/4 & 3/8 &
  6/11 \\
  
  $\TREE^{rev}_{pp}(a,\nil) \models \TREE_{pp}(a,\nil)$ & True & 3/8 & 2/4 &
  5/10 \\
  
  \hline \hline

  $\TLL_{pp}(a,\nil,c,\nil) \models \TLL^{rev}_{pp}(a,\nil,c,\nil)$ & True & 4/8 &
  5/10 & 16/26 \\

  $\TLL^{rev}_{pp}(a,\nil,c,\nil) \models \TLL_{pp}(a,\nil,c,\nil)$ & True & 5/10 &
  4/8 & 13/22\\
  
  $ \exists l,r,z.~ a\mapsto(l,r,\nil,\nil) * \TLL (l,c,z)* \TLL(r,z,\nil) \models
  \TLL(a,c,\nil)$ & True & 4/7 & 4/8 & 13/22 \\
  
  $ \TLL(a,c,\nil) \models \exists l,r,z.~ a\mapsto(l,r,\nil,\nil) * \TLL (l,c,z)*
  \TLL(r,z,\nil)$ & False & 4/8 & 4/7 & 13/21\\

  \hline

\end{tabular}
}

  \label{TaExpRes}

\vspace*{-6mm}

\end{center}
\end{table}

\ifLongVersion
\begin{figure}[t]
\[\begin{array}{rcl}

  \DLL_{rev}(hd,p,tl,n) & \equiv & hd\mapsto (n,p) ~\wedge~ hd=tl
  ~|~ \exists x.~tl\mapsto (n,x) * \DLL_{rev}(hd,p,x,tl) \\

  \DLL_{mid}(hd,p,tl,n) & \equiv & hd\mapsto (n,p) ~\wedge~ hd=tl 
  ~|~ hd\mapsto(tl,p) * tl\mapsto(n,hd) \\
  &|& \exists x,y,z.~x\mapsto (y,z) * \DLL(y,x,tl,n) * \DLL_{rev}(hd,p,z,x) \\

  \TREE_{pp}(x,b) & \equiv & x \mapsto (nil,nil,b)
  ~|~ \exists l,r.~x\mapsto (l,r,b) * \TREE_{pp}(l,x) * \TREE_{pp}(r,x) \\

  \TREE_{pp}^{rev}(t,b) & \equiv &  t\mapsto (nil,nil,b)
  ~|~ \exists x,up. x \mapsto (nil,nil,up) * \TREE^{aux}_{pp}(t,b,up,x)\\

  \TREE^{aux}_{pp}(t,b,x,d) & \equiv & \exists r. x\mapsto (d,r,b) 
    * \TREE_{pp}(r,x) ~\wedge~ x=t \\
  & | & \exists l. x\mapsto (l,d,b) * \TREE_{pp}(l,x) ~\wedge~ x=t \\
  & | & \exists r,up. x\mapsto (d,r,up) * \TREE^{aux}_{pp}(t,b,up,x) 
    * \TREE_{pp}(r,x) \\
  & | & \exists l,up. x\mapsto (l,d,up) * \TREE^{aux}_{pp}(t,b,up,x) 
    * \TREE_{pp}(l,x) \\

  \TLL_{pp}(r,p,ll,lr) & \equiv & r \mapsto (nil,nil,p,lr) ~\wedge~ r = ll \\ 
  & | & \exists x,y,z.~ r \mapsto (x,y,p,nil) * \TLL_{pp}(x,r,ll,z) *
  \TLL_{pp}(y,r,z,lr)\\

  \TLL^{rev}_{pp}(t,p,ll,lr)& \equiv & ll\mapsto(nil,nil,p,lr) ~\wedge~ ll=t \\
  & | & \exists up,z. ll\mapsto(nil,nil,up,z) * \TLL^{aux}_{pp}(t,p,up,ll,z,lr)
  \\

  \TLL^{aux}_{pp}(t,p,x,d,z,lr)& \equiv & \exists r. x\mapsto (d,r,p,nil) *
  \TLL_{pp}(r,x,z,lr) ~\wedge~ x=t\\
  & | & \exists r,up,q. x\mapsto (d,r,up,nil) * \TLL^{aux}_{pp}(t,p,up,x,q,lr) *
  \TLL_{pp}(r,x,z,q) \\

\end{array}\]

  \vspace*{-4mm}
 
  \caption{Inductive definitions of predicates used in experiments.}

  \vspace*{-2mm}
 
  \label{fig:indDefForExp}

\end{figure}

For each experiment, we first replaced each left/right-hand side of an
entailment query containing a points-to predicate by a new top-level
predicate and then performed the parameter elimination. We have not
yet implemented this step, but we tightly followed
Alg.~\ref{alg:eliminate}. Subsequently, we have applied our
tool. Column \emph{Answer} of Table~\ref{TaExpRes} shows the results
given by the tool. The answers provided by the tool are all correct
despite the incompleteness of the approach for non-local systems. The
experiments show that our approach is indeed capable of handling quite
complex predicates, including the possibility of encoding the same
data structures starting from different entry points. Moreover, the
approach is also rather efficient: columns $|A_{lhs}|$, $|A_{rhs}|$,
and $|A^r_{rhs}|$ provide information about the number of states and
transitions of the respective automata. Note that no reduction
technique was applied to reduce the size of $A^r_{rhs}$. 
\fi

\enlargethispage{4mm}

\vspace*{-3mm}\section{Conclusion}\vspace*{-2.5mm}

We have presented a novel automata-theoretic decision procedure for the
entailment problem of a non-trivial subset of SL with inductive predicates. Our
reduction to TA can deal with the problem that the same recursive structure may
be represented differently when viewed from different entry points. To deal with
this problem, our procedure uses a special closure operation which closes a
given TA representation with respect to all rotations of its spanning trees. Our
procedure is sound and complete for recursive structures in which all edges are
local with respect to a spanning tree. For this case, using our reduction, we
show that the entailment problem is EXPTIME-complete. For structures outside of
this class, our procedures is incomplete but still sound. We have implemented
our approach in a prototype tool which we tested through a number of non-trivial
experiments. In the experiments, the incompleteness of our approach has never
showed up. Moreover, the experiments also showed that our approach is quite
efficient. In the future, we plan to improve the implementation, extend the
experiments, as well as to integrate our decision procedure into some
verification tool dealing with programs with dynamic linked structures.


\bibliographystyle{splncs03}

\begin{thebibliography}{10}
\providecommand{\url}[1]{\texttt{#1}}
\providecommand{\urlprefix}{URL }

\bibitem{fossacs14}
Antonopoulos, T., Gorogiannis, N., Haase, C., Kanovich, M., Ouaknine, J.:
  Foundations for decision problems in separation logic with general inductive
  predicates. In: To appear in Proc. of FOSSACS'14 (2014)

\bibitem{barhillel-perles-shamir61}
Bar-Hillel, Y., Perles, M., Shamir, E.: On formal properties of simple phrase
  structure grammars. Zeitschrift für Phonetik, Sprachwissenschaft und
  Kommunikationsforschung  14(2),  143--172 (1961)

\bibitem{spaceinvader}
Berdine, J., Calcagno, C., Cook, B., Distefano, D., O'Hearn, P., Wies, T.,
  Yang, H.: Shape analysis for composite data structures. In: Proc. CAV'07.
  LNCS, vol. 4590. Springer (2007)

\bibitem{berdine-calcagno-ohearn04}
Berdine, J., Calcagno, C., O'Hearn, P.W.: A decidable fragment of separation
  logic. In: Proc. of FSTTCS'04. LNCS, vol. 3328. Springer (2004)

\bibitem{anti-chains}
Bouajjani, A., Habermehl, P., Holik, L., Touili, T., Vojnar, T.:
  Antichain-based universality and inclusion testing over nondeterministic
  finite tree automata. In: Proc. of CIAA. LNCS, vol. 5148. Springer (2008)

\bibitem{Brotherston:2010}
Brotherston, J., Kanovich, M.: Undecidability of propositional separation logic
  and its neighbours. In: Proceedings of the 2010 25th Annual IEEE Symposium on
  Logic in Computer Science. pp. 130--139. LICS '10 (2010)

\bibitem{infer}
Calcagno, C., Distefano, D.: Infer: An automatic program verifier for memory
  safety of c programs. In: Proc. of NASA Formal Methods'11. LNCS, vol. 6617.
  Springer (2011)

\bibitem{tata05}
Comon, H., Dauchet, M., Gilleron, R., Jacquemard, F., Lugiez, D., Tison, S.,
  Tommasi, M.: {Tree Automata Techniques and Applications} (2005), \\URL:
  {\ttfamily http://www.grappa.univ-lille3.fr/tata}

\bibitem{cook-haase-ouaknine-parkinson-worell11}
Cook, B., Haase, C., Ouaknine, J., Parkinson, M.J., Worrell, J.: Tractable
  reasoning in a fragment of separation logic. In: Proc. of CONCUR'11. LNCS,
  vol. 6901. Springer (2011)

\bibitem{predator}
Dudka, K., Peringer, P., Vojnar, T.: Predator: A practical tool for checking
  manipulation of dynamic data structures using separation logic. In: Proc. of
  CAV'11. LNCS, vol. 6806. Springer (2011)

\bibitem{enea13}
Enea, C., Saveluc, V., Sighireanu, M.: Compositional invariant checking for
  overlaid and nested linked lists. In: Proc. of ESOP'13. pp. 129--148 (2013)

\bibitem{slide}
Iosif, R., Rogalewicz, A., Vojnar, T.: Slide: Separation logic with inductive
  definitions, \\URL: {\ttfamily
  http://www.fit.vutbr.cz/research/groups/verifit/tools/slide/}

\bibitem{Iosif13}
Iosif, R., Rogalewicz, A., Simacek, J.: The tree width of separation logic with
  recursive definitions. In: Proc. of CADE-24. LNCS, vol. 7898. Springer (2013)

\bibitem{libvata}
Lengal, O., Simacek, J., Vojnar, T.: Vata: a tree automata library, \\URL:
  {\ttfamily http://www.fit.vutbr.cz/research/groups/verifit/tools/libvata/}

\bibitem{sleek}
Nguyen, H.H., Chin, W.N.: Enhancing program verification with lemmas. In: Proc
  of CAV'08. LNCS, vol. 5123. Springer (2008)

\bibitem{Wies13}
Piskac, R., Wies, T., Zufferey, D.: Automating separation logic using smt. In:
  Proc. of CAV'13. LNCS, vol. 8044 (2013)

\bibitem{reynolds02}
Reynolds, J.: {Separation Logic: A Logic for Shared Mutable Data Structures}.
  In: Proc. of LICS'02. IEEE CS Press (2002)

\end{thebibliography}

\end{document}